%% file: main.tex
\pdfoutput=1
\documentclass[11pt]{article}
\usepackage[a4paper,margin=1in]{geometry}
\input{PreambleHopExpanders}

\title{Length-Constrained Directed Expander Decomposition\\and Length-Constrained Vertex-Capacitated Flow Shortcuts}

\author{
    Bernhard Haeupler\thanks{Partially funded by the Ministry of Education and Science of Bulgaria's support for INSAIT as part of the Bulgarian National Roadmap for Research Infrastructure and through the European Research Council (ERC) under the European Union's Horizon 2020 research and innovation program (ERC grant agreement 949272).}\\
    \small bernhard.haeupler@inf.ethz.ch\\
    \small INSAIT, Sofia University "St. Kliment Ohridski" \& ETH Zürich
    \and
    Yaowei Long\thanks{Part of this work was done while at INSAIT, Sofia University "St. Kliment Ohridski", Bulgaria. This work was partially funded from the Ministry of Education and Science of Bulgaria (support for INSAIT, part of the Bulgarian National Roadmap for Research Infrastructure).}\\
    \small yaoweil@umich.edu\\
    \small University of Michigan
    \and
    Thatchaphol Saranurak\thanks{Supported by NSF Grant CCF-2238138. Part of this work was done while at INSAIT, Sofia University "St. Kliment Ohridski", Bulgaria. This work was partially funded from the Ministry of Education and Science of Bulgaria (support for INSAIT, part of the Bulgarian National Roadmap for Research Infrastructure).}\\
    \small thsa@umich.edu\\
    \small University of Michigan
    \and
    Shengzhe Wang\\
    \small shengzhe.wang@inf.ethz.ch\\
    \small ETH Zürich
}

\date{}

\begin{document}

\maketitle

\pagenumbering{gobble}

\input{0_abstract}

\clearpage

\tableofcontents

\clearpage

\pagenumbering{arabic}
\input{1_intro}

\input{1.2_overview}

\clearpage
\input{2_prelim}
\clearpage
\input{3_directed}
\clearpage
\input{4_vertex}
\clearpage
\input{5_vertex_shortcut}
\clearpage

\bibliographystyle{alpha}
\bibliography{main}

\clearpage
\appendix
\input{appendix_3_directed}
\end{document}

%% file: PreambleHopExpanders.tex
\usepackage[utf8]{inputenc}
\usepackage{amsmath, amssymb, amsthm, amsfonts}
\usepackage{xspace}
\usepackage[linktocpage]{hyperref}
\usepackage[dvipsnames]{xcolor}

\usepackage[sans]{dsfont}
\usepackage{mathtools}
\usepackage{url}
\usepackage{appendix}
\usepackage{algorithm}
\usepackage{algpseudocode}
\usepackage{bbm}
\usepackage{thm-restate,color,xcolor,xspace}
\definecolor{DarkRed}{rgb}{0.80,0,0}
\hypersetup{
    unicode=false,          % non-Latin characters in Acrobatӳ bookmarks
    colorlinks=true,        % false: boxed links; true: colored links
    linkcolor=DarkRed,          % color of internal links (change box color with linkbordercolor)
    citecolor=OliveGreen,        % color of links to bibliography
    filecolor=magenta,      % color of file links
    urlcolor=cyan           % color of external links
}
\usepackage[capitalize, noabbrev]{cleveref}
\usepackage{calc}
\usepackage[T1]{fontenc}

\usepackage{graphicx} % Required for inserting images

\newtheorem{theorem}{Theorem}[section]

\newtheorem{lemma}[theorem]{Lemma}

\newtheorem{claim}[theorem]{Claim}\crefname{claim}{claim}{claims}
\newtheorem{remark}[theorem]{Remark}
\newtheorem{corollary}[theorem]{Corollary}

\newtheorem{definition}[theorem]{Definition}

\newcounter{note}[section]

\newcommand{\eps}{\varepsilon}
\newcommand{\reals}{\mathbb{R}}

\newcommand{\naturalnumbers}{\mathbb{N}}
\newcommand{\calP}{\mathcal{P}}

\newcommand{\mcS}{\mathcal{S}}

\newcommand{\poly}{\operatorname{\text{{\rm poly}}}}

\newcommand{\N}{\mathbb{N}}

\newcommand{\diam}{\operatorname{diam}}
\newcommand{\spa}{\operatorname{spars}}
\newcommand{\sep}{\operatorname{sep}}

\newcommand{\cond}{\operatorname{cond}}
\newcommand{\dist}{\operatorname{dist}}

\newcommand{\leng}{\operatorname{leng}}
\newcommand{\conge}{\operatorname{cong}}
\newcommand{\step}{\operatorname{step}}

\newcommand{\ball}{\operatorname{ball}}
\newcommand{\vol}{\operatorname{vol}}

\global\long\def\poly{\mathrm{poly}}%
\global\long\def\vol{\mathrm{vol}}%
\global\long\def\eps{\epsilon}%
\global\long\def\D{\mathcal{D}}%
\global\long\def\N{\mathcal{N}}%
\global\long\def\reals{\mathbb{R}}%
\global\long\def\l{\ell}%
\global\long\def\step{\mathrm{step}}%
\global\long\def\cond{\mathrm{cond}}%
\global\long\def\spars{\mathrm{spars}}%
\global\long\def\sep{\mathrm{sep}}%
\global\long\def\dist{\mathrm{dist}}%
\global\long\def\rdist{\overline{\mathrm{dist}}}
\global\long\def\deg{\mathrm{deg}}%
\global\long\def\diam{\mathrm{diam}}%
\global\long\def\cov{\mathrm{cov}}%
\global\long\def\ball{\mathrm{ball}}%
\global\long\def\supp{\mathrm{supp}}%
\global\long\def\IN{\mathrm{in}}%
\global\long\def\OUT{\mathrm{out}}%
\global\long\def\MID{\mathrm{mid}}%

\global\long\def\Dem{\mathrm{Dem}}
\global\long\def\cong{\mathrm{cong}}
\global\long\def\path{\mathrm{path}}
\global\long\def\vvalue{\mathrm{value}}
\global\long\def\ratio{\mathrm{ratio}}

\renewcommand{\l}{\ell}

\newcommand{\wah}{w_h^\alpha}

\newcommand{\wahA}{w_{h,A}^\alpha}

\newcommand{\DahA}{D^\alpha_{h,A}}

\newcommand{\Gvc}{G_{\mathrm{vc}}}
\newcommand{\Vvc}{V_{\mathrm{vc}}}
\newcommand{\Evc}{E_{\mathrm{vc}}}
\newcommand{\Avc}{A_{\mathrm{vc}}}

\newcommand{\Fvc}{F_{\mathrm{vc}}}
\newcommand{\Pvc}{P_{\mathrm{vc}}}
\newcommand{\Cvc}{C_{\mathrm{vc}}}
\newcommand{\Gec}{G_{\mathrm{ec}}}
\newcommand{\Veca}{V_{\mathrm{ec}}}
\newcommand{\Eec}{E_{\mathrm{ec}}}
\newcommand{\Aec}{A_{\mathrm{ec}}}

\newcommand{\Fec}{F_{\mathrm{ec}}}
\newcommand{\Pec}{P_{\mathrm{ec}}}
\newcommand{\Cec}{C_{\mathrm{ec}}}
\newcommand{\vin}{v_{\mathrm{in}}}
\newcommand{\vmid}{v_{\mathrm{mid}}}
\newcommand{\vout}{v_{\mathrm{out}}}
\newcommand{\uin}{u_{\mathrm{in}}}
\newcommand{\umid}{u_{\mathrm{mid}}}
\newcommand{\uout}{u_{\mathrm{out}}}

\renewcommand{\poly}{\operatorname{poly}}

\newcommand{\myparskip}{3pt}
\parskip \myparskip

\DeclareMathOperator*{\argmin}{arg\,min}

%% file: 0_abstract.tex
\begin{abstract}
We show the existence of \emph{length-constrained expander decomposition} in directed graphs and undirected vertex-capacitated graphs. Previously, its existence was shown only in undirected edge-capacitated graphs \cite{haeupler2022hopexpander,haeupler2024new}. 
% The quality of our decompositions matches the best bound for undirected edge-capacitated graphs. 
Along the way, we prove the multi-commodity maxflow-mincut theorems for length-constrained expansion in both directed and undirected vertex-capacitated graphs. 

Based on our decomposition, we build a \emph{length-constrained flow shortcut} for undirected vertex-capacitated graphs, which roughly speaking is a set of edges and vertices added to the graph so that every multi-commodity flow demand can be routed with approximately the same vertex-congestion \emph{and} length, but all flow paths only contain few edges. This generalizes the shortcut for undirected edge-capacitated graphs from \cite{Haeupler2024emulator}. 
% The quality of our shortcut is slightly worse because of the structural barriers unique to the vertex-capacitated setting. 

Length-constrained expander decomposition and flow shortcuts have been crucial in the recent algorithms in undirected edge-capacitated graphs \cite{Haeupler2024emulator,haeupler2024dynamic}. Our work thus serves as a foundation to generalize these concepts to directed and vertex-capacitated graphs. 
% For example, any algorithm for constructing our shortcut in close-to-linear time would immediately imply the first close-to-linear $O(1)$-approximate multi-commodity flow algorithm in vertex-capacitated graphs. 
% using the framework of \cite{Haeupler2024emulator}.

\end{abstract}

%% file: 1_intro.tex
\section{Introduction}

Expander decomposition found its early applications in property testing \cite{goldreich1998sublinear}, clustering \cite{kannan2004clusterings}, and approximation algorithms \cite{chekuri2005multicommodity} and, for the last two decades, has been the crucial ingredient in important developments of fast graph algorithms. This includes the first almost-linear time algorithms for spectral sparsifiers and Laplacian solvers \cite{spielman2004nearly}, approximate max flow \cite{kelner2014almost,sherman2013nearly,RST14}, deterministic global min cut \cite{li2021deterministic}, exact max flow \cite{ChenKLPGS22}, as well as many almost-optimal dynamic algorithms for minimum spanning trees \cite{spanningforest}, shortest paths \cite{chuzhoy2019new,bernstein2022deterministic}, sparsifiers \cite{bernstein2020fully}, $k$-edge-connectivity \cite{jin2022fully}, minimum cuts \cite{jin2024fully,el2025fully}, and more \cite{GRST21}. Significant effort \cite{DBLP:conf/soda/SaranurakW19,chang2019improved,DBLP:conf/focs/ChuzhoyGLNPS20,chang2020deterministic,li2021deterministic,li2023near,hua2023maintaining,agassy2023expander,gottesburen2024practical,chen2025parallel} has then focused on constructing expander decomposition itself.

Below, we discuss two successful orthogonal generalizations of expander decomposition.

\paragraph{Vertex and Directed Expander Decomposition.}

In 2005, Chekuri, Khanna, and Shepherd \cite{chekuri2005multicommodity} showed that the construction of expander decomposition in undirected edge-capacitated graphs naturally extends to work in undirected vertex-capacitated graphs and applies them for approximating all-or-nothing vertex-capacitated flow problems. Later, this was extended to directed graphs, an even more general setting \cite{chekuri2015all}.\footnote{In fact, expander decomposition was only implicit in \cite{chekuri2005multicommodity,chekuri2015all} as their definitions were specific to their applications. The purely graph-theoretic definition was later formalized in \cite{BernsteinGS20}.}

Since 2020, almost-linear time expander decomposition algorithms in these generalized settings have been developed \cite{BernsteinGS20,DBLP:conf/focs/LongS22,hua2023maintaining,sulser2024simplenearoptimalalgorithmdirected} and found impressive applications. For the vertex-capacitated ones, they were crucial for the fastest deterministic vertex connectivity algorithms \cite{saranurak2022deterministic,nalam2023deterministic} and data structures for connectivity queries under vertex failures \cite{DBLP:conf/focs/LongS22,long2025connectivity,jiang2025new}. For the directed ones, they were used for dynamic algorithms in directed graphs \cite{BernsteinGS20} and  the new combinatorial approaches for exact max flow \cite{chuzhoy2024maximum,bernstein2024maximum}.

\paragraph{Length-Constrained Expander Decomposition and Flow Shortcuts.}

More recently, Haeupler, R\"{a}cke, and Ghaffari \cite{haeupler2022hopexpander} introduced \emph{length-constrained expanders} (LC-expanders). At a very high level, these are graphs such that any ``reasonable'' demand can be routed with low congestion \emph{and} length. In contrast, normal expanders only guarantee low congestion. \cite{haeupler2022hopexpander} constructed LC-expander decomposition and applied it to show universally optimal distributed algorithms. In general, LC-expander decomposition is much more effective for problems that simultaneously concern length and congestion.

Based on the new decomposition, \cite{Haeupler2024emulator} introduced the notion of \emph{LC-flow shortcut}\footnote{It was called a \emph{low-step flow emulator} in \cite{Haeupler2024emulator}.}, a new kind of graph augmentation. Roughly speaking, an LC-flow shortcut is a set of edges and vertices added to the graph so that every multi-commodity flow demand can be routed with approximately the same congestion \emph{and} length, but all flow paths only have a few edges. This is formalized as follows (see \Cref{sec:preliminray} for background).

% \newpage

\begin{definition}
[Length-Constrained Flow Shortcut]
\label{def:LCFlowShortcut}
Given a graph $G=(V,E)$, we say an edge set $E'$ (possibly with endpoints outside $V$) is a $t$-step flow shortcut of $G$ with length slack $\lambda$ and congestion slack $\kappa$ if

\begin{itemize}

\item (Forward Mapping) for every demand $D$ routable in $G$ with congestion $1$ and length $h$, $D$ is routable in $G\cup E'$ with congestion 1, length $\lambda h$, and maximum step $t$, and

\item (Backward Mapping) for every demand $D$ on $V(G)$ routable in $G\cup E'$ with congestion 1 and length $h$, $D$ is routable in $G$ with congestion $\kappa$ and length $h$.

\end{itemize}

\end{definition}

In any undirected edge-capacitated graph, \cite{Haeupler2024emulator} showed, for any $\epsilon>0$, the existence of a $O(1/\epsilon)$-step LC-flow shortcut $E'$ of size $|E'|\le O(n^{1+O(\epsilon)} \mathrm{polylog}(n))$ with length slack $O(1/\epsilon^{3})$ and congestion slack $n^{O(\epsilon)}$.\footnote{The shortcut $E'$ in \cite{Haeupler2024emulator} actually has $O(1/\eps^4)$ length slack and $O(1/\eps^2)$ maximum step, but this is only because they tried to ensure that all endpoints of $E'$ are in $V$. Allowing endpoints outside $V$, one can replace their \emph{router} with a star and improve the quality to be as we stated.} Combined with newly developed close-to-linear time LC-expander decomposition \cite{haeupler2024new,haeupler2025cut}, they also obtained a close-to-linear time construction for LC-flow shortcuts albeit with worse quality.

LC-flow shortcuts have led to significant further progress. This includes the first close-to-linear time constant-approximation algorithm for minimum cost multi-commodity flow \cite{Haeupler2024emulator}. The dynamic but weaker version of flow shortcuts was also the key object in the first deterministic dynamic constant-approximate distance oracle with $O(n^{\epsilon})$ update time \cite{haeupler2024dynamic}.

However, all applications of LC-expander decomposition until now are limited to undirected edge-capacitated graphs.

\subsection{Our Results}

To extend the reach of the expander decomposition paradigm further, the history above suggests the following research question:

\begin{center}

\emph{Can we construct length-constrained expander decomposition and flow shortcuts\\ beyond undirected edge-capacitated graphs?}

\par\end{center}

Indeed, we answer this question affirmatively. In this paper, we focus on the existential results, but the arguments naturally give polynomial-time algorithms.
For future work, we are working towards almost-linear-time constructions, which would lead to further applications for minimum cost (multi-commodity) flow in vertex-capacitated and directed graphs.
Below, we discuss our contribution in more detail.

\paragraph{Length-Constrained Directed and Vertex Expander Decompositions.}

We formalize the notions of length-constrained expanders in directed graphs and in undirected vertex-capacitated graphs (\Cref{sec:basic_LCDE,sec:basic_LCVE}). Then, we show the existence of length-constrained expander decomposition in directed graphs (\Cref{thm:directed_edge_ED}) and in undirected vertex-capacitated graphs (\Cref{thm:vertex_ED}). Along the way, we also show that the definition of length-constrained expanders based on cuts is almost equivalent to the characterization based on multi-commodity flow (\Cref{thm:routing_LCDE,thm:routing_LCVE}). This can be viewed as a version of the approximate multicommodity maxflow mincut theorem \cite{leighton1999multicommodity} but for length-constrained expansion in directed and vertex-capacitated graphs.

While this part does not require technical novelty, it is an important foundation for our paper and, we believe, for future work using this concept.

\paragraph{Length-Constrained Vertex-Capacitated Flow Shortcuts.}

Our main technical contribution (\Cref{thm:vertex_emulator}) is to show that, for any undirected \emph{vertex-capacitated} graph and any $\epsilon>0$, there exists a $2^{O(\frac{1}{\epsilon})}$-step flow shortcut $E'$ of size $|E'|=O(n^{1+O(\epsilon)}\mathrm{polylog}(n))$ with length slack $O(1/\epsilon^{3})$ and congestion slack $n^{O(\epsilon)}$. This generalizes the flow shortcut of \cite{Haeupler2024emulator} in undirected edge-capacitated graphs.

Our trade-off between size, length slack, and congestion slack matches the one of \cite{Haeupler2024emulator}. However, our step-bound is $2^{O(\frac{1}{\epsilon})}$ instead of $O(1/\epsilon^{2})$. This is due to technical barriers unique to vertex-capacitated graphs, which also requires us to use very different analysis. We leave as a very interesting open problem if it is possible to obtain $\text{poly}(1/\epsilon)$ steps.

We note that obtaining similar LC-flow shortcuts on directed graphs is currently out of reach because it would give the breakthrough on \emph{reachability shortcuts}. Given a graph $G=(V,E)$, an edge set $E'$ is a $t$-step reachability shortcut of $G$ if, for every pair of vertices $u,v\in V$, $u$ can reach $v$ in $G$ if and only if $u$ can reach $v$ in $G\cup E'$ using at most $t$ steps.
Observe that an LC-flow shortcut in a directed graph is strictly stronger than a reachability shortcut. It is a major open problem whether there exists a $n^{o(1)}$-step reachability shortcut of size $n^{1+o(1)}$.\footnote{When the endpoints of $E'$ must be in $V$, \cite{Hesse03,HuangP21,bodwin2023folklore} already showed that there is no  $\Omega(n^{1/4})$-step reachability shortcut of size $O(n)$. The lower bounds extend to the shortcut of size $n^{1+\epsilon}$ with a worse step bound.}%\thatchaphol{TODO: once the paper with Gary is out, cite the paper that this is actually known to be impossible.}

% We note that obtaining similar LC-flow shortcuts on directed graphs is impossible, because it would contradict a lower bound for \emph{reachability shortcuts}. Given a graph $G=(V,E)$, an edge set $E'$ (possibly with endpoints outside $V$) is a $t$-step reachability shortcut of $G$ if, for every pair of vertices $u,v\in V$, $u$ can reach $v$ in $G$ if and only if $u$ can reach $v$ in $G\cup E'$ using at most $t$ steps.
% Observe that an LC-flow shortcut in a directed graph is strictly stronger than a reachability shortcut, but unfortunately
% \cite{HuangP21} showed that there cannot be a $n^{1/6-o(1)}$-step reachability shortcut $E'$ of size $|E'|\le n^{1+o(1)}$.\footnote{\cite{HuangP21} only explicitly consider the setting when all endpoints of the shortcut $E'$ are in $V$. However, their argument extends to give an information-theoretic lower bound of $n^{1+o(1)}$ (similar to the argument in \cite{AbboudBP18}). This is because their set $\cal{P}$ of paths connecting critical pairs are edge-disjoint (see their Lemma 2.2), so there are $2^|\cal{P}|$ many subgraphs with distinct reachability profile and}

%% file: 1.2_overview.tex
\subsection{Our Techniques}

Next, we give a technical overview of our LC-flow shortcut on vertex-capacitated graphs. We will explain how the strategy used in \cite{Haeupler2024emulator} fails in our setting and how we overcome the obstacle. 
% \paragraph{LC-flow Shortcuts.}
% Our LC-flow shortcut on vertex-capacitated graphs is based on our hierarchy of LC-vertex expander decomposition. This high-level approach was also used by  \cite{Haeupler2024emulator} in edge-capacitated graphs. 
% % The second part of our paper demonstrates low-step flow shortcut graphs for \emph{vertex-capacitated graphs}, and our construction is based on length-constrained (LC) vertex expander hierarchy. We point out that, low-step flow shortcut graphs for edge-capacitated graphs have been shown by \cite{Haeupler2024emulator}, which is also based on LC-edge expander hierarchy. 
% However, their argument fails in the context of vertex-capacitated graphs. We will discuss the obstacle and how we can bypass it. 
For simplicity, here we only consider graphs with \emph{unit capacity}. Also, we only construct a slightly weaker notion of LC-flow shortcut in the sense that, it receives an additional length parameter $h$ and the forward mapping only guarantees that every demand routable in $G$ with length $h'\leq h$ and congestion $1$ is routable in $G\cup E'$ with length $\lambda h$, congestion $1$ and step $t$.

\paragraph{Preliminaries.}
First, we give a brief background on length-constrained expansion. A\emph{ demand $D:V\times V\rightarrow\mathbb{R}_{\ge0}$ }assigns value to pair of vertices $(u,v)$ and $D$ is \emph{$h$-length} is it assigns non-zero values only to vertex pairs of distance $\dist_{G}(u,v)\le h$. A demand $D$ is \emph{routable} with congestion $\kappa$ and length $\lambda$ if there exists a multi-commodity flow routing $D$ with congestion $\kappa$ and length $\lambda$. $D$ \emph{respects} a \emph{node-weighting} $A:V\rightarrow\mathbb{R}_{\ge0}$ if for each vertex $u$, $\sum_{v}D(u,v)\le A(u)$. Let $|A|=\sum_{u\in V}A(u)$. For any $s\ge1$, $A$ is \emph{$(h,s)$-length $\phi$-expanding in $G$} if every $h$-length $A$-respecting demand is routable with length $hs$ and congestion $O(\frac{\log n}{\phi})$.\footnote{Our definition in the paper (\Cref{def:LC directed expansion}) is actually cut-based. This almost-equivalent flow-based definition follows from \Cref{thm:routing_LCVE} and is more convenient in this overview.} 
A length-constrained cut $C$, generally speaking, assigns to each edge an integral length increase, and $G - C$ is the graph $G$ applied with the length increase from cut $C$.
We informally say that $G$ is a \emph{length-constrained expander (LC-expander)} if a node-weighting $A$ whose support is the whole vertex set $V$ is expanding in $G$. An $(h,s)$-length $\phi$-expander decomposition for $A$ is a length-constrained cut $C$ such that $A$ is $(h,s)$-length $\phi$-expanding in $G-C$.

% First, we give some background needed below. A node-weighting $A:V\rightarrow \mathbb{R}_{\ge 0}$
% explain: node weighting, what is $G-C$, LC vertex expansion, expander decomposition, expander hierarchy \thatchaphol{might need to define this word even in Section 5}

\begin{figure}[htbp]
    \centering
    \includegraphics[width = 0.6\textwidth]{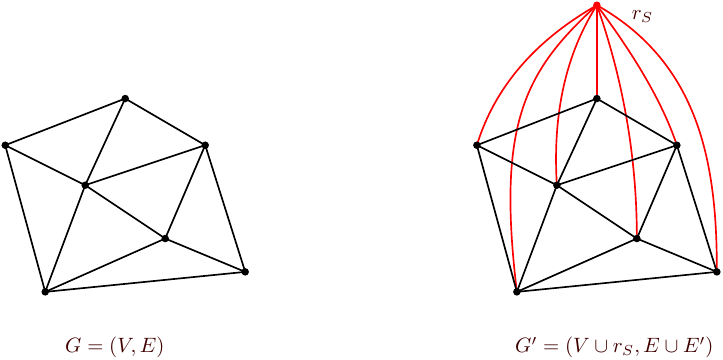}
    \caption{An LC-flow shortcut of a low-diameter LC-expander. }
    \label{fig:shortcut_example}
\end{figure}

\paragraph{Warm-up: Shortcutting LC-expanders.} 
% \paragraph{Flow Shortcut Graphs for Length-Constrained Vertex Expanders.} 
Before explaining the obstacle,
% and our key ideas, for better understanding, 
we first show how to shortcut an LC-vertex expander as a warm-up. 
Suppose that a node-weighting $A$ is $(h,s)$-length $\phi$-vertex expanding in $G$. Say,  $A:=\mathbbm{1}_{V}$ (i.e., $A(v) =1$ for all $v\in V$).
% is $(h,s)$-length $\phi$-vertex expanding in $G$. 

Suppose further that $G$ has diameter at most $h$. In this case, our shortcut is simply a star $S$ connecting each original vertex $v$ to a Steiner vertex $r_S$ with an $(hs)$-length $A(v)$-capacity edge.
We can shortcut any feasible flow in $G$ with $2$ steps. An illustrative example is shown in \cref{fig:shortcut_example}. The length slack is $O(s)$ since an $h$-length original feasible flow is mapped forward to a $(2hs)$-length feasible flow in the star. The congestion slack is $O(\log n/\phi)$ because any feasible flow in the star induces an $h$-length $A$-respecting demand. Since $A$ is $(h,s)$-length $\phi$-expanding in $G$,
% The routing characterization of LC-vertex expanders (\Cref{thm:routing_LCVE}) says that any $h$-length $A$-respecting demand can be routed in $G$ with length $hs$ and (vertex-)congestion $O(\log n/\phi)$. 
% Thus, 
we route such demand in $G$ with congestion $O(\log n/\phi)$ and without length increasing. 

In general, the diameter can be large. Thus, we can construct a \emph{sparse neighborhood cover} to decompose the graph into clusters with diameter $h$, such that (1) for each vertex $v$, there is a cluster containing all vertices within distance $h/s$ from $v$, and (2) each vertex is inside $n^{O(1/s)}$ clusters. Then, we can construct a shortcut by adding an $(hs)$-edge-length star on each cluster. By a similar argument, we obtain a flow shortcut graph for $(h/s)$-length original flows with length slack $O(s^{2})$, congestion slack $O(n^{O(1/s)}\log n/\phi)$ and step $2$.

So far, when we build an LC-flow shortcut for an LC-expander, the vertex-capacitated setting presents no difficulties compared to the edge-capacitated setting, because the above simple approach works in both settings. However, the differences between the two settings arise when generalizing this approach to general graphs via expander hierarchies. 

\paragraph{Previous Approach: Shortcutting General Graphs via Boundary-Linkedness.} 
The key idea of \cite{Haeupler2024emulator} is to exploit a hierarchy of \emph{boundary-linked} LC-expander decomposition, defined as follows. Let $G$ be an edge-unit-capacity graph. Initialize the node-weight $A_0 = \deg_G$. For each level $0\le i \le d$, compute a cut $C_{i+1} \subseteq E$ of size $|C_{i+1}|\approx \phi |A_i|$ such that $A_i + \deg_{C_{i+1}}$ is $(h,s)$-length $\phi$-expanding in $G-C_{i+1}$ where $\deg_{C_{i+1}}(v)$ counts the number of $C_{i+1}$-edges incident to $v$.\footnote{In the actual construction, $C_{i+1}$ assigns fractional values to edges and is called a \emph{moving cut}, defined in \Cref{sec:basic_LCDE}. Here, we assume $C_{i+1}$ is a classic edge cut for simplicity.}
The cut $C_{i+1}$ is called the \emph{boundary-linked} LC-expander decomposition for $A_i$ because it gives a stronger expansion guarantee of $A_i + \deg_{C_{i+1}}$ instead of just $A_i$. Then, we set $A_{i+1} := \deg_{C_{i+1}}$ and continue to the next level $i+1$. By setting $\phi = 1/(n^{O({1/s})}n^\eps)$, we have that $d = O(1/\eps)$ and $A_{d+1} = 0$.

From the above construction, we conclude that,  for each $i$,  $A_i+A_{i+1}$ is $(h,s)$-length $\phi$-expanding in $G-C_{i+1}$. 
Therefore,  as we have seen in the warm-up, we can add stars on the support of $A_i+A_{i+1}$ so that any flows routing $h$-length $(A_{i}+A_{i+1})$-respecting demands in $G-C_{i+1}$ can be shortcut.

% The key idea of \cite{Haeupler2024emulator} is to exploit \emph{boundary-linked} LC-expander hierarchy.
% %\thatchaphol{need to say that below we assume for simplicity that expander decomposition is a pure cut} 
% Consider an edge-unit-capacitated graph $G$. Start with the initial node-weighting $A_{0} = \deg_{G}$. At each level $0\leq i\leq d$, the hierarchy contains a cut $C_{i+1}\subseteq E$\footnote{Precisely, the cut $C_{i+1}$ will be a \emph{moving cut} defined in \Cref{sec:basic_LCDE}. Here for simplicity we assume $C_{i+1}$ is a classic edge cut.}, which decides the next-level node-weighting $A_{i+1}:=\deg_{C_{i}}$ (i.e. $A_{i+1}(v)$ is the number of incident $C_{i+1}$-edges for each $v$), such that $A_{i} + A_{i+1}$ is $(h,s)$-length $\phi$-expanding in $G-C_{i+1}$ (Exceptionally, at the top level $d$, $A_{d}$ is expanding in $G$). We note that \emph{boundary-linked LC-expander decomposition} can compute such $C_{i+1}$ with $|C_{i+1}|\approx \phi |A_{i}|$. Hence the number $d$ of levels can be a constant by setting $\phi = 1/n^{\epsilon}$.

% Intuitively, the node-weighting $A_{i+1}$ is the \emph{boundary node-weighting} at level $i$ since it is incident to the edge cut $C_{i+1}$, and the expander is boundary-linked in the sense that $A_{i}+A_{i+1}$ is expanding instead of just $A_{i}$. Then as we have seen in the warm-up, the expander at each level $i$ enables us to add star graphs so that any flows routing $h$-length $(A_{i}+A_{i+1})$-respecting demands in $G-C_{i+1}$ can be shortcut. 

Now consider a feasible $h$-length flow in $G$. The boundary-linkedness suggests a natural \emph{bottom-up} shortcut scheme. For each flow path $P$, we can think of routing $P$'s \emph{head packet} and \emph{tail packet} (initially at $P$'s left and right endpoints, denoted by $u_{0}$ and $v_{0}$) to the same place via shortcuts. Take the head packet as an example. Start with $u_{0}\in\supp(A_{0}) = V$. At each level $0\leq i\leq d-1$, let $u_{i+1}$ be the left endpoint of the first $P\cap C_{i+1}$-edge behind $u_{i}$. By definition, $u_{i+1}\in \supp(A_{i+1})$ and $P$'s subpath between $u_{i}$ and $u_{i+1}$ is disjoint from $C_{i+1}$, so we can use star graphs at level $i$ to route the head packet from $u_{i}$ to $u_{i+1}$ within $2$ steps. In sum, each of the head and tail packets is routed from the bottom up until they reach $u_{d},v_{d}\in\supp(A_{d})$, and the top-level star graphs can route them together. The total number of steps is $O(d) = O(1/\epsilon)$\footnote{We note that the step bound in \cite{Haeupler2024emulator} is $O(1/\epsilon^{2})$ because they used powers of expander graphs instead of star graphs to avoid creating vertices outside $G$, which brought another $O(1/\epsilon)$ factor.}.

\paragraph{The Obstacle from Vertex Cuts.}
The overall strategy of the above approach is to shortcut flow paths from a vertex of $A_{i}$ to an endpoint of edges in $C_{i+1}$. This was possible since the boundary-linked expander decomposition guarantees that $A_{i}+\deg_{C_{i+1}}$ is expanding. 

In the vertex-capacitated graph, however, the cut $C_{i+1}\subseteq V$ is now a vertex set. To follow the same strategy, we have two natural options. We shortcut flow from a vertex of $A_{i}$ to either (1) a vertex in $C_{i+1}$, or (2) a neighbor of $C_{i+1}$. 

In the first case, the strategy requires that $A_{i}+C_{i+1}$ is expanding in $G-C_{i+1}$. This is trivially impossible because $C_{i+1}$ is not even in the graph $G-C_{i+1}$. In the second case, let $N(C_{i+1})$ denote the neighbors of $C_{i+1}$ that are not in $C_{i+1}$. The strategy requires $A_{i}+N(C_{i+1})$ is expanding in $G-C_{i+1}$. However, possibly $N(C_{i+1})$ is very big and has size $|N(C_{i+1})|=\Omega(n|C_{i+1}|)$. It is unlikely that expander decomposition exists to guarantee the expansion of such a large node-weighting. Even if it exists, we would set $A_{i+1}=N(C_{i+1})$ and, hence, we cannot guarantee $|A_{i+1}|\ll|A_{i}|$. So the number of levels of the hierarchy is unbounded. 

In either option, this overall strategy fails in the vertex-capacitated graphs. At a very high level, this is because edges have two endpoints while vertices may have an unbounded number of neighbors.

% \paragraph{The Obstacle in Vertex-Capacitated Graphs.} The obstacle preventing us to directly apply the \cite{Haeupler2024emulator} approach in vertex-capacitated graphs is that, we still have no definition of boundary-linked LC-vertex expander hierarchy that is reasonable for constructing LC-flow shortcuts. The central reason is that an edge has only two incident vertices while a vertex may have unbounded numbers of adjacent vertices.

% Specifically, in a vertex expander hierarchy, each cut $C_{i+1}\subseteq V$ is a vertex cut. Analogous to the edge expander hierarchy, we may define the boundary node-weighting $B_{i+1}$ at level $i$ by setting $B_{i+1}(v)$ to be the number of adjacent $C_{i+1}$-vertices for each $v\in V - C_{i+1}$. However, in contrast to the edge expander hierarchy where $B_{i+1}=A_{i+1}:=\deg_{C_{i+1}}$ has size $|B_{i+1}| = 2|C_{i+1}|$, now $|B_{i+1}|$ is not linearly upper bounded by $|C_{i+1}|$, so $|C_{i+1}|\approx\phi |A_{i}|$ is no longer guaranteed by the boundary-linked expander decomposition (we will not elaborate the reason behind) and the number of levels may be unbounded.

\paragraph{Our Approach: Top-Down Analysis without Boundary-linkedness.} We construct a similar hierarchy of LC-vertex expander decomposition \emph{without} boundary-linkedness as follows. Let $G=(V,E)$ be a vertex-unit-capacity graph. Initialize node-weighting $A_{0} = \mathbbm{1}_{V}$. At each level $0\leq i\leq d$, computes a cut $C_{i+1}\subseteq V$ such that $A_{i}$ is $(h,s)$-length $\phi$-vertex-expanding in $G-C_{i+1}$, and set $A_{i+1}:=\mathbbm{1}_{C_{i+1}}$. In particular, the top level $d$ has $C_{d+1} = \emptyset$. The LC-vertex-expander decomposition guarantees $|C_{i+1}|\approx \phi |A_{i}|$, so the number $d$ of levels is $O(1/\epsilon)$ by choosing proper $\phi$.

Next, we construct the shortcut as follows. For each $i$, by the expansion of $A_i$, we can add stars on the support of $A_i$ into our shortcut so that any flows routing $h$-length $A_i$-respecting demands in $G-C_{i+1}$ can be shortcut. 
To analyze the shortcut quality, we will no longer try to route from $A_{i}$ to $A_{i+1}$ as in the edge-capacitated setting, because we no longer have boundary-linkedness guarantee.
% , the star graphs we add at each level $0\leq i\leq d$ can only shortcut flows routing $h$-length $A_{i}$-respecting demands in $G-C_{i+1}$, so we can no longer jump from $A_{i}$ to $A_{i+1}$.

Our analysis is instead \emph{top-down}. At each level $i$, we shortcut the current flow path as much as possible, and then the prefix and suffix that have not yet been shortcut will be deferred to lower levels as subproblems. To be more concrete, say our initial goal is to shortcut a flow path $P$ in a feasible $h$-length original flow. At each level $0\leq i\leq d$, assume we will receive a subpath $P'$ of $P$ with length at most $h$ in $G-C_{i+1}$ (note that $P$ is a valid input to the top level $d$ because $C_{d+1}$ is empty). We will shortcut $P'$ using star graphs at levels up to $i$ as follows (see \Cref{fig:shortcut_example_2level} for an illustration when $i=1$).
\begin{figure}[htbp]
    \centering
    \includegraphics[width = 0.7\textwidth]{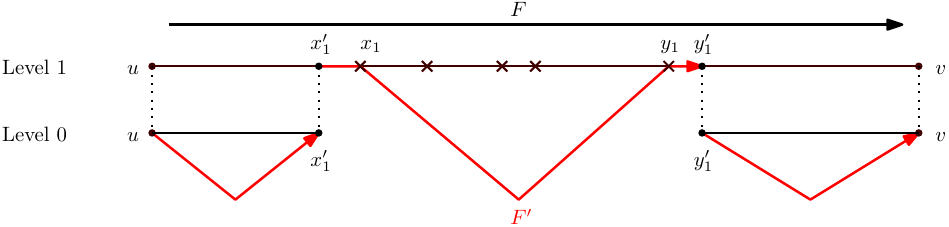}
    \caption{A toy example of forward mapping given we have 2 levels in total. Crossings represent cut vertices in $C_{1}$ along the witness path $P_{u,v}$.}
    \label{fig:shortcut_example_2level}
\end{figure}
%\begin{itemize}
%\item 

\noindent{\underline{Step 1.}} Let $x_{i}$ and $y_{i}$ be the first and last $P'$-vertices in $\supp(A_{i})$ respectively. We can easily shortcut the subpath $P'[x_{i},y_{i}]$ (i.e. the subpath from $x_{i}$ to $y_{i}$) within $2$ steps using the star graphs at level $i$. 
%\item 

\noindent{\underline{Step 2.}} Let $x'_{i}$ be the $P'$-vertex right before $x_{i}$ and let $y'_{i}$ be the $P'$-vertex right after $y_{i}$. We regard shortcutting the prefix $P'[u,x'_{i}]$ and the suffix $P'[y'_{i},v]$ as two subproblems at level $i-1$, where $u,v$ are endpoints of $P'$. Note that both $P'[u,x'_{i}]$ and $P'[y'_{i},v]$ has length at most $h$ in $G-C_{i}$ because they are disjoint from $C_{i}$ by definition.

%\item 
\noindent{\underline{Step 3.}} After the recursion, we obtain shortcuts for both $P'[u,x'_{i}]$ and $P'[y'_{i},v]$. The shortcut for $P'$ is given by concatenating shortcuts for $P'[u,x'_{i}], P'[x_{i},y_{i}]$ and $P'[y'_{i},v]$ using two original edges $(x'_{i},x_{i})$ and $(y_{i},y'_{i})$.
%\end{itemize}

It is not hard to see the final step bound is $2^{O(d)} = 2^{O(1/\epsilon)}$ since the recursion has $d$ levels and each level has two branches. We note that the actual argument is more complicated because the cuts $C_{i}$ are actually moving cuts which have fractional cut values, and there is no clear partition of $P'$ into $3$ parts. 

%\thatchaphol{Remind the reader again that the actual argument is more complicated because $C$ is a moving cut, there is there is clear partition of $P'$ into $3$ parts.}

%% file: 2_prelim.tex
\section{Preliminaries}\label{sec:preliminray}

This section includes preliminaries for directed graphs and vertex-capacitated graphs in \Cref{sect:PreDigraphs} and \Cref{sect:PreVertexGraphs}. We note that in \Cref{sec:LCDE} we consider directed graphs, while in \Cref{sec:LCVE} and \Cref{sec:VertexShortcut} we focus on vertex-capacitated graphs.

We always use $n$ and $m$ to denote the number of vertices and edges of the original graph in the context (for example, the original graph of each main theorem is its input graph). Throughout the paper, all input graphs will have length and capacity functions, and we assume that all lengths and capacities are \emph{positive integers} upper bounded by $N = \poly(n)$. To simplify notation, sometimes we may hide $O(\log N)$ factors in $O(\log n)$ or $n^{\epsilon}$ (where $\epsilon$ is a constant). 

%\noindent {\bfseries Important Remark about Polynomially Bounded Object Sizes:} \\
%All graphs and other objects defined in this paper are assumed to be of polynomial size in $n$, i.e., for any arbitrary but fixed constant $c_{\max}$ we assume that all objects are of size at most $N < n^{c_{\max}}$. 
%In particular we only consider graphs with $m < N$.
%\footnote{While our constructions and algorithms graphs sometimes grow an input graph or other objects a bit such a growth is never by more than a polynomial factor overall.}
%This polynomial bound on object sizes in this paper also allows us to treat logarithmic upper bounds in the sizes of these objects as essentially interchangeable, e.g., for any constant bases $b,b'$ we have that 
%$\log_b n \leq \log_b m_{\max} = \Theta(\log_{b'} N)$. 
%$\log_b n = \Theta(\log_{b'} N)$. 
%Throughout the paper we therefore use $O(\log N)$ without any explicitly chosen basis to denote such quantities.

\subsection{Directed Graphs}
\label{sect:PreDigraphs}

%\paragraph*{Graphs.}
%\noindent{\textbf{Graphs.}}
Let $G=(V,E)$ be a directed graph with $n := |V|$ vertices and $m:= |E|$ edges. 
%By default, $G$ is directed and allowed to have self-loops but not parallel edges.
Let $\l_G : E \rightarrow \naturalnumbers^+$ denote the edge length function of $G$. A path from vertex $v$ to vertex $w$ is called a $(v,w)$-path. For any path $P$, $\l_G(P) = \sum_{e \in P} \l_G(e)$ and let $|P|$ denote the number of edges in $P$ (we also say $P$ has $|P|$ \emph{steps}). The \emph{distance} between vertices $v$ and $w$ is $\dist_G(v, w) = \min_{P: (v,w)\textrm{-path}} \l_G(P)$.
A \emph{ball} of radius $r$ around a vertex $v$ is a $\ball_G(v,r) = \{ w \mid \dist_G(v,w) \le r\}$.
We further define that $\dist_G(v,w) = \infty$ if vertex $w$ cannot be reached from vertex $v$.

Let $u_G : E \rightarrow \reals_{>0}$ denote the edge capacity function of $G$. 
Sometimes we use $u(e)$ to represent the capacity over edge $e$ if $G$ is explicitly mentioned.
We further define the directed degree of a vertex $v$ in $G$, denoted by $\deg^{in}_G(v)$ and $\deg^{out}_G(v)$, as $\sum_{(w,v) \in E}u_G(w,v)$ and $\sum_{(v,w) \in E}u_G(v,w)$ respectively.
Namely, they represent the capacity summation of incoming edges to $v$ and outgoing edges from $v$.
We use $\deg_G^{sum}(v) = \deg^{in}_G(v) + \deg^{out}_G(v)$ to represent the capacity of all edges incident to $v$.
%, and further for $S\subset V$ we let the volume of $S$ be $\vol(S)=\sum_{v\in S} \deg_G^{sum}(v)$. 
We use $\deg_G^{min}(v) = \min\{\deg^{in}_G(v), \deg^{out}_G(v)\}$ to measure the minimum of in-degree and out-degree of a vertex $v$. 
%We let $E(A,B)$ for $A,B \subset V$ denote the edges $(v,w)$ in $E$ with $v \in A$ and $w \in B$. 
%We let $\dist_G(v,w)$ denote the length of the shortest directed path from vertex $v$ to vertex $w$, and define $\dist_G(v,w) = \infty$ if vertex $w$ cannot be reached from vertex $v$. 

%\paragraph*{Multicommodity Flows.} 
\medskip\noindent{\textbf{Multicommodity Flows.}}
A \emph{(multicommodity) flow/routing} in $G$ is a function $F:{\cal P}\to \mathbb{R}_{\geq 0}$ (where ${\cal P}$ denotes the set of simple paths in $G$) that assigns each simple path $P$ in $G$ a flow value $F(P) \geq 0$.  We define $P$ to be a \emph{flow path} of $F$ if $F(P) > 0$. Further $P$ is a \emph{$(v,w)$-flow path} of $F$ if $P$ is both a $(v,w)$-path and a flow path of $F$. Let $\path(F)$ denote the set of all paths $P$ where $F(P) > 0$. The \emph{value} of $F$ is denoted by $\vvalue(F) = \sum_{P\in\path(F)}F(P)$.
%, and we call it the support of $F$. 
%A flow is referred to as a single-commodity flow if all its flow paths share the same endpoints; otherwise, it is a multi-commodity flow.
We point out that, for each flow path $P\in\path(F)$, we also regard $P$ as a flow with only one flow path. The value of this flow $P$ is $\vvalue(P) = F(P)$ unless otherwise stated.

The \emph{congestion of $F$ on an edge $e$} is $\conge_F(e) = \frac{F(e)}{u_G(e)}$ where $F(e) = \sum_{P: P \ni e} F(P)$ denotes the total flow value of all paths passing through $e$. 
The \emph{congestion} of $F$ is $\conge(F) = \max_{e \in E} \conge_F(e)$. 
We define a flow as feasible if its congestion does not exceed one.
The \emph{step} of $F$ is the maximum step count across all flow paths in $F$, given by $\step(F) = \max_{P \in \path(F)} |P|$. Similarly, the \emph{length} of $F$ is defined as the maximum length of all flow paths of $F$, i.e., $\leng(F) = \max_{P \in \path(F)} \l(P)$. Sometimes we will refer to $\leng(F,G')$ as the length of $F$ in $G'$, where $G'$ is the same graph as $G$ except that $G'$ may have different edge lengths. 

%\paragraph*{Node-Weighting.} 
\medskip\noindent{\textbf{Node-Weightings.}}
A \emph{node-weighting} $A:V\rightarrow\mathbb{R}_{\ge0}$ of $G$ assigns a non-negative value $A(v)$ to a vertex $v$. The \emph{size} of $A$ is denoted by $|A|=\sum_{v}A(v)$ and let $\supp(A) := \{v : A(v) > 0\}$. For two node-weightings $A,A'$ we define $\min(A,A')$
%$A-A'$ 
and $A+A'$ as pointwise operations, and we write $A \preceq A'$ if $A$ is pointwise at most $A'$.

We further define the volume of a vertex subset $S \subseteq V$ w.r.t a node-weighting $A$ to be $\vol_A(S) = \sum_{v \in S} A(v)$. Note that the degree function is also a node-weighting, 
and we let $\vol_{G}^{sum}(S) = \sum_{v \in S} \deg_{G}^{sum}(v)$.
%and we let $\vol_G(S) = \sum_{v \in S} \deg_{G}^{min}(v)$ and $\vol_{G}^{sum}(S) = \sum_{v \in S} \deg_{G}^{sum}(v)$.

%\paragraph*{Demands.} 
\medskip\noindent{\textbf{Demands.}}
A \emph{demand} $D:V\times V\rightarrow\mathbb{R}_{\ge0}$ assigns
a non-negative value $D(v,w) \ge 0$ to each ordered pair of vertices in $V$ to specify the units of demand $D(v,w)$ vertex $v$ wants to send to vertex $w$. We note that it is always safe to assume $D(u,u) = 0$.
The size of a demand is written as $|D|$ and is defined as $\sum_{v,w} D(v,w)$. For each vertex $v$, we let $D(v,\cdot) = \sum_{w}D(v,w)$ and $D(\cdot,v) = \sum_{w}D(w,v)$ denote the total demand starting and ending at $v$ respectively. We say a demand $D$ is
\begin{itemize}
    \item \emph{h-length constrained}: if for each $v,w \in V$ with $D(v,w) > 0$, we have $\dist_G(v,w) \le h$;
    \item \emph{symmetric}: if for any pairs of $v,w \in V$, $D(v,w) = D(w,v)$;
    \item \emph{sub-demand for $D'$}: if for any pairs of $v,w \in V$, $D(v,w) \le D'(v,w)$;
    \item \emph{$A$-respecting for some node-weighting $A$}: if for any $v\in V$, $\max\{D(v,\cdot),D(\cdot,v)\}\leq A(v)$.
\end{itemize}
%\begin{definition}[Demand Load and Respect]
%    The load of a demand $D$, denoted with $\load(D)$, is the node-weighting which assigns the node $v$ the weight $\max\{\sum_{a \in V}D(v,a),\sum_{a \in V}D(a,v)\}$.  We write . We say demand $D$ is \emph{$A$-respecting} if $\load(D) \preceq A$. 
    %We say that $D$ is \emph{minimum-degree-respecting} if $D$ is $\deg_{G}^{min}$-respecting. 
    %The directed load of a demand $D$ over vertex $v$, denoted with $\load^{in}_{D}(v)$ and $\load^{out}_{D}(v)$, is defined to be $\sum_{w\in V} D(v,w)$ and $\sum_{w\in V} D(w,v)$ respectively. We say demand $D$ is \emph{$G$-respecting} if for any vertex $v \in V$, we have $\load^{in}_{D}(v) \le \deg^{in}_G(v)$ and $\load^{out}_{D}(v) \le \deg^{out}_G(v)$.
%\end{definition}

The \emph{demand routed by $F$} (or the \emph{corresponding demand of $F$}), denoted by $\Dem(F)$ or $D_{F}$, has $D_{F}(u,v) = \sum_{\text{$(u,v)$-flow path $P\in\path(F)$}} F(P)$. We say a demand $D$ can be routed with length $h$, congestion $\gamma$ and step $t$ if there is a flow $F$ routing $D$ with length $h$, congestion $\gamma$ and step $t$.

\subsection{Vertex-Capacitated Graphs}
\label{sect:PreVertexGraphs}

An undirected vertex-capacitated graph, denoted by $G=(V,E)$ has lengths and capacities on both edges and vertices. Formally, $G$ has length function $\ell_{G}:V\cup E\to \mathbb{N}^{+}$ and capacity function $u_{G}:V\cup E\to \mathbb{N}^{+}$. We will use $u_{V(G)}$ to denote the restriction of $u_{G}$ on $V$, i.e. the \emph{vertex} capacity function of $G$. Naturally, for any simple path $P$ in $G$, its length $\ell_{G}(P) = \sum_{\text{vertices}\ v\in P}\ell_{G}(v) + \sum_{\text{edges}\ e\in P}\ell_{G}(e)$, but its step $|P|$ is still the number of edges in $P$. For each vertex $v\in V$, its (capacitated) degree is $\deg_{G}(v) = u_{G}(v) + \sum_{e\in E\text{ incident to }v}u_{G}(e)$.

% Compared to edge cuts, vertices can also separate the graph into sub-communities that admit favorable properties in expander decompositions.
% Thus it is natural to consider graphs where vertices also admit capacity and length.

%Some previous notations for edge-capacitated graphs naturally generalizes to vertex-capacitated graphs, like $h$-length demand and node-weighting.
%For those where vertex capacity plays a role,
%we rephrase and adapt notations for vertex-capacitated graphs as follows:

%\paragraph{Vertex-Capacitated Flows.} 
\medskip\noindent{\textbf{Vertex-Capacitated Flows.}}
The notation of flow naturally extends to the setting of vertex-capacitated graphs. The only difference is the definition of congestion. 
We note that the congestion of flow $F$ over a vertex $v$ is the $\text{cong}_F(v) = \frac{F(v)}{u_G(v)}$ where $F(v) = \sum_{P: P \ni v} F(P)$ denotes the total flow value of all flow paths going through $v$. 
Thus the general congestion of flow $F$ also considers the congestion over vertices, i.e.,
\[\conge_F = \max\{\max_{v \in V(G)} \conge_F(v), \max_{e \in E(G)} \conge_F(e)\}\]
We further note that the step of flow $F$ in vertex-capacitated graphs still depends on the number of edges in flow paths.
% Similar to edges, the congestion of flow $F$ over a vertex $v$ is the $\text{cong}_F(v) = \frac{F(v)}{u_G(v)}$ where $F(v) = \sum_{P: v \in P} F(P)$ denotes the total flow value of all flow-paths going through $v$. The congestion of $F$ is $\conge_F = \max\{\max_{v \in V(G)} \conge_F(v), \max_{e \in E(G)} \conge_F(e)\}$. The length of $F$ is $\leng_F = \max_{P \in \supp(F)}\l_G(P)$, which measures the maximum length of all flow-paths of $F$. The step of $F$ is $\step_F = \max_{P \in \supp(F)}|P|$ where $|P|$ denotes the total number of edges and vertices on the path $P$. 

%\paragraph*{Node-Weightings and Demands.} Their definitions are exactly the same as above.

%\paragraph*{Neighborhood Covers.} 
\medskip\noindent{\textbf{Neighborhood Covers.}}
We will use neighborhood covers in \Cref{sec:VertexShortcut}.
Given a graph $G$ with lengths, a \emph{clustering} $\mcS$ in $G$ is a collection of pairwise disjoint vertex sets $S_1, \cdots, S_{|\mcS|}$, called \emph{clusters}. A \emph{neighborhood cover} $\N$ with \emph{width} $\omega$ and \emph{covering radius} $h$ is a collection of $\omega$ many clusterings $\mcS_1, \cdots, \mcS_\omega$ such that for every node $v$ there exists a cluster $S \in \mcS_i$ such that $\ball(v,h) \subseteq S$, where $\ball(v,h) = \{u\in V(G)\mid \dist_{G}(u,v)\leq h\}$. We use $S \in \N$ to denote that $S$ is a cluster in some clustering of $\N$. The clustering $\N$ has \emph{diameter} $h_{\diam}$ if every cluster $S \in \N$ has diameter at most $h_{\diam}$ i.e., $\max_{u,v \in S} \dist_G(u,v) \le h_{\diam}$. We note that the shortest path between $u, v \in S$ may use vertices that are not in the cluster.

\begin{theorem}[\cite{peleg2000distributed}]\label{thm:nbhCover}
    Given a vertex-capacitated graph $G$ with a length parameter $h$ and an integer $k\geq 1$, there exists a  neighborhood cover $\N$ with covering radius $h$, diameter $h_{\diam} \le (2k-1)\cdot h$ and width $\omega = n^{O(1/k)}k$. % The algorithm has $O(|E(G)|hk\omega)$ work and $O(hk\omega)$ depth.
\end{theorem}

%% file: 3_directed.tex
\section{Length-Constrained Directed Expansion}\label{sec:LCDE}
In this section, we follow the theory of length-constrained expansion and extend it to the setting of directed graphs. 
We start with the generalization of notations from length-constrained expanders in \cref{sec:basic_LCDE}, which serves as the foundation for subsequent results.
Next, we characterize length-constrained expansion in directed graphs with routing in \cref{sec:routing_LCDE}, and show the existence of length-constrained directed expander decomposition in \cref{sec:LCDED}.

% \subsection{Section Preliminaries}
\subsection{Basic Concepts of Length-Constrained Directed Expansion}\label{sec:basic_LCDE}
% \paragraph{Basic Concepts of Length-Constrained Directed Expansion.}
The following definition of moving cuts and separation was introduced by Haeupler, Wajc and Zuzic in \cite{DBLP:conf/focs/HaeuplerWZ20}.
\begin{definition}[Length-Constrained Cut]\label{def:movingcut}
An $h$-length moving cut $C: E \mapsto \{0,\frac{1}{h},\frac{2}{h},\dots,1\}$ assigns to each edge $e$ a fractional cut value between zero and one which is a multiple of $\frac{1}{h}$. The \emph{size} of $C$ is defined
as $|C|=\sum_{e} u(e) \cdot C(e)$. The length increase associated with the $h$-length moving cut $C$ is denoted with $\l_{C,h}$ and defined as assigning an edge $e$ the length increase $\l_{C,h}(e) = h\cdot C(e)$. Any moving cut which
only assigns cut values equal to either $0$ or $1$ is called a pure moving cut. We define the degree of a moving cut over vertex $v$ to be $\deg_{C}(v) = \sum_{e\ni v}u_G(e)\cdot C(e)$.
\end{definition}

\begin{definition}[$h$-Length Separated Demand]
    For any demand $D$ and any $h$-length moving cut $C$, we define the amount of $h$-length separated demand as the sum of demands between vertices that are $h$-length separated by $C$. We denote this quantity with $\sep_h(C,D)$, i.e.,
    \begin{align*}
        \sep_{h}(C,D) = \sum_{u,v : \dist_{G-C}(u,v)>h} D(u,v).
    \end{align*}
\end{definition}

\begin{definition}[$h$-Length Sparsity of a Cut $C$ for Demand $D$]\label{dfn:CDSparse}
For any demand $D$ and any $h$-length moving cut $C$ with $\sep_{h}(C,D)>0$, the $h$-length sparsity of $C$ with respect to $D$ is the ratio of $C$'s size to how much demand it $h$-length separates i.e.,
\begin{align*}
    \spa_{h}(C,D) = \frac{|C|}{\sep_{h}(C,D)}.
\end{align*}
\end{definition}

Above we generalize the definition of length-constrained moving cut w.r.t arbitrary directed $h$-length demand. However, for the definition of a directed length-constrained expander, we restrict to symmetric $h$-length demands.

\begin{definition}[$(h,s)$-Length Sparsity of a Cut w.r.t.\ a Node-Weighting]\label{def:sparsity}
The $(h,s)$-length sparsity of any $h\cdot s$-length moving cut $C$ with respect to a node-weighting $A$ is defined as:
\begin{align*}
    \spa_{(h,s)}(C,A) = \min_{A\text{-respecting h-length symmetric demand}\ D} \spa_{h \cdot s} (C,D).
\end{align*}
\end{definition}
Intuitively, $(h\cdot s)$-length sparsity of a cut measures how much it $h\cdot s$-length separates $h$-length demand w.r.t its own size.
Furthermore, for a given node-weighting, we associate the sparsest cut w.r.t the node-weighting with its conductance.

\begin{definition}[$(h,s)$-Length Conductance of a Node-Weighting]
    The $(h,s)$-length conductance of a node-weighting $A$ in a graph $G$ is defined as the $(h,s)$-length sparsity of the sparsest $h\cdot s$-length moving cut $C$ with respect to $A$, i.e.,
    \begin{align*}
        \cond_{(h,s)}(A) = \min_{h \cdot s\text{-length moving cut } C} \spa_{(h,s)}(C,A).
    \end{align*}
\end{definition}

\begin{definition}[$(h,s)$-Length $\phi$-Expanding Node-Weightings]\label{def:LC directed expansion}
    We say a node-weighting $A$ is $(h,s)$-length $\phi$-expanding if the $(h,s)$-length conductance of $A$ in $G$ is at least $\phi$.
\end{definition}

To see the connection, in \Cref{sec:appendix_directedED}, we explain how our notion of length-constrained directed expansion generalizes the non-length-constrained version of directed expansion.
Lastly, we give the formal definition of length-constrained directed expander decompositions as follows:

\begin{definition}[Length-Constrained Directed Expander Decomposition]
    Given a graph $G = (V,E)$, a directed $(h,s)$-length $\phi$-expander decomposition for a node-weighting $A$ with length slack $s$ and cut slack $\kappa$ is an $h \cdot s$-length cut $C$ of size at most $\kappa \cdot \phi|A|$ such that $A$ is $(h,s)$-length $\phi$-expanding in $G - C$.
\end{definition}

\subsection{Routing Characterization of Length-Constrained Directed Expansion}\label{sec:routing_LCDE}
% \paragraph{Routing Characterization of Length-Constrained Directed Expansion.} 
The definition of $\phi$-expanding characterizes the sparsity of moving cuts in directed graphs.
%It is also closely related to demand routing, making it an important tool for solving graph flow problems.
With the routing characterization, we certify the notation to be meaningful and show that sparsity is closely related to demand routing.
\begin{theorem}[Routing Characterization of Length-Constrained Directed Expanders]\label{thm:routing_LCDE}
Given a directed graph $G$ and node-weighting $A$, for any $h \ge 1$, $\phi < 1$ and $s \ge 1$ we have:
\begin{itemize}
    \item If $A$ is $(h,s)$-length $\phi$-expanding in $G$, then every $A$-respecting $h$-length symmetric demand can be routed in $G$ wth congestion at most $O(\frac{\log N}{\phi})$ and length at most $h \cdot s$.
    \item If $A$ is not $(h,s)$-length $\phi$-expanding in $G$, then some $A$-respecting $h$-length symmetric demand cannot be routed with congestion at most $\frac{1}{2\phi}$ and length at most $\frac{h\cdot s}{2}$.
\end{itemize}
\end{theorem}
%We remark that the above theorem also applies to the asymmetric demands.

The proof idea of \cref{thm:routing_LCDE} is similar to the undirected case as shown in \cite{haeupler2022hopexpander}, 
and for completeness, we restate and adapt the proof for the directed setting in \cref{sec:appendix_LC_routing}.

\subsection{Length-Constrained Directed Expander Decomposition: Existence}\label{sec:LCDED}

% \paragraph{Existence of Length-Constrained Directed Expander Decompositions.} 
Now, we prove the existence of length-constrained directed expander decompositions.
The following theorem formally states the result:
\begin{restatable}{theorem}{thmDirectedEdgeED}\label{thm:directed_edge_ED}
    For any $G = (V,E)$, a node-weighting $A$, $h > 1$, $\alpha \ge 1$, $\phi < 1$ and a length slack parameter $s = O(\log n)$, there is a directed $(h, s)$-length $\phi$-expander decomposition for $A$ with cut slack $\kappa = O(n^{O(\frac{1}{s})}\log n)$.
\end{restatable}
The proof of \Cref{thm:directed_edge_ED} again follows the undirected-case proof in \cite{haeupler2022hopexpander}, and we append it
%in \Cref{sect:ProofDiED}
for completeness. The basic idea to prove the above main theorem is that we can continuously find a cut with $(h,s)$-sparsity less than $\phi$ in the graph if it is not an $(h,s)$-length $\phi$-expander for node-weighting $A$. We can apply the cut to the graph and repeat the same procedure. Finally the union of those cuts will render the graph as an $(h,s)$-length $\phi$-expander.

In order to argue the upper bound for the total cut size, we need to consider the summation of each individual cut size. However, the size of each cut depends on the sparsity associated with different demands, which adds complexity to the problem.
Thus, we first introduce a special base demand, called exponential demand, in \cref{sec:exp_demand} to relate all other demands in terms of sparsity. 
Using this, we apply a potential argument to prove the above main theorem in \cref{sec:proof_LCDED}.

In \Cref{sec:appendix_linkedness}, we will discuss \emph{boundary-linked} LC-directed expander decomposition (also called \emph{linked} LC-directed expander decomposition). Expander decompositions with boundary-linkedness have been shown to be very useful in the length-constrained undirected setting and the classic (i.e. non-length-constrained) setting. Hence we include this part which may lead to future applications, although it has no application in our work.

\subsubsection{Exponential Demand}\label{sec:exp_demand}
Exponential demand can be viewed as a worst-case demand because, for any sparse cut with a witnessing demand, it admits comparable sparsity w.r.t that same cut. 
Specifically, we would like to show the following lemma:
\begin{restatable}{lemma}{lemSepExpDemand}\label{lem:separationexpdemand}
    Given a directed graph $G = (V,E)$ and a node-weighting $A$, suppose that some $A$-respecting $h$-length symmetric demand $D$ has $h \cdot s$-length sparsity at most $\phi$ w.r.t. some $h\cdot s$-length moving cut $C$. Then the $\alpha$-exponential demand $\DahA$ has $\frac{h\cdot s}{2}$-length sparsity at most $2^{8\alpha+1}\phi$ w.r.t the same cut $C$.
\end{restatable}
We will develop the definition of the exponential demand (i.e. \Cref{def:ExponentialDemand}) in the following part and start with a specific weight function base on distance.
\paragraph{Exponential Distance Weight}
For a directed graph $G = (V,E)$, and a length-bound $h$, we define \emph{$h$-length $\alpha$-exponential distance weight} of a vertex $u$ w.r.t.\ vertex $v$ as 
\begin{equation*}
\wah(u,v):= \left\{\begin{array}{ll}
                        1 & u=v\\
           2^{-\alpha\cdot\rdist(u,v)/h} & \rdist(u,v)\le \frac{2h\log_2 n}{\alpha} \\
                        0 & \text{otw.} \end{array}\right.
\end{equation*}
where $1 \le \alpha \le \log n$ and $\rdist(u,v) = \dist(u,v) + \dist(v,u)$ is the round-trip distance between vertex pairs. It immediately follows that the round-trip distance is symmetric, i.e., $\rdist(v,u) = \rdist(u,v)$.

The following lemma serves as an introduction to the property of exponential distance weight.
\begin{lemma}\label{lem:exponential-basics}
For any graph $G$, length-bound $h$, vertices $u,v,w\in V$ the following hold for the $h$-length $\alpha$-exponential distance weights:
\begin{enumerate}
    \item $\sum_{a \in V} \wah(u,a) \le n$.
    \item $\sum_{a \in V} \wah(u,a) \ge 1$.\label{pro:lower}
    \item $\wah(u,v) \ge 2^{-\alpha\cdot\rdist(u,v)/h} - 1/n^2$.\label{pro:single_lower} % $\max\{2^{-\rdist(u,v)/h} , 0\}$
    \item  $2^{-\alpha\cdot\rdist(w,v)/h} \cdot \wah(u,v) - 1/n^2  \le \wah(u,w) \le 2^{\alpha\cdot\rdist(w,v)/h} \cdot (\wah(u,v) + 1/n^2)$\label{pro:closebound}
\end{enumerate}
\end{lemma}
\begin{proof}
We provide the proofs as an itemized list corresponding to the statements in the lemma.
\begin{enumerate}
    \item Every vertex in $v$ receives a weight $\wah(u,a)$ of at most one and $|V| = n$.
    \item Follows because $\wah(u,u)=1$.
    \item If $\rdist(u,v) < \frac{2h\log_2 n}{\alpha}$, then it is true for $\wah(u,v) = 2^{-\alpha\cdot\rdist(u,v)/h}$ and $\frac{1}{n^2} > 0$; If $\rdist(u,v) \ge \frac{2h\log_2 n}{\alpha}$, then $2^{-\alpha\cdot\rdist(u,v)/h} \le \frac{1}{n^2}$ and $\wah(u,v) = 0$, which concludes the property.
    \item  By triangle inequality we have $ \rdist(u,w) + \rdist(w,v) = \dist(u,w) + \dist(w,u) + \dist(w,v) + \dist(v,w) \ge \dist(u,v) + \dist(v,u) = \rdist(u,v)$. This further gives $\rdist(u,w) \ge \rdist(u,v) - \rdist(w,v)$, and thus
    \begin{align*}
        2^{-\alpha\cdot\rdist(u,w)/h} & \le 2^{\alpha\cdot\rdist(w,v)/h} \cdot 2^{-\alpha\cdot\rdist(u,v)/h}
        %\wah(u,w) & \le 2^{\alpha\cdot\rdist(w,v)/h} \cdot 2^{-\alpha\cdot\rdist(u,v)/h}\\
        %& \le 2^{\alpha\cdot\rdist(w,v)/h} \cdot (\wah(u,v) + 1/n^2)
    \end{align*}
    We note that $\wah(u,w) \le 2^{-\alpha\cdot\rdist(u,w)/h}$, and by Property~\ref{pro:single_lower}, we upper bound $2^{-\alpha\cdot\rdist(u,v)/h}$ by $\wah(u,v) + 1/n^2$ to get
    \begin{align*}
        \wah(u,w) \le 2^{\alpha\cdot\rdist(w,v)/h} \cdot (\wah(u,v) + 1/n^2)
    \end{align*}
    Since the choice of vertices is symmetric, by swapping vertices $w$ and $v$, we can similarly have that
    \begin{align*}
         \wah(u,v) & \le 2^{\alpha\cdot\rdist(v,w)/h} \cdot (\wah(u,w) + 1/n^2)\\
         2^{-\alpha\cdot\rdist(v,w)/h} \cdot \wah(u,v) & \le \wah(u,w) + 1/n^2\\
         2^{-\alpha\cdot\rdist(v,w)/h}  \cdot \wah(u,v) - 1/n^2 & \le \wah(u,w) 
    \end{align*}
    And note that $\rdist(v,w) = \rdist(w,v)$, which concludes the proof.
\end{enumerate}
\end{proof}

The intuition behind the distance weight is that we assign more weight to closer vertices. 
%Similarly, we can also incorporate the node weighting restricted to the graph into the overall exponential weight. 
%Specifically, we further define \emph{$h$-length $\alpha$-exponential weight} w.r.t a node-weighting $A$ as
% \begin{align*}
%     \wahA(u,v) := A(v)\cdot \wah(u,v)
% \end{align*}
We let
% \begin{align*}
%     \wahA(u) := \sum_{b \in V} A(b) \wah(u,b)
% \end{align*}
\begin{align*}
    \wah(u) := \sum_{b \in V} \wah(u,b)
\end{align*}
denote the  \emph{$h$-length $\alpha$-exponential weight} of $u$. 
It can be interpreted as a normalization factor of the exponential weight over vertex $u$, which determines the fraction of demand that $u$ is supposed to send out to any other vertex. 
\paragraph{Mixing Factor}
To explicitly denote this ratio, we define $M^\alpha_{h}(u,v)$ to be the \emph{$h$-length $\alpha$-mixing factor} from vertices $u$ to $v$ as follows:
\begin{align*}
    M^\alpha_{h}(u,v) = \frac{\wah(u,v)}{\wah(u)}
\end{align*}
It immediately follows that $0 \le M^\alpha_{h}(u,v) \le 1$. 

This mixing factor instructs a unit demand between every pair of vertices according to the distance. 
We are interested in the overlap of the mixing factor, in particular, for close enough vertex pairs $u$ and $v$, we count the summation of $\min\{M^\alpha_{h}(u,b), M^\alpha_{h}(v,b)\}$ over every vertex $b \in V$. It turns out that there exists a lower bound for the overlap with the following lemma:
%\thatchaphol{I think we need $\alpha \le \log n$ through out. Here we definitely need it so that $B_u \cap B_v \neq \emptyset$.} \shengzhe{added.}

\begin{lemma}\label{lem:mixing_factor}
    Given a directed graph $G = (V,E)$, $1 \le \alpha \le \log n$, for any pair of vertices $u,v$ where $\rdist(u,v) \le 2h$, we have
    \begin{align*}
        \sum_{b\in V} \min\{M^\alpha_{h}(u,b), M^\alpha_{h}(v,b)\} & \ge 2^{-8\alpha}
    \end{align*}
\end{lemma}

\begin{proof}
    \begin{align*}
        \sum_{b\in V} \min\{M^\alpha_{h}(u,b), M^\alpha_{h}(v,b)\} = \sum_{b\in V} \min\{\frac{\wah(u,b)}{\wah(u)}, \frac{\wah(v,b)}{\wah(v)}\}
    \end{align*}
    % \begin{align*}
    %     \sum_{b\in V} \min\{\frac{\DahA(u,b)}{A(u)}, \frac{\DahA(v,b)}{A(v)}\} \ge \sum_{b\in V} \min\{\frac{A(b)\cdot w_h^d(u,b)}{\sum_{a \in V} A(a) w_h^d(u,a)}, \frac{A(b)\cdot w_h^d(v,b)}{\sum_{a \in V} A(a) w_h^d(v,a)}\}
    % \end{align*}
    %Since $D(u,v) \le A(u)$ and $D(u,v) \le A(v)$. 
    %Then we focus on the left summation. 
    We note that only vertex $b$ where $\wah(u,b) > 0$ and $\wah(v,b) > 0$ will contribute to the above summation otherwise the minimum will take the value 0. Then for simplicity of notation, we let
    \begin{align*}
        B_u = \{w \in V \mid \rdist(u,w) \le \frac{2h\log_2 n}{\alpha}\}
    \end{align*}
    denote the ball centered on vertex $u$ with the radius of $\frac{2h\log_2 n}{\alpha}$ in round-trip distance. For any vertex $v \in V\setminus B_u$, we have $\wah(u,v) = 0$.
    Further we can rewrite the summation as 
    \begin{align}
        \sum_{b\in V} \min\{\frac{\wah(u,b)}{\wah(u)}, \frac{\wah(v,b)}{\wah(v)}\} & = \sum_{b\in B_u \cap B_v} \min\{\frac{\wah(u,b)}{\wah(u)}, \frac{\wah(v,b)}{\wah(v)}\}\\
        & \ge \sum_{b\in B_u \cap B_v} \frac{\min\{\wah(u,b), \wah(v,b)\}}{\max\{\wah(u),\wah(v)\}} \label{ineq: min_minmax}
    \end{align}
    For inequality~(\ref{ineq: min_minmax}), we use the fact that $\min\{\frac{a}{b}, \frac{c}{d}\} \ge \frac{\min\{a, c\}}{\max\{b, d\}}$.
    To further bound the above summation, we first build the relationship between weights over vertices $u$ and $v$.
    For the numerator, for vertices $u,v$ and $b \in B_u \cap B_v$, and by triangle inequality that $\rdist(b,u) \ge \rdist(b,v) - \rdist(u,v)$, we have
    %\thatchaphol{how?}\shengzhe{added with triangle inequality}
    \[\wah(b,u) \le 2^{-\alpha\cdot\rdist(b,u)/h} \le 2^{\alpha\cdot\rdist(u,v)/h} \cdot 2^{-\alpha\cdot\rdist(b,v)/h}\]
    Note that we have $\rdist_G(u,v) \le 2h$ and $2^{-\alpha\cdot\rdist(b,v)/h} = \wah(v,b)$ since $b \in B_v$, we actually tighten the inequality and get
    \begin{align*}
        \wah(u,b) & = \wah(b,u) \le 2^{2\alpha} \cdot \wah(v,b)
    \end{align*}
    We rewrite it as
    \begin{align*}
        \wah(v,b) \ge 2^{-2\alpha} \cdot \wah(u,b)
    \end{align*}
    For the denominator, from the same property and the symmetry between vertices $u$ and $v$, we have
    % \begin{align*}
    %     \wahA(v) & = \sum_{c \in V} A(c)\wah(v,c)\\
    %     & \le \sum_{c \in V} A(c) \cdot 2^{2\alpha} \cdot (\wah(u,c) + \frac{1}{n^2})\\
    %     & = 2^{2\alpha} \sum_{c \in V} A(c)\wah(u,c) + 2^{2\alpha}\cdot \frac{1}{n^2} \sum_{c \in V} A(c)\\
    %     & \le 2^{2\alpha} \wahA(u) + 2^{2\alpha}\cdot \frac{1}{n}\\
    %     & \le 2^{2\alpha+1} \wahA(u)
    % \end{align*}
    \begin{align*}
        \wah(v) & = \sum_{b \in B_u} \wah(v,b) + \sum_{b \in V\setminus B_u} \wah(v,b)\\
        & \le \sum_{b \in B_u}  2^{2\alpha} \cdot \wah(u,b) + \sum_{b \in V \setminus B_u} 2^{2\alpha} \cdot (\wah(u,b) + \frac{1}{n^2})\\
        %& = 2^{2\alpha} \sum_{b \in V} A(b)\wah(u,b)\\
        & = \sum_{b \in V}  2^{2\alpha} \cdot \wah(u,b) + \sum_{b \in V \setminus B_u} 2^{2\alpha} \cdot \frac{1}{n^2}\\
        & \le 2^{2\alpha} (\wah(u) + \frac{1}{n})\\
        & \le 2^{2\alpha+1}\wah(u)
    \end{align*}
    where we use the fact that $\frac{1}{n} < \wah(u)$.
    %where w.l.o.g we can assume $u \in \supp(A)$, and thus $A(u) \ge \frac{1}{n}$ and $\frac{1}{n} \le A(u)\wah(u,u) \le \wahA(u)$.

%\yaowei{The inequality above only holds for $b\in B_{u}$}

    Now combine two inequalities for the numerator and the denominator respectively we have
    \begin{align*}
        \sum_{b\in B_u \cap B_v} \frac{\min\{\wah(u,b), \wah(v,b)\}}{\max\{\wah(u),\wah(v)\}} & \ge \sum_{b\in B_u \cap B_v} \frac{ \min\{\wah(u,b), 2^{-2\alpha}\cdot \wah(u,b)\}}{\max\{\wah(u),2^{2\alpha+1} \wah(u)\}}\\
        & = \frac{1}{2^{4\alpha+1}} \cdot \frac{\sum_{b\in B_u \cap B_v} \wah(u,b)}{\wah(u)}\\
        & = \frac{1}{2^{4\alpha+1}} \cdot \frac{\sum_{b\in B_u \cap B_v} \wah(u,b)}{\sum_{b\in B_u} \wah(u,b)}\\
        & = \frac{1}{2^{4\alpha+1}} \cdot \frac{\sum_{b\in B_u \cap B_v} \wah(u,b)}{\sum_{b\in B_u \cap B_v} \wah(u,b) + \sum_{b\in B_u \setminus B_v} \wah(u,b)} 
    \end{align*}
    Note that in the second equation, we rewrite the denominator since only vertices in the $B_u$ contribute to the $\wahA(u)$. We divide it into two parts in the last equation and provide a bound between them in what follows.
    
    For a vertex $b \in B_u \setminus B_v$, we know that it should be far away from vertex $v$, but close enough to vertex $u$. However, under the restriction that $\rdist(u,v) \le 2h$, $b$ has to lie near the boundary of $B_u$, namely $\rdist(u,b) \ge \frac{2h\log_2 n}{\alpha} - 2h$, and thus $\wah(u,b) \le 2^{2\alpha} \cdot \frac{1}{n^2}$. Consequently,
    \begin{align}
        \sum_{b\in B_u \setminus B_v} \wah(u,b) & \le  \sum_{b\in B_u \setminus B_v} 2^{2\alpha} \cdot \frac{1}{n^2}\\
        & \le 2^{2\alpha} \cdot \frac{1}{n} \label{ineq: upperbound_N}\\
        & \le 2^{2\alpha} \cdot \wah(u,u) \label{ineq: lowerbound_N}\\
        & \le 2^{2\alpha} \cdot \sum_{b\in B_u \cap B_v} \wah(u,b) \label{ineq:BuvIntersect}
    \end{align}
    In the inequality~(\ref{ineq: upperbound_N}), we use the assumption that 
    $|V| \leq n$.
    %$\sum_{b \in V} A(b) \le N$.
    Then in the next inequality~(\ref{ineq: lowerbound_N}), 
    we have $\frac{1}{n} \leq 1 = \wah(u,u)$ The inequality~(\ref{ineq:BuvIntersect}) is because $u \in B_{v}$
    %$B_{u}$ intersects $B_{v}$ 
    (since $\alpha\leq \log_{2}n$).
    %w.l.o.g we can assume $u \in \supp(A)$, and thus $A(u) \ge \frac{1}{N}$ and $\frac{1}{N} \le A(u)\wah(u,u) \le \wahA(u)$.
    This further concludes that
    \begin{align*}
        \frac{\sum_{b\in B_u \cap B_v} \wah(u,b)}{\sum_{b\in B_u \cap B_v} \wah(u,b) + \sum_{b\in B_u \setminus B_v} \wah(u,b)} \ge \frac{1}{2^{2\alpha} + 1}
    \end{align*}

    And finally, we are able to show that
    \begin{align*}
        \sum_{b\in V} \min\{M^\alpha_{h}(u,b), M^\alpha_{h}(v,b)\} \ge \frac{1}{2^{4\alpha+1}} \cdot \frac{1}{2^{2\alpha} + 1} \ge 2^{-8\alpha}
    \end{align*}
\end{proof}

The above lemma shows that we can have a mixing factor with favorable overlap in the graph $G$. The mixing factor can be viewed as an instructor of how we send out the demand in the graph. 
It further helps define the exponential demand more easily as follows:

\begin{definition}[$\alpha$-Exponential Demand w.r.t $h$-length]
\label{def:ExponentialDemand}
   In a directed graph $G = (V, E)$, the $\alpha$-exponential demand $\DahA$ w.r.t $h$-length and a node-weighting $A$ is the demand defined as:
   \begin{align*}
       \DahA(u,v) & = A(u)\cdot M^\alpha_{h}(u,v) + A(v)\cdot M^\alpha_{h}(v,u) 
       %& = A(u)\cdot\frac{A(v)\cdot w_h^d(u,v)}{\sum_{a \in V} A(a) w_h^d(u,a)} + A(v)\cdot\frac{A(u)\cdot w_h^d(v,u)}{\sum_{a \in V} A(a) w_h^d(v,a)}
   \end{align*}
\end{definition}
%\thatchaphol{put this into a remark. Explain this in more detail} 
For simplicity, we sometimes hide the parameter $\alpha$ and refer to the demand as the exponential demand.
\begin{remark}
    We note that for vertex pair $u, v \in V$ that is far away from each, i.e. $\rdist(u,v) > 2h \log n / \alpha$, $\DahA(u,v)$ is zero.
    Since $2h\log n/ \alpha > h$, the $\alpha$-exponential demand is actually not an $h$-length demand.
    Furthermore, suppose there are vertices $u$ and $v$ such that $A(u) \ll A(v)$. From the symmetric construction, exponential demand might send more value from $u$ than $A(u)$, thereby failing to respect node-weighting $A$.
    In conclusion, the $\alpha$-exponential demand is neither guaranteed to be an $h$-length demand nor an $A$-respecting demand, but we will see in the following that it works like a worst-case demand and connects with other $h$-length demands in terms of sparsity.
\end{remark}

\begin{corollary}\label{cor:sym_expD}
    The $\alpha$-exponential demand $\DahA$ is symmetric,
    \begin{align*}
        \DahA(u,v) = \DahA(v,u)
    \end{align*}
\end{corollary}

% \begin{claim}\label{clm:mixing}
%     Given a graph $G = (V,E)$, $\alpha>0$ and a node-weighting $A$, for any pair of vertices $u,v \in \supp(A)$ where $\dist(u,v) \le h$ and $\dist(v,u) \le h$,  %for any $A$-respecting $h$-length symmetric demand $D$ where $D(u,v) = D(v,u) > 0$, %the exponential demand $D_h$ is $h$-hop mixing with $\alpha = \frac{1}{12}$.
%     we have% for some constant $c$,
%     \begin{align*}
%         \sum_{b\in V} \min\{\frac{\DahA(u,b)}{A(u)}, \frac{\DahA(v,b)}{A(v)}\} & \ge 2^{-8\alpha} \\
%         %\sum_{b \in V} \min\{\frac{w_{h,A}(u,b)}{w_{h,A}(u)},\frac{w_{h,A}(v,b)}{w_{h,A}(v)} \} & \ge \frac{1}{100}
%     \end{align*}
% \end{claim}

% In the directed setting, the key component for the existence of a expander decomposition where cutting few edges is the following claim:
As mentioned in the beginning, the reason we are interested in such an exponential demand is that it evenly mixed in the entire graph such that it can be regarded as a worst-case demand.
We are now ready for the proof of the main lemma in this section.

% \begin{lemma}\label{lem:separationexpdemand}
%     Given a directed graph $G = (V,E)$ and a node-weighting $A$, suppose that some $A$-respecting $h$-length symmetric demand $D$ has $h \cdot s$-length sparsity at most $\phi$ w.r.t. some $h\cdot s$-length moving cut $C$. Then the $\alpha$-exponential demand $\DahA$ has $\frac{h\cdot s}{2}$-length sparsity at most $2^{8\alpha+1}\phi$ w.r.t the same cut $C$. 
%     %where $A$ is a different base node-weighting and $\supp(A) \subseteq \supp(A')$.\future{requires that for any pair of vertices $u,v \in \supp(A') \setminus \supp(A)$, $D(u,v) = 0$}
%     %\bernhard{I believe we should be able to add "with respect to cut $C$."}
% \end{lemma}
\lemSepExpDemand*
%\yaowei{In the above lemma, it should be a moving cut has sparsity with respect to some demand.}\shengzhe{added.}

\begin{proof}
    We first note that only separated demand pairs in $D$ are involved in the sparsity w.r.t cut $C$, thus we restrict to a sub-demand $\widehat{D}$ of $D$ where $\widehat{D}(u,v) = D(u,v)$ if $\dist_{G-C}(u,v) > h\cdot s$, otherwise $\widehat{D}(u,v) = 0$. It naturally follows that
    \begin{align*}
        \sep_{h\cdot s}(C,D) = \sep_{h\cdot s}(C,\widehat{D}) = |\widehat{D}| \ge \frac{|C|}{\phi}.
    \end{align*}

    Further, we construct an intermediate demand $D'$ to relate the exponential demand $\DahA$ to the fully separated demand $\widehat{D}$. This helps to build a connection between their sparsity w.r.t $C$.
    Intuitively, we decompose and reroute every demand pair in $\widehat{D}$ according to the configuration and size of exponential demand $\DahA$.
    Specifically, for any demand $\widehat{D}(u,v)$ in $\widehat{D}$, for every vertex $b \in V$, we add $\widehat{D}(u,v) \cdot \min\{ \frac{\DahA(u,b)}{A(u)}, \frac{\DahA(v,b)}{A(v)}\}$ to $D'(u,b)$ and $D'(b,v)$.
    %\thatchaphol{more intuition of this step is helpful.}\shengzhe{added.}
    From such a construction, the value of $D'(u,v)$ comes from non-zero demand $\widehat{D}(u,w)$ and $\widehat{D}(w,v)$, and in total, we have that
    \begin{align*}
        D'(u,v) = & \sum_{w:\widehat{D}(u,w)>0} \widehat{D}(u,w) \cdot \min\{\frac{\DahA(u,v)}{A(u)}, \frac{\DahA(w,v)}{A(w)}\}\\
        + & \sum_{w:\widehat{D}(w,v)>0} \widehat{D}(w,v) \cdot \min\{\frac{\DahA(w,u)}{A(w)}, \frac{\DahA(v,u)}{A(v)}\}
    \end{align*}
    In other words, demand $D'$ depends on the fully separated demand $\widehat{D}$ and the exponential demand.
    By the property that $\widehat{D}$ is an $A$-respecting demand, it turns out that demand $D'$ is a sub-demand for $2\cdot \DahA$.
    To see this, we note that $\min\{a,b\} \le a$ and $\min\{a,b\} \le b$ for any $a,b \in \reals$, then we have,
    \begin{align}
        D'(u,v) & \le \sum_{w:\widehat{D}(u,w)>0} \widehat{D}(u,w) \cdot \frac{\DahA(u,v)}{A(u)} + \sum_{w:\widehat{D}(w,v)>0} \widehat{D}(w,v) \cdot \frac{\DahA(v,u)}{A(v)}\\
        & = \sum_{w:\widehat{D}(u,w)>0} \frac{\widehat{D}(u,w)}{A(u)} \cdot \DahA(u,v) + \sum_{w:\widehat{D}(w,v)>0} \frac{\widehat{D}(w,v)}{A(v)} \cdot \DahA(u,v) \label{eq:sym_D}\\
        & \le 2\cdot \DahA(u,v)\label{ineq:sub2D}
    \end{align}
    We use symmetry in equation~(\ref{eq:sym_D}) where $\DahA(v,u) = \DahA(u,v)$. And for the last inequality~(\ref{ineq:sub2D}), we use the fact that $\sum_{w:\widehat{D}(u,w)>0} \widehat{D}(u,w) \le A(u)$ since $\widehat{D}$ is $A$-respecting. Since the inequality applies to every vertex pair $u, v \in V$, it confirms that $D'$ is a sub-demand for $2\cdot \DahA$. This directly gives that $\sep_{h\cdot s/2}(C, \DahA) \ge \frac{1}{2}\sep_{h\cdot s/2}(C, D')$. 

    Then we show that the $\sep_{h\cdot s/2}(C, D')$ is at least some fraction (dependent on $\alpha$) of the amount of the fully separated demand $\widehat{D}$ w.r.t the same cut $C$. For each demand $\widehat{D}(u,v)$ that contributes to the $\sep_{h\cdot s}(C, \widehat{D})$, we have $\dist_{G-C}(u,v) > h\cdot s$. As a result, for any other vertex $b$, we have either $\dist_{G-C}(u,b) > \frac{h\cdot s}{2}$ or $\dist_{G-C}(b,v) > \frac{h\cdot s}{2}$. Then either from $D'(u,b)$ or from $D'(b,v)$, the amount of $\widehat{D}(u,v) \cdot \min\{\frac{\DahA(u,b)}{A(u)}, \frac{\DahA(v,b)}{A(v)}\}$ is contributed to $\sep_{h\cdot s/2}(C, D')$. If we sum up all vertices $b$ from the vertex set, in total
    \begin{align}
        \sum_{b \in V} \widehat{D}(u,v) \cdot \min\{\frac{\DahA(u,b)}{A(u)}, \frac{\DahA(v,b)}{A(v)}\} & \ge \widehat{D}(u,v) \cdot \sum_{b \in V}\min\{M^\alpha_{h}(u,b), M^\alpha_{h}(v,b)\} \\ 
        & \ge 2^{-8\alpha}\widehat{D}(u,v)\label{eq:mixing_factor}
    \end{align}
    is contributed to the overall separated amount of $D'$ for a single pair $(u,v)$. 
    The reason that we can apply \cref{lem:mixing_factor} to get inequality~(\ref{eq:mixing_factor}) is that $D$ is a symmetric $h$-length demand, and thus $\dist(u,v) \le h$ and $\dist(v,u) \le h$, which means $\rdist(u, v) \le 2h$.
    After summing up all pairs, we have
    \begin{align*}
        \sep_{h\cdot s/2}(C, D') \ge \sum_{u,v \in V} 2^{-8\alpha}\widehat{D}(u,v) = 2^{-8\alpha}|\widehat{D}| = 2^{-8\alpha}\sep_{h\cdot s}(C, \widehat{D})
    \end{align*}
    And finally we can conclude that 
    \begin{align*}
        \sep_{h\cdot s/2}(C, \DahA) \ge \frac{1}{2} \sep_{h\cdot s/2}(C, D') \ge 2^{-8\alpha-1}\sep_{h\cdot s}(C, \widehat{D}) = 2^{-8\alpha-1}\sep_{h\cdot s}(C, D).
    \end{align*}
    This gives the lemma.
\end{proof}
The above lemma helps to relate every demand associated with a sparse cut to the exponential demand. In other words, whenever we are faced with a sparse cut, we can stick with the exponential demand for analysis by only losing a factor of $\exp(\alpha)$.

\subsubsection{Existential Proof of the Decomposition}\label{sec:proof_LCDED}
Finally we will show the existence of length-constrained expander decompositions. 
For the existence of the length-constrained expander decomposition for a graph $G$ w.r.t some node-weighting $A$, we can find sparse cuts iteratively from the graph. Namely, if the graph $G$ is not an expander, it is guaranteed to admit a sparse moving cut $C$. We can apply this cut to the graph and get a new graph $G' = G - C$. 
%For the linked expander decompositions, we also need to update the node-weighting accordingly. 
This can be done iteratively until the updated graph is already expander, or in other words, there does not exist any sparse moving cut.
This gives a sequence of moving cuts, and we can combine them as a single moving cut to show the existence of expander decompositions.

We first formally describe the sequence of moving cuts as follows:

\begin{definition}[Sequence of Moving Cuts]\label{dfn:seq_moving_cuts}
Given a directed graph $G = (V, E)$, and node-weighting $A$, let $(C_1, C_2, \dots,C_n)$ be a sequence of $h \cdot s$ moving cuts, let $G - \sum_{j < i} C_j$ denote the graph that is applied with cuts from $C_1$ to $C_{i-1}$. 
%and let $A + \sum_{j < i}L_{C_j}^{\l}$ denote the node-weighting that is added with linked node-weighting from $L_{C_1}^\l$ to $L_{C_{i-1}}^\l$. 
We define $(C_1, C_2, \dots,C_n)$ as a sequence of $\phi$-sparse moving cuts
if and only if the $(h,s)$-length sparsity of $C_i$ w.r.t $A$ in $G - \sum_{j < i} C_j$ is at most $\phi$.
\end{definition}

It would be less interesting if those moving cuts have very large size, thus it is important to bound the overall size of those moving cuts.
\begin{lemma}\label{lem:existential_cut_sequence}
Let $C_1, \ldots, C_k$ be an sequence of $\phi$-sparse $h\cdot s$-length
cuts for some node-weighting $A$ in the graph $G$ where $h > 1$, $\phi < 1$, $1 \le \alpha \le \log n$ and
$s > \frac{4\log_2 n}{\alpha}$, then $\sum_{i} |C_i| \leq (2^{8\alpha+2}\phi\ln n)\cdot |A|$.
\end{lemma}
\begin{proof}
    Let $G_1$ denote the initial graph $G$, $G_i = G - \sum_{j < i} C_j$.
    %From our assumption of the sequence of moving cuts, for every $i$ there exists a symmetric $h$-length $A$-respecting demand $D^{*}_i$ in the graph $G_i$ such that the $h\cdot s$ sparsity of $C_i$ w.r.t $D^{*}_i$ is at most $\phi$. W.l.o.g, we can assume $D^{*}_i$ is fully separated, and let $\sep^i_{h \cdot s}(C_i, D^{*}_i)$ denote the amount of demand separated by $C_i$ in $G_i$, we have that $\sep^i_{h\cdot s}(C_i, D^{*}_i) \ge \frac{|C_i|}{\phi}$.
    Inspired by the \cref{lem:separationexpdemand}, we will introduce the exponential demand for each graph $G_i$ w.r.t the same node-weighting $A$. To avoid clutter, let $w_{i}$ denote the exponential distance weight $w^{\alpha}_{h}$ with respect to the graph $G_i$. 
    %We note that throughout the sequence of moving cuts, we use the same base node-weighting $\deg_G^{sum}$ in the exponential demand for different $G_i$ because they only differ by the length of edges while each vertex still has the same degree. 
    %Thus we use $w^i_{h,G}$ to denote the exponential weight w.r.t node-weighting $\deg_G^{sum}$ in graph $G_i$ and 
    We further use $D_{i}$ to denote the corresponding exponential demand $D^i_{h, A}$ w.r.t graph $G_i$.
    %It is worth noting that though exponential demands share the same base node-weighting $\deg_G^{sum}$, they are still different from each other because the node-weighting $A_i$ is changed for each $i$.
    
    In the graph $G_{i}$, by \cref{lem:separationexpdemand} and the fact that $\spars_{(h,s)}(C,A)\leq \phi$ (meaning there exists a symmetric $h$-length $A$-respecting demand $D^{*}_{i}$ such that $\spars_{hs}(C,D^{*}_{i})\leq \phi$), we have
    \begin{align*}
        \sep^i_{h\cdot s/2}(C_i, D_{i}) \ge 2^{-8\alpha-1} \cdot \frac{|C_i|}{\phi}.
    \end{align*}

    Further, we define a potential function $P_i: V \to \reals$ w.r.t the graph $G_i$. It assigns a value to each vertex $u$ with the amount of $P_i(u) =  A(u)\ln(w_{i}(u))$. We note the fact that $w_{i}(u) \ge w_{i}(u,u) \ge 1$, which guarantees that $P_i(u) \ge 0$ for all $i$ and vertices $u$. 

    % We divide the potential change into two phases to simplify analysis. In the first phase, we apply cut $C_i$ to the graph $G_i$ and then the resulting graph is $G   _{i+1}$. The distance of some edges increases due to $C_i$, but the node-weighting still remains as the $A_{i}$. Then in the second phase, we added the corresponding linked node-weighting $L_{C_i}^\l$ to $A_i$ and get $A_{i+1}$. The distance between vertices in the graph $G_{i+1}$ remains the same in this phase.
    
    Start with graph $G_i$, each vertex $u$ will have potential $P_i(u)$. After applying cut $C_i$ to the graph $G_i$, we first get the resulting graph $G_{i+1}$ with same node-weighting $A$. %and let $P'_i(u)$ denote the potential of vertex $u$ at this intermediate phase, 
    we have $P_{i+1}(u) = A(u)\ln(w_{i+1}(u))$. Since we only increase the length of some edges in $G_i$, the exponential weight can only decrease between any vertex pairs. Consequently, there is always a decrease from $P_i(u)$ to $P_{i+1}(u)$, and we have
    \begin{align}
        P_i(u) - P_{i+1}(u) & = A(u)\cdot (\ln(w_{i}(u)) - \ln(w_{i+1}(u)))\\
         & = A(u) \cdot (- \ln(1 - (1 - \frac{w_{i+1}(u)}{w_{i}(u)})) )\\
        & \ge A(u) \cdot (1 - \frac{w_{i+1}(u)}{w_{i}(u)}) \label{ineq: log_ineq}\\
        & \ge A(u) \cdot ( \frac{w_{i}(u) - w_{i+1}(u)}{w_{i}(u)})
    \end{align}
    We use the fact that $-\ln (1 - x) \ge x$ when $0 \le x < 1$ for inequality~(\ref{ineq: log_ineq}).
    
    In the graph $G_{i}$, the distance of each demand pair that contributes to $\sep^i_{h\cdot s/2}(C_i,D_{i})$ will be at least $\frac{h\cdot s}{2} > \frac{2h\log_2 n}{\alpha}$. In other words, let $(u,v)$ be a demand pair that contributes $D_{i}(u,v)$ to the separation, we have that $\dist_{G_{i+1}}(u,v) > \frac{2h\log_2 n}{\alpha}$ and $\rdist_{G_{i+1}}(u,v) = \rdist_{G_{i+1}}(v,u) > \frac{2h\log_2 n}{\alpha}$, so $w^{i+1}_{h}(u,v) = w^{i+1}_{h}(v,u) = 0$. This allows us to further lower bound the potential reduction as follows.
    
    \begin{align*}
    \sum_{u\in V}(P_{i}(u) - P_{i+1}(v)) &\geq \sum_{u\in V}A(u) \cdot  \frac{w_{i}(u) - w_{i+1}(u)}{w_{i}(u)}\\
    &= \sum_{\text{ordered } (u,v)\in V\times V}A(u)\cdot\frac{w_{i}(u,v) - w_{i+1}(u,v)}{w_{i}(u)}\\
    &\geq \sum_{\substack{\text{ordered }(u,v)\in V\times V\text{ s.t. }\\\rdist_{G_{i+1}}(u,v)>hs/2}} A(u)\cdot\frac{w_{i}(u,v)}{w_{i}(u)}\\
    &= \frac{1}{2}\sep^{i}_{h\cdot s/2}(C_{i},D_{i}),
    \end{align*}
    where the last equality is by $D_{i}(u,v) = A(u)\cdot \frac{w^{i}_{h}(u,v)}{w^{i}_{h}(u)} + A(v)\cdot \frac{w^{i}_{h}(v,u)}{w^{i}_{h}(v)}$ for each $u,v\in V$.
    %This further means that 
    %\begin{align*}
    %    A(u)\cdot \frac{w^{i}_{h}(u,v)}{w^{i}_{h}(u)} + A(v)\cdot \frac{w^{i}_{h}(v,u)}{w^{i}_{h}(v)}
    %\end{align*}
    %is contributed to the total potential reduction $\sum_{u \in V} (P_i(u) - P_{i+1}(u))$. And we note that the above amount is exactly $D^i_{h,A}(u,v)$. However, if $D^i_{h,A}(u,v)$ and $D^i_{h,A}(v,u)$ both contribute to the separation, then the amount that contributes to the potential reduction still remains as the above. 
    As a result, the overall potential reduction is at least 
    \begin{align*}
        \sum_{u \in V} P_i(u) - P_{i+1}(u) \ge \frac{1}{2} \cdot \sep^i_{h\cdot s/2}(C_i, D_i) \ge 2^{-8\alpha-2} \cdot \frac{|C_i|}{\phi}
    \end{align*}
    %where $\sum_{u\in V}P_1(u) \ge \sum_{u\in V}P_2(u) \ge \dots \ge \sum_{u\in V}P_{k+1}(u)$. 
    Finally, we can come to a conclusion over the summation of size of all cuts.
    \begin{align*}
        \sum_{i}|C_i| & \le 2^{8\alpha+2}\phi \cdot \sum_{i}\sum_{u\in V} (P_{i}(u) - P_{i+1}(u)) \\
        &\le 2^{8\alpha+2}\phi \cdot \sum_{u\in V} P_{1}(u)\\
        & \le 2^{8\alpha+2}\phi \cdot \sum_{u\in V} A(u)\ln(w_{1}(u))\\
        &\le 2^{8\alpha+2}\phi \cdot |A| \cdot \ln n.
    \end{align*}
    For the last inequality, we use that $\ln(w_{1}(u)) \le \ln n$. This concludes the proof.
\end{proof}

The upper bound over the size of the sequence of moving cuts directly implies the existence of length-constrained directed expander decompositions.
\thmDirectedEdgeED*
%\future{$\kappa = n^{O(\frac{1}{s})}\log n$}
\begin{proof}
    From graph $G$, if node-weighting $A$ is already $(h,s)$-length $\phi$-expanding in $G$, then we are done, because the empty cut is a valid expander decomposition. Otherwise, there exists an $h\cdot s$-length cut $C$ with $h\cdot s$ sparsity strictly smaller than $\phi$. We take an arbitrary cut satisfying the above condition and denote it as $C_1$. It is further applied to graph $G$ to get $G_2 = G - C_1$. 
    %and we update the node-weighting to get $A_2 = A + L_{C_1}^\l$. 
    With one further step, if we assume that $A$ is still not $(h,s)$-length $\phi$-expanding in $G_2$, we can find an $h\cdot s$-length cut $C_2$ similarly as above. We update $G_3 = G_2 - C_2 = G - (C_1 + C_2)$ where $C_1 + C_2$ represents that we union two cuts together by summing the cut value on each edge. W.l.o.g we can assume that $C_1 + C_2$ is still an $h\cdot s$-length cut since it is meaningless to make the length increase of an edge larger than $h\cdot s$ when we are considering the $h\cdot s$-length sparsity.

    We repeat the above procedure until we reach some integer $k$ where 
    %$A_{k+1} = A + \sum_{j\le k} L_{C_j}^\l$ 
    $A$ is $(h,s)$-length $\phi$-expanding in $G_{k+1} = G - \sum_{j \le k} C_j$.
    Let $C_{\le k}$ denote the union of all such cuts, we can assume it is an $h\cdot s$-length cut as discussed above. 
    Then by definition $C_{\le k}$ is a valid expander decomposition for $G$ and $A$, and \cref{lem:existential_cut_sequence} guarantees that $|C_{\le k}| = \sum_{j\le k}|C_j| \le 2^{8\alpha+2}\phi \cdot |A| \cdot \ln n$ as long as $s > \frac{4\log_2 n}{\alpha}$.
    This gives that $\kappa \le 2^{8\alpha+2}\cdot\ln n = O(n^{O(\frac{1}{s})}\log n)$.
\end{proof}

%% file: 4_vertex.tex
\section{Length-Constrained Vertex Expansion}
% \section{Vertex-Capacitated Graphs}
\label{sec:LCVE}
In this section we extend the theory of length-constrained expander decomposition to vertex-capacitated graphs. See \Cref{sect:PreVertexGraphs} for preliminaries of vertex-capacitated graphs. 

The basic concepts of length-constrained vertex expansion are analogous to those of length-constrained directed expansion in \Cref{sec:basic_LCDE}. The major difference is that now a moving cut $C$ can assign cut values to both vertices and edges. See \cref{sec:basic_LCVE} for formal description of the basic concepts.

%\begin{definition}[Vertex-capacitated Graphs]
%    Let $G = (V, E)$ be a vertex-capacitated graph. Every edge $e \in E$ is undirected, with a positive capacity $u_G(e)$ and a positive length $l_G(e)$. 
%    Likewise, each vertex $v \in V$ is assigned a positive capacity $u_G(v)$ and a positive length $\l_G(v)$. 
%    A path from vertex $v$ to vertex $w$ is called a $(v, w)$-path. 
%    For any path $P$, the length $\l_G(P)$ is given by $\l_G(P) = \sum_{e \in P} \l_G(e) + \sum_{v \in P} \l_G(v)$. 
%    Furthermore, the distance between vertices $v$ and $w$ is $\dist_G(v, w) = \min_{P: (v, w)\text{-path}}\l_G(P)$.
%\end{definition}

%Specifically, we show the existence of length-constrained expander decomposition for vertex-capacitated graphs.
%This can be done by reducing the problem to the context of directed edge-capacitated graphs.
%We develop the argument formally in the remainder of this section.

The main results of this section is the existence of length-constrained expander decomposition for vertex-capacitated graphs (\Cref{thm:vertex_ED}) and the routing characterization of length-constrained vertex expanders (\Cref{thm:routing_LCVE}).

\begin{restatable}[Existential $(h,s)$-length Expander Decomposition for Vetex-Capacitated Graphs]{theorem}{thmvertexED}\label{thm:vertex_ED}
    For any vertex-capacitated graph $\Gvc = (\Vvc, \Evc)$, node-weighting $\Avc$, $h > 1$, $\phi < 1$ and a length slack parameter $s = O(\log n)$, there is an $(h, s)$-length $\phi$-expander decomposition for $A$ with cut slack $\kappa = O(n^{O(\frac{1}{s})}\log n)$.
\end{restatable}

\begin{restatable}[Routing Characterization of Length-Constrained Vertex Expanders]{theorem}{thmroutingLCVE}\label{thm:routing_LCVE}
    Given a vertex-capacitated graph $\Gvc$ and node-weighting $\Avc$, for any $h \ge 1$, $\phi < 1$ and $s \ge 1$ we have:
\begin{itemize}
    \item If $\Avc$ is $(h,s)$-length $\phi$-expanding in $\Gvc$, then every $h$-length $\Avc$-respecting demand can be routed in $\Gvc$ wth congestion at most $O(\frac{\log N}{\phi})$ and dilation at most $h \cdot s$.
    \item If $\Avc$ is not $(h,s)$-length $\phi$-expanding in $\Gvc$, then some $h$-length $\Avc$-respecting demand cannot be routed with congestion at most $\frac{1}{6\phi}$ and dilation at most $\frac{h\cdot s}{2}$.
\end{itemize}
\end{restatable}

In \cref{sec:reduction}, we introduce a key reduction that transforms vertex-capacitated graphs into directed edge-capacitated graphs, demonstrating their equivalence. 
This equivalence is crucial for the proofs of \Cref{thm:vertex_ED} and \Cref{thm:routing_LCVE} in \cref{sec:proof_LCVED} and \cref{sec:routing_LCVE}.

\subsection{Basic Concepts of Length-Constrained Vertex Expansion}\label{sec:basic_LCVE}
We start with defining concepts related to length-constrained vertex expanders.

\paragraph{Moving Cuts for Vertex-Capacitated Graphs.} 
A notable distinction for vertex-capacitated graphs is that moving cuts can be applied to vertices by exerting a similar length increase on them. 

An $h$-length moving cut $C: V \cup E \rightarrow \{0, \frac{1}{h}, \ldots, 1\}$ on a vertex-capacitated graph $G$ assigns to each edge $e$ and each vertex $v$ a fractional cut value between zero and one which is a multiple of $\frac{1}{h}$. The size of $C$ is defined as $|C| = \sum_{e}u_G(e)\cdot C(e) + \sum_{v}u_G(v)\cdot C(v)$. 
The length increase associated with the $h$-length moving cut $C$ is denoted with $\l_{C,h}$. 
Generalizing from the length increase over edges, 
the moving cut $C$ assigns a vertex $v$ length increase $\l_{C,h}(v) = h \cdot C(v)$. 
We similarly define the degree of the vertex moving cut over a vertex $v$ to be $\deg_{C}(v) = u_G(v)\cdot C(v) + \sum_{e: e \ni v} u_G(e) \cdot C(e)$. %We use $\deg_C$ to denote the node-weighting that assigns vertex $v$ the value of $\deg_C(v)$.

By applying the cut $C$ to a vertex-capacitated graph $G$, the resulting graph is $G- C$ where the length of each vertex and edge increases accordingly. 
To distinguish between two types of moving cuts, we will clarify the type of graph where the moving cut is applied.
%to the cut value and the length parameter $h$.
\begin{remark}
   From the setting of vertex capacity, we note that the undirected edge-capacitated graph is a special case for the vertex-capacitated graph.
   We can reduce an arbitrary undirected edge-capacitated graph to a vertex-capacitated graph by setting the length of vertices to some small constant and allowing arbitrarily large capacity over vertices. 
   Then it is too expensive to have a fractional cut over any vertices. 
   Thus if we generalize previous results to vertex-capacitated graphs, we actually build up a more general framework for length-constrained expanders and expander decompositions.
\end{remark}

\paragraph{$(h,s)$-Length Sparsity for Vertex-Capacitated Graphs.} The definition of sparsity of moving cuts, and the conductance of the node-weighting similarly generalize from the edge cut cases. It may be useful for the reader to recall the definition in \cref{sec:basic_LCDE}.

%\thatchaphol{Why can't we say that the notations below are defined in the same way?}\shengzhe{modified}

% The $h$-length separation of an $h$-length moving cut $C$ w.r.t to any demand $D$ measures the amount of demand that are $h$-length separated by $C$. Then the sparsity of $C$ w.r.t $D$ is the ratio of the cut size to the separation amount i.e.,
% \[\spa_h(C, D) = \frac{|C|}{\sep_h(C, D)}.\]
% Further the $(h,s)$-length sparsity of any $h\cdot s$-length moving cut $C$ w.r.t a node-weighting $A$ takes the minimum over the $h\cdot s$-length sparsity of $C$ w.r.t all $A$-respecting $h$-length demand $D$. 
We remark that for undirected vertex-capacitate graphs, the $(h,s)$-length sparsity is no longer restricted to symmetric demands.
\[\spa_{(h, s)}(C, A) = \min_{A\text{-respecting $h$-length demand}\ D} \spa_{h \cdot s} (C,D).\]
% We note that the vertex-capacitated graph is not directed, and we do not restrict it to symmetric demands.
% Finally the $(h,s)$-length conductance of a node-weighting $A$ is defined as the minimum of the $(h,s)$-length sparsity w.r.t every $h\cdot s$-length moving cut $C$.
% \[\cond_{(h,s)}(A) = \min_{h \cdot s\text{-length moving cut } C} \spa_{(h,s)}(C,A).\]
As mentioned in \cref{sec:basic_LCDE}, we say that $A$ is $(h,s)$-length $\phi$-expanding in $G$ if $\cond_{(h,s)}(A) \ge \phi$. The definition of length-constrained expander decompositions for vertex-capacitated graphs follows.%\thatchaphol{make the notation consistent with the directed case.}\shengzhe{modified}
\begin{definition}[Length-Constrained Vertex Expander Decompostion]
    Given a vertex-capacitated graph $G$,
    an $(h,s)$-length $\phi$-expander decomposition for a node-weighting $A$ with length slack $s$ and cut slack $\kappa$ is an $h\cdot s$-length cut $C$ of size at most $\kappa\cdot\phi |A|$ such that $A$ is $(h,s)$-length $\phi$-expanding in $G-C$.
\end{definition}

\subsection{Reduction to Directed Edge-Capacitated Graphs}\label{sec:reduction}
\begin{figure}[htbp]
    \centering
    \includegraphics[width=0.9\linewidth]{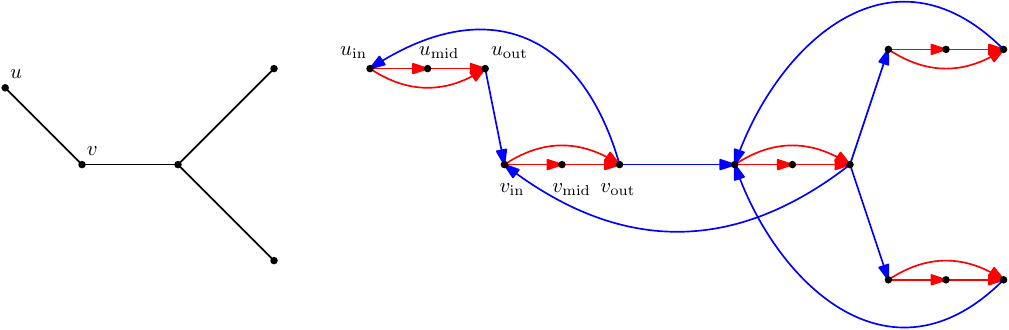}
    \caption{(Left) A vertex-capacitated graph. (Right) The corresponding directed edge-capacitated graph.}
    \label{fig:v_cap_e_cap}
\end{figure}
In \cref{sec:LCDE}, we presented several results concerning directed edge-capacitated graphs, including the existence of expander decompositions.
It further turns out that we can reduce an undirected vertex-capacitated graph to a directed edge-capacitated graph to apply those results.
This is an essential step to prove previous expansion results for vertex-capacitated graphs, and we start with a detailed introduction to the reduction.
%This helps to build up all those results in vertex-capacitated graphs.

Namely, for any vertex-capacitated graph $\Gvc = (\Vvc, \Evc)$ with a node-weighting $\Avc$, one can construct a directed edge-capacitated graph $\Gec = (\Veca, \Eec)$ with a corresponding node-weighting $\Aec$ through the following reduction:
\begin{enumerate}
    \item Let $\Veca = \varnothing$, $\Eec = \varnothing$ and $\Aec(v) = 0$ for all vertices at the beginning.
    \item For each vertex $v \in \Vvc$, add three vertices $\vin$, $\vmid$ and $\vout$ to $\Veca$; Add $(\vin, \vmid)$, $(\vmid, \vout)$ and $(\vin, \vout)$ to $\Eec$. The length of each edge is set as $\l_{\Gvc}(v)$, and the capacity of each edge is set as $u_{\Gvc}(v)$. We call such edges inner edges.
    \item For each edge $e = (u,v) \in \Evc$, add $(\uout, \vin)$ and $(\vout, \uin)$ to $\Eec$. The length of them is set as $l_{\Gvc}(e)$ and the capacity is set as $u_{\Gvc}(e)$. We call such edges outer edges.
    \item For each vertex $v \in \supp(\Avc)$, set $\Aec(\vmid) = \Avc(v)$.
\end{enumerate}

An example is shown in \cref{fig:v_cap_e_cap}. In the right part, inner edges and outer edges are colored red and blue respectively. 
In what follows, we use $\Gvc$ to denote an arbitrary vertex-capacitated graph, and $\Gec$ always refers to the edge-capacitated graph constructed from $\Gvc$ as described above. 
For simplicity, we may occasionally omit to mention $\Gvc$, but whenever $\Gec$ is referenced, it is understood to be based on a specific vertex-capacitated graph.

Also in the undirected setting, the direction of demands does not make a difference. 
Thus for simplicity, when we talk about the routing, we can stick to a symmetric demand $D$ in both types of graphs.
Namely, for any demand $D$ in $\Gvc$, we can have an equivalent symmetric demand $D'$ by balancing the demand in both directions.
And we further use the same notation $D'$ in $\Gec$ where $D'(\umid, \vmid) = D'(\vmid, \umid) = D'(u, v)$ whenever $D'(u, v) > 0$ for any pair of vertices $u, v \in \Vvc$.

\paragraph{Reduction from Vertex-Capacitated to Directed Edge-Capacitated Graphs.}
As mentioned before, the motivation to construct a directed edge-capacitated graph $\Gec$ that shares a similar structure as the vertex-capacitated graph $\Gvc$ is that we can transform problems into the directed and edge-capacitated setting which we have explored in \cref{sec:LCDE}.
To apply our previous results to the vertex-capacitated case, it is essential to elucidate the connection between the two types of graphs. Specifically, we need to show the equivalence of expansion between two types of graphs in our construction as follows:
\begin{restatable}{theorem}{thmMutualExpanding}\label{thm:mutual_expanding}
    Let $\Gvc = (\Vvc, \Eec)$ be a vertex-capacitated graph with node-weighting $\Avc$, and $\Gec = (\Veca, \Eec)$ be the corresponding edge-capacitated graph with node-weigthing $\Aec$.
    
    \begin{enumerate}
        \item If $\Avc$ is $(h,s)$-length $\phi$-expanding in $\Gvc$, then $\Aec$ is $(h,s)$-length $\phi$-expanding in $\Gec$.
        \item If $\Aec$ is $(h,s)$-length $\phi$-expanding in $\Gec$, then $\Avc$ is $(h,s)$-length $\frac{\phi}{3}$-expanding in $\Gvc$.
    \end{enumerate}
\end{restatable}
We develop the theorem gradually with some necessary lemmas concerning the equivalence of sparse cut in both types of graphs.

\begin{figure}[htbp]
    \centering
    \includegraphics[]{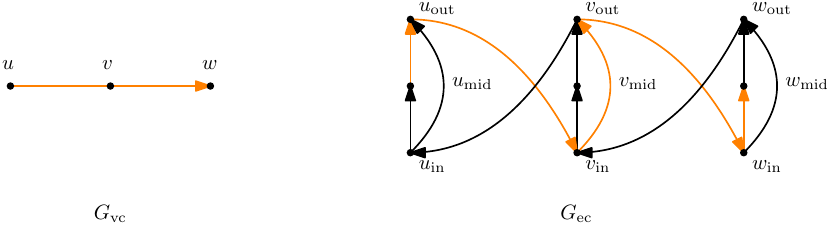}
    \caption{Suppose there is a path $\{u, v, w\}$ in the vertex-capacitated graph $\Gvc$(Left), then we can follow it and create a path $\{\umid, \uout, \vin, \vout, w_{\IN}, w_{\MID}\}$ in $\Gec$(Right). We note that both paths have the same length.}
    \label{fig:gvc_gec_path}
\end{figure}

We first note an important property of our construction is that the equivalence of pairwise distance is kept.
\begin{lemma}\label{lem:eq-distance}
    Given a vertex-capacitated graph $\Gvc = (\Vvc, \Evc)$ and the corresponding edge-capacitated graph $\Gec = (\Veca, \Eec)$, for any $u, v \in \Vvc$, $\dist_{\Gvc}(u, v) = \dist_{\Gec}(\umid, \vmid) = \dist_{\Gec}(\vmid, \umid)$.
\end{lemma}
\begin{proof}
    For any shortest path between $u$ and $v$ in $\Gvc$, it is possible to give a path of the same length in $\Gec$ between $\umid$ and $\vmid$ by following the trajectory.
    For the other direction, a similar argument also applies for the shortest path between $\umid$ and $\vmid$ because it always goes directly from some vertex $a_{\mathrm{in}}$ to $a_{\mathrm{out}}$ to achieve shorter length.
    An example is shown in \cref{fig:gvc_gec_path}.
    We further note that the directed distance between $\umid$ and $\vmid$ is symmetric due to the construction, and this gives the lemma.
\end{proof}

% The above lemma partially demonstrates the equivalence of routing between the vertex-capacitated graph $\Gvc$ and the correspondingly constructed graph $\Gec$.
% To exploit previous results on directed edge-capacitated graphs, it is natural to extend the discussion to the equivalence of cuts between these two types of graphs.
% From \cref{thm:linked_edge_ED}, we establish that for any directed edge-capacitated graph, there exists an expander decomposition, which is a cut.
As mentioned above, we then discuss the equivalence of moving cuts between vertex-capacitated graphs and directed edge-capacitated graphs.
% The next step is to transform the cut over the corresponding vertex-capacitated graph.
We note that it would be convenient to normalize the cut in graph $\Gec$ for simplicity.
Actually, in $\Gec$, for any $h$-length moving cut $C$ with a fully separated demand $D$ such that $\spa_h(C, D) = \frac{|C|}{|D|} \le \phi$, 
we can construct a normalized cut $C'$ as follows:
\begin{enumerate}
    \item For any group of vertices $\vin$, $\vmid$ and $\vout$, $C'$ cut the same length over inner edges between them by taking the maximum cut value over those edges from $C$, i.e.,
    %$C'((\vin, \vmid)) = C'((\vmid, \vout)) = C'((\vin, \vout))$.
    \begin{align*}
        C'((\vin, \vmid)) = C'((\vmid, \vout)) = C'((\vin, \vout)) \\
        = \max\{C((\vin, \vmid)), C((\vmid, \vout)), C((\vin, \vout))\}
    \end{align*}
    \item For any pair of outer edges $(u_{\OUT}, v_{\IN})$ and $(v_{\OUT}, u_{\IN})$, $C'$ also cuts the same length over them by taking the larger cut value over two edges from $C$, i.e.,
    %$C'((u_{\OUT}, v_{\IN})) = C'((v_{\OUT}, u_{\IN}))$.
    \[ C'((u_{\OUT}, v_{\IN})) = C'((v_{\OUT}, u_{\IN})) = \max\{ C((u_{\OUT}, v_{\IN})), C((v_{\OUT}, u_{\IN}))\}\]
\end{enumerate}
Cut $C'$ is normalized in the sense that the cut value over the same group edges is the same, and thus it easily induces a cut in the base graph $\Gvc$. What is more, we do not lose more than a constant factor in terms of sparsity.

\begin{lemma}\label{lem:uniform_cut}
    Given the constructed graph $\Gec$ and the $h$-length moving cut $C$ with a fully separated demand $D$ such that $\spa_h(C, D) = \frac{|C|}{|D|} \le \phi$, the normalized cut $C'$ has sparsity $\spa_h(C', D) \le 3\spa_h(C, D) \le 3\phi$.
    % In $\Gec$, given any $h$-length moving cut $C$ with a fully separated demand $D$ such that $\spa_h(C, D) = \frac{|C|}{|D|} \le \phi$, then there exists a cut $C'$ in $\Gec$ where:
    % \begin{enumerate}
    %     \item For any group of vertices $\vin$, $\vmid$ and $\vout$, $C'$ cut the same length over inner edges between them, i.e., $C'((\vin, \vmid)) = C'((\vmid, \vout)) = C'((\vin, \vout))$.
    %     \item For any pair of outer edges $(u_{\OUT}, v_{\IN})$ and $(v_{\OUT}, u_{\IN})$, $C'$ also cuts same length over them, i.e., $C'((u_{\OUT}, v_{\IN})) = C'((v_{\OUT}, u_{\IN}))$.
    %     \item $\spa_h(C', D) \le 3\spa_h(C, D) \le 3\phi$.
    % \end{enumerate}
\end{lemma}
\begin{proof}
    % We can construct $C'$ from $C$ as the following: for any group of vertices $\vin$, $\vmid$ and $\vout$ we set
    % \begin{align*}
    %     C'((\vin, \vmid)) = C'((\vmid, \vout)) = C'((\vin, \vout)) = \max\{C((\vin, \vmid)), C((\vmid, \vout)), C((\vin, \vout))\}
    % \end{align*}
    % For any pair of outer edge $(u_{\OUT}, v_{\IN})$ and $(v_{\OUT}, u_{\IN})$, we set 
    % \[ C'((u_{\OUT}, v_{\IN})) = C'((v_{\OUT}, u_{\IN})) = \max\{ C((u_{\OUT}, v_{\IN})), C((v_{\OUT}, u_{\IN}))\}\]
    
    %To argue about the sparsity of $C'$ w.r.t $D$, it suffices to show that $C$ also never cuts any outer edges. Suppose $C$ assigns length cut to some outer edge $e$, then $C(e) \ge \frac{1}{h}$. We can assume the $u_{G_2}(e) > h\phi|D|$ since it is large enough. We have $|C| > \frac{1}{h} \cdot u_{G_2}(e) > \phi |D|$, and further  $\spa_h(C, D) > \phi$ which gives contradiction. \textcolor{red}{Not True}
    %Then $C$ can be considered as a sub moving cut of $C'$ in the sense of $C(e) \le C'(e)$ for any edge $e$.
    We first note that from our construction of $C'$, the size of $C'$ only triples at most, which means $|C'| \le 3|C|$. 
    Further, since we only increase the cut length, the separation value w.r.t $D$ cannot be less and $\sep_h(C', D) \ge \sep_h(C, D)$.
    This over all shows that
    \[\spa_h(C', D) = \frac{|C'|}{\sep_h(C', D)} \le \frac{3|C|}{\sep_h(C, D)} = 3\spa_h(C, D)  \le 3\phi\]
\end{proof}

With the help of cut normalization, we further establish the equivalence of $(h,s)$-length sparsity between two types of graphs. 
In general, if a sparse cut exists in the edge-capacitated graph, we can construct a sparse normalized cut and transfer it to the vertex-capacitated graph while preserving the cut's sparsity.
A similar argument also applies to the other direction.

\begin{lemma}\label{lem:mutual_sparse_cut}
    Let $\Gvc = (\Vvc, \Eec)$ be a vertex-capacitated graph with node-weighting $\Avc$, and $\Gec = (\Veca, \Eec)$ be the corresponding edge-capacitated graph with node-weigthing $\Aec$.
    \begin{enumerate}
        \item\label{item:mutual_sparse_cut1} If there is an $h\cdot s$-length moving cut $C$ in vertex-capacitated $\Gvc$ such that $\spa_{(h,s)}(C, \Avc) < \phi$, then there exists an $h\cdot s$-length moving cut $C'$ in $\Gec$ such that $\spa_{(h,s)}(C', \Aec) < 3\phi$.
        \item\label{item:mutual_sparse_cut2} If there is an $h\cdot s$-length moving cut $C'$ in edge-capacitated $\Gec$ such that $\spa_{(h,s)}(C', \Aec) < \phi$, then there exists an $h\cdot s$-length moving cut $C$ in $\Gvc$ such that $\spa_{(h,s)}(C, \Avc) < \phi$.
    \end{enumerate}
\end{lemma}
\begin{proof}
We will prove \Cref{item:mutual_sparse_cut1} and \Cref{item:mutual_sparse_cut2} separately.

\noindent{\textbf{Proof of \Cref{item:mutual_sparse_cut1}.}}    
    Since $\spa_{(h,s)}(c, \Avc) < \phi$, let 
        \[D^* = \argmin_{\Avc\text{-respecting } h\text{-length demand } D} \spa_{h\cdot s} (C, D)\]
        then we have $\spa_{h\cdot s}(C, D^*) < \phi$. W.l.o.g, we can assume $D^*$ is symmetric and fully separated. 
        As mentioned in the reduction, this immediately gives an $\Aec$-respecting symmetric $h$-length demand $D^*$ in $\Gec$ as well.
        To see that $D^*$ is also $h$-length in $\Gec$, we note that from \cref{lem:eq-distance}, we have
        \[\dist_{\Gec}(u_{\MID}, v_{\MID}) = \dist_{\Gec}(v_{\MID}, u_{\MID}) =\dist_{\Gvc}(u , v)\]
        % Since $\Dvc$ is $h$-length, we can have that $\Dec$ is also $h$-length in the directed setting.

        Then we construct a sparse cut $C'$ based on $C$, and show that the sparsity of $C'$ w.r.t $D^*$ is also dependent on $\phi$. For any vertex $u \in \Vvc$ where $C(u) > 0$, we set
        \[C'((u_{\IN}, u_{\MID})) = C'((u_{\MID}, u_{\OUT})) = C'((u_{\IN}, u_{\OUT})) = C(u)\]
        For any $e = (u,v) \in \Evc$ where $C(e) > 0$, we set
        \[C'((u_{\OUT}, v_{\IN})) = C'((v_{\OUT}, u_{\IN})) = C(e)\]
        From the construction, we have that $|C'| \le 3|C|$. Then we show that $D^*$ is also fully separated by $C'$. We know that for any vertex pairs $u, v \in \Vvc$ where $\D^*(u, v) > 0$, the distance between them is bounded by $h$ in $\Gvc$ but at least $h\cdot s$ after applying cut in $\Gvc - C$, i.e.,
        $\dist_{\Gvc}(u ,v) \le h$ and $\dist_{\Gvc - C}(u, v) > h\cdot s$. 
        In the edge-capacitated graph $\Gec$, the key observation is that the equivalence of distance is still kept.
        We still have $\dist_{\Gvc - C}(u, v) = \dist_{\Gec - C'}(u_{\MID}, v_{\MID})$, because we set the length increase over edges in $\Gec$ same as the corresponding vertices or edges in the original graph $\Gvc$.
        Namely, if we construct an edge-capacitated graph based on $\Gvc-C$, we get exactly $\Gec - C'$.
        %Then for any shortest path between $u$ and $v$ in $\Gvc$, it is possible to give a path of the same length in $\Gec$ between $\umid$ and $\vmid$ by following the trajectory, and vice versa.
        Again from \cref{lem:eq-distance}, this implies that $\dist_{\Gec}(u_{\MID}, v_{\MID}) \le h$ and $\dist_{\Gec - C'}(u_{\MID}, v_{\MID}) > h\cdot s$, which means that $D^*$ is also fully separated by $C'$ in $\Gec$. Thus we have $\sep_{h\cdot s}(C, D^*) = \sep_{h\cdot s}(C', D^*)$. 
        We can conclude that
        \begin{align*}
            \spa_{(h, s)}(C', \Aec) & = \min_{\Aec\text{-respecting } h\text{-length symmetric demand } D} \spa_{h\cdot s}(C', D) \\
            & \le \spa_{h\cdot s}(C', D^*)\\
            & = \frac{|C'|}{\sep_{h\cdot s}(C', D^*)}\\
            & \le \frac{3|C|}{\sep_{h\cdot s}(C, D^*)}\\
            & \le 3\phi
        \end{align*}

        \noindent{\textbf{Proof of \Cref{item:mutual_sparse_cut2}.}}
        Since $\spa_{(h,s)}(C', \Aec) < \phi$, let 
        \[D^* = \argmin_{\Aec\text{-respecting } h\text{-length symmetric demand } D} \spa_{h\cdot s} (C', D)\]
        
        From \cref{lem:uniform_cut}, we know there exists a normalized $h\cdot s$-length cut $C_{\mathrm{n}}$ from $C'$. 
        We note that the separation from $C_{\mathrm{n}}$ is at least as large as that from cut $C'$, i.e. \[\sep_{h\cdot s}(C_{\mathrm{n}}, D^*) \ge \sep_{h\cdot s}(C', D^*)\]  
        
        Similarly, we construct a $h\cdot s$-length moving cut $C$ in the vertex-capacitated graph $\Gvc$. Namely, for any vertex $u \in \Vvc$, we set $C(u) = C_{\mathrm{n}}((u_{\IN}, u_{\MID}))$. 
        For any edge $e = (u, v) \in \Evc$, we set $C(e) = C_{\mathrm{n}}((u_{\OUT}, v_{\IN}))$.
        From the property of $C_{\mathrm{n}}$ that the length increase is the same among the same group of inner edges and outer edges, we immediately have $|C| \le \frac{1}{2} |C_{\mathrm{n}}|$. However, the construction is equivalent to that we assign the maximum value of the cut length among the same group of inner edges or the pair of outer edges from $C'$ to the corresponding vertex or edge in $\Gvc$ respectively. This actually shows that $|C| \le |C'|$, which is a tighter bound.
        %Then we construct a demand $\Dvc$ by adding $\Dec'(\umid, \vmid)$ to $\Dvc(u, v)$ for any pair of vertices $u,v \in \Vvc$ and corresponding vertices $\umid, \vmid \in \Veca$.
        From a similar argument as above, we can show that the demand $D^*$ is $\Avc$-respecting and $h$-length in $\Gvc$ as well.
        From the construction of $C$ and \cref{lem:eq-distance}, we have $\sep_{h\cdot s} (C, D^*) = \sep_{h\cdot s}(C_{\mathrm{n}}, D^*) \ge \sep_{h\cdot s}(C', D^*)$.
        We further have 
        \begin{align*}
            \spa_{(h, s)}(C, \Avc) & = \min_{\Avc\text{-respecting } h\text{-length demand } D} \spa_{h\cdot s}(C, D) \\
            & \le \spa_{h\cdot s}(C, D^*)\\
            & = \frac{|C|}{\sep_{h\cdot s}(C, D^*)}\\
            & \le \frac{|C'|}{\sep_{h\cdot s}(C', D^*)}\\
            & \le \phi\\
            %& = \phi
        \end{align*}

\end{proof}

We summarize the equivalence of $(h,s)$-length sparsity with the main theorem in this section.

\thmMutualExpanding*
\begin{proof}
    Given $\Avc$ is $(h,s)$-length $\phi$-expanding in $\Gvc$,
    suppose $\Aec$ is not $(h,s)$-length $\phi$-expanding in $\Gec$, then there exists a cut $C'$ in $\Gec$ such that $\spa_{(h,s)}(C',\Aec) < \phi$.
    From \cref{lem:mutual_sparse_cut}, we can find a cut $C$ in $\Gvc$ where $\spa_{(h,s)}(C,\Avc) < \phi$. This contradicts that $\Avc$ is $(h,s)$-length $\phi$-expanding in $\Gvc$.

    A similar argument applies to the other direction.
\end{proof}

\subsection{Existential Proof of the Decomposition}\label{sec:proof_LCVED}
In this section, we present one of the key results, namely the existence of an expander decomposition for undirected graphs with vertex capacity. 
% with a relatively small cut size. 
Following the reduction introduced before, we apply \cref{thm:directed_edge_ED} to the edge-capacitated graph $\Gec$ to obtain the expander decomposition and then translate the cut into the vertex-capacitated graph. 
We still need to demonstrate that the translated cut forms a valid expander decomposition for $\Gvc$, and the resulting cut has a favorable size.

\thmvertexED*
\begin{proof}
    Given the vertex-capacitated graph $\Gvc$ and the node-weighting $\Avc$, we construct the corresponding edge-capacitated graph $\Gec$ and node-weighting $\Aec$ as defined before.
    Whenever $\Aec$ is not $(h,s)$-length $\phi$-expanding in $\Gec$, we can find a sparse cut in $\Gec$. We can always normalize such a cut in $\Gec$ as shown in \cref{lem:uniform_cut}, and apply it to the graph $\Gec$. This process can be done iteratively until the resulting graph does not admit a sparse cut.  
    This follows the idea from \cref{lem:existential_cut_sequence} and \cref{thm:directed_edge_ED}. 
    We note that we do not apply the normalization to the union of all sparse cuts because the normalization will change the conductance of the node-weighting in the graph applied with the cut.
    %to show that there exists a normalized cut $C_2$ as shown in Claim~\ref{clm:uniform_cut} such that $C_2$ is an $\l$-linked $(h,s)$-length $3\phi$-expander decomposition for node-weighting $A_2$.

    Specifically, from graph $\Gec^{(1)} = \Gec$, if $\Aec$ is already $(h,s)$-length $3\phi$-expanding in $\Gec$, then we are done. Otherwise, we take an arbitrary $h\cdot s$-length cut $C'$ with $h\cdot s$ sparsity strictly smaller than $3\phi$ w.r.t the node-weighting $\Aec$, and from \cref{lem:uniform_cut}, we can construct a normalized cut $\Cec^{(1)}$. We note that the normalization only increases the length, thus for any $\Aec$-respecting demand $D$, we have $\sep_{h \cdot s}(\Cec^{(1)}, D) \ge \sep_{h \cdot s}(C', D)$. Considering the size of the normalized cut can increase by up to threefold, this gives that $\spa_{h\cdot s}(\Cec^{(1)}, \Aec) \le 9\phi$.
    We then apply $\Cec^{(1)}$ to graph $\Gec$ to update $\Gec^{(2)} = \Gec - \Cec^{(1)}$.

    We repeat the above procedure, until we reach to some integer $k$ where $\Aec$ is $(h,s)$-length $3\phi$-expanding in $\Gec^{(k+1)} = \Gec - \sum_{j \le k}\Cec^{(j)}$. We take the union of all these cuts and set a threshold of $h\cdot s$ over the length increase of every edge and vertex. We name the resulting $h\cdot s$-length cut $\Cec$, and by definition, $\Cec$ is a valid expander decomposition for the edge-capacitated graph $\Gec$ and $\Aec$.
    From a similar analysis from \cref{lem:existential_cut_sequence}, 
    we can demonstrate that an upper bound still exists on the summation of all these cut sizes with the potential function technique.
    We have $|\Cec| \le O(2^{8\alpha}\log n)\cdot \phi |\Aec|$. What is more important is that $\Cec$ is normalized.

    From $\Cec$, we can construct a corresponding vertex cut $\Cvc$. Namely, for vertex $u \in V_1$, we assign $\Cvc(u) = \Cec((u_{\IN},u_{\MID}))$ which is the same among the same group of inner edges. For edge $e = (u, v) \in \Evc$, we assign $\Cvc((u,v)) = \Cec((u_{\OUT}, v_{\IN}))$ which is the same between the pair of outer edges.
    We apply $\Cvc$ to $\Gvc$, and we get the resulting graph $\Gvc - \Cvc$. We observe that the edge-capacitated graph $\Gec - \Cec$ exactly corresponds to the graph $\Gvc - \Cvc$. We further note that the node-weighting $\Avc$ corresponds to $\Aec$ due to our construction at the beginning. 
    Since we have shown $\Aec$ to be $(h,s)$-length $3\phi$-expanding in $\Gec - \Cec$,
    from \cref{thm:mutual_expanding}, we have that $\Avc$ is $\phi$-expanding in $\Gvc - \Cvc$. From that $|\Cvc| \le \frac{1}{2}|\Cec|$ and $|\Avc| = |\Aec|$, we also have $|\Cvc| \le O(2^{8\alpha}\log n)\cdot \phi |\Avc|$. This concludes the theorem similarily as \cref{thm:directed_edge_ED}.
\end{proof}

\subsection{Routing Characterization of Length-Constrained Vertex Expansion}\label{sec:routing_LCVE}
Furthermore, due to the structural similarity introduced by the reduction, an efficient routing for a given demand in one type of graph naturally translates into a corresponding favorable routing in the other. 
This reduction then facilitates the transfer of routing characterizations from length-constrained directed expansion to vertex-capacitated graphs as follows:

\thmroutingLCVE*

%To prove the above theorem, we will use the reduction and refer to \cref{thm:routing_LCDE}. 
\begin{proof}
    We first construct a corresponding edge-capacitated graph $\Gec$ with node-weighting $\Aec$ as describe before.
    Since $\Avc$ is $(h,s)$-length $\phi$-expanding in $G$, from \cref{thm:mutual_expanding}, we have that $\Aec$ is also $(h,s)$-length $\phi$-expanding in $\Gec$. From \cref{thm:routing_LCDE}, we have that every $h$-length $\Aec$-respecting symmetric demand can be routed in $\Gec$ with congestion at most $O(\frac{\log N}{\phi})$ and dilation at most $h \cdot s$. 
    We will exploit the routing in $\Gec$ to route demands in the vertex-capacitated graph.

    Namely, for any $h$-length $\Avc$-respecting demand $D$, 
    we can assume it is symmetric.
    As mentioned before, $D$ is also an $h$-length $\Aec$-respecting symmetric demand in the edge-capacitated graph. 
    Then there exists a flow $\Fec$ in $\Gec$ that witnesses the routing with desired congestion and dilation. 
    Consider any $(\umid, \vmid)$-flow path $\Pec$ of $\Fec$, we can construct a corresponding $(u, v)$-flow path $\Pvc$ by following the $\Pec$ as shown in \cref{fig:gvc_gec_path}.
    The observation is that for any edge in $\Gec$ with $\uout$ as the head vertex, it ends at some vertex $\vin$. And from the vertex $\vin$, any path must arrive at $\vout$ in either one or two steps. 
    Thus we can record groups of vertices that are passed by $\Pec$, and create the path $\Pvc$ by traversing vertices in $\Gvc$ corresponding to groups of vertices in $\Gec$ one by one.
    In general, we construct a corresponding flow $\Fvc$ path in $\Gvc$ by following the flow path of $\Fec$.
    From the length and capacity property of the reduction,
    flow $\Fvc$ in $\Gvc$ indeed routes $D$ with at most the same length and up to three times congestion compared to the original congestion, which is still $O(\frac{\log N}{\phi})$.

    If $\Avc$ is not $(h,s)$-length $\phi$-expanding in $\Gvc$, from \cref{thm:mutual_expanding}, we have that $\Aec$ cannot be $(h,s)$-length $3\phi$-expanding in $\Gec$. 
    Then there exists an $h$-length $\Aec$-respecting demand $D$ such that it cannot be routed with congestion at most $\frac{1}{6\phi}$ and length at most $\frac{h\cdot s}{2}$.
    % $\spa_{h\cdot s}(C', D) < 3\phi$.
    % From \cref{thm:routing_LCDE}, we have that any $\frac{h \cdot s}{2}$ routing of $D$ has congestion at least $\frac{1}{6\phi}$.

    Similarly, we have an $h$-length $\Avc$-respecting demand $D$ in $\Gvc$.
    Suppose $D$ can be routed in $\Gvc$ with congestion at most $\frac{1}{6\phi}$ and length at most $\frac{h \cdot s}{2}$, and the routing is witnessed by a flow $\Fvc$.
    We construct a flow $\Fec$ in $\Gec$ by similarily following the trajectory of each flow path in $\Fvc$.
    The resulting flow $\Fec$ witnesses a routing for $D$ with at most the same congestion and the same length. But this gives a contradiction. Thus we also find a $h$-length $\Avc$-respecting demand $D$ that cannot be routed in $\Gvc$ with congestion at most $\frac{1}{6} \phi$ and length at most $\frac{h \cdot s}{2}$.
\end{proof}

%% file: 5_vertex_shortcut.tex
\section{Length-Constrained Vertex-Capacitated Flow Shortcuts}\label{sec:VertexShortcut}

In this section, we show the existence of LC-flow shortcuts in vertex-capacitated graphs. The \Cref{def:HLengthShortcut} below generalizes \Cref{def:LCFlowShortcut} of LC-flow shortcuts in the sense that an additional length parameter $h$ is given and the forward mapping only holds for demands routable in $G$ with congestion $1$ and length at most $h$.

\begin{definition}
[Length-Constrained Flow Shortcut]
\label{def:HLengthShortcut}
Given a graph $G=(V,E)$, we say an edge set $E'$ (possibly with endpoints outside $V$) is an $t$-step $h$-LC-flow shortcut of $G$ with length slack $\lambda$ and congestion slack $\kappa$ if

\begin{itemize}

\item (Forward Mapping) for every demand $D$ routable in $G$ with congestion $1$ and length $h'\leq h$, $D$ is routable in $G\cup E'$ with congestion 1, length $\lambda h'$, and maximum step $t$, and

\item (Backward Mapping) for every demand $D$ on $V(G)$ routable in $G\cup E'$ with congestion 1 and length $h'$, $D$ is routable in $G$ with congestion $\kappa$ and length $h'$.

\end{itemize}

\end{definition}

We note that in \cite{Haeupler2024emulator} there is an analogous but weaker definition of $h$-LC-flow shortcut, in which the forward mapping only guarantee that $D$ is routable in $G\cup E'$ with length $\lambda h$ instead of $\lambda h'$ (and the same congestion and step). That is, the length slack in \Cref{def:HLengthShortcut} is \emph{competitive} in the sense that it upper bounds the ratio between the lengths of the shortcut flow and the original flow. Hence, by choosing a sufficiently large $h$, the total length of vertices and edges in $G$, an $h$-LC-flow shortcut is automatically an LC-flow shortcut.

\begin{theorem}\label{thm:vertex_emulator}
    Given a vertex-capacitated graph $G$ with parameters $\epsilon = \Omega(\frac{1}{\log n})$, there exists a $t$-step LC-flow shortcut $E'$ with length slack $O(1/\epsilon^{3})$, congestion slack $O(n^{O(\epsilon)}\log^{3}n/\epsilon^{2})$, $t = 2^{O(\frac{1}{\epsilon})}$ and size $|E'|\leq O(n^{1+O(\epsilon)}\log n/\epsilon^{2})$.
\end{theorem}

\Cref{thm:vertex_emulator} is the main theorem of this section. In what follows, actually we will focus on constructing $h$-LC-flow shortcut. Setting $h = (m+n)N$ gives the LC-flow shortcut in \Cref{thm:vertex_emulator}.

\subsection{The Construction}
%\thatchaphol{Important: analyze and write down the size (number of edges) in our shortcut too}

%Given any vertex-capacitated graph $G$, we show that there always exists an $h$-length $O(1)$-step flow shortcut graph $G'$ w.r.t $G$ with small length and congestion slacks as shown in \Cref{thm:vertex_emulator}.

\begin{algorithm}
\caption{LC-FlowShortcut$(G,\epsilon,h)$}
\label{algo:emulator}
\begin{algorithmic}[1]
\State Initialize $A_{0} = u_{V(G)}$, where $u_{V(G)}$ denotes the vertex capacity function of $G$.
\State Initialize $s=1/\epsilon$, $\phi = 1/(n^{\epsilon}\kappa)$ ($\kappa = O(n^{O(1/s)}\log n)$ is the cut slack from \Cref{thm:vertex_ED}).
\State Initialize $i\gets 0$.
\While{$|A_{i}|> 0$}
\For{$j$ from $1$ to $\lceil \log_{2}h \rceil$}
\State $h_{j}\gets 2^{j}, h_{\cov,j} = 4h_{j}, h_{\diam,j} = h_{\cov,j}\cdot s$.
\State\label{line:LCED}$C_{i+1,j}\gets$ an $(h_{\diam,j},s)$-length $\phi$-expander decomposition of $A_{i}$ in $G$ by \Cref{thm:vertex_ED}.
\State\label{line:NC}${\cal N}_{i,j}\gets$ a neighborhood cover with covering radius $h_{\cov,j}$, diameter $h_{\diam,j}$ in $G-C_{i+1,j}$ by \Cref{thm:nbhCover}.
\State\label{line:StarGraphs}$H_{i,j}= \bigcup_{S\in{\cal N}_{i,j}} H_{S}$, where $H_{S}$ is the $h_{j}s$-length $A_{i}$-capacitated star graph on $S$.
\EndFor
\State\label{line:NextLevelNW} $A_{i+1} = \sum_{j}\frac{h_{\diam,j}\cdot s}{h_{j}}\cdot \deg_{C_{i+1,j}} = 4s^{2}\cdot\sum_{j}\deg_{C_{i+1,j}}$.
\EndWhile
\State Return $E' = \bigcup_{i,j} E(H_{i,j})$.
\end{algorithmic}
\end{algorithm}

The construction of the flow shortcut graph $G'$ is given by \Cref{algo:emulator}. The star graphs in \Cref{line:StarGraphs} are formally defined in \Cref{def:StarGraphs}.

\begin{definition}[Star Graphs]
\label{def:StarGraphs}
Given a graph $G$ with a node-weighting $A$ and a length parameter $h$, the $h$-length $A$-capacitated star graph on some $S\subseteq V(G)$, denoted by $H_{S}$, has
\[
V(H_{S}) = (\supp(A)\cap S)\cup \{r_{S}\}\text{ and }E(H_{S}) = \{(v,r_{S})\mid v\in V(H_{S})\setminus\{r_{S}\}\},
\]
where the vertex $r_{S}$ is a Steiner vertex serving as the center and $V(H_{S})\setminus\{r_{S}\}$ are original vertices. The length and capacity of each original vertex is unchanged, while $r_{S}$ has length $1$ and capacity $\sum_{v\in S} A(v)$. Each edge $(v,r_{S})$ has length $h$ and capacity $A(v)$.
\end{definition}

In short, \Cref{algo:emulator} mainly constructs an length-constrained expander hierarchy $\{A_{i}, C_{i+1,j}\mid 0\leq i\leq d,1\leq j\leq \lceil \log_{2}h\rceil\}$, where $d$ is the largest $i$ such that $|A_{i}|>0$. We point out that $C_{d+1,j}$ is a zero cut for all $j$. Then the shortcut graph $G'$ is obtained by adding star graphs on neighborhoods of each LC-expander $G-C_{i+1,j}$.

We remark that we do LC-expander decompositions with different length parameters $h_{j}$ at one level because we aim at a shortcut graph with length slack significantly smaller than its step bound. Intuitively, if we only use LC-expanders with length parameter around $h$ to shortcut an original $h$-length flow path $P$, then inevitably each step will have length around $h$, which means the length slack cannot go far below the number of steps. Now, providing LC-expanders with different length parameters, when we want to shortcut a subpath of $P$ with length $h'$ far smaller than $h$, we can choose the appropriate LC-expander to obtain a shortcut with length around $h'$ instead of $h$. Another benefit is that this automatically gives a competitive length slack (this is why in \Cref{def:HLengthShortcut} we define the length slack of $h$-LC-flow shortcut to be competitive).

We first argue the size bound of $E'$. Observe that, in \Cref{line:LCED}, we have $|C_{i+1,j}|\leq \kappa\phi |A_{i}|\leq |A_{i}|/n^{\epsilon}$ by \Cref{thm:vertex_ED}. In \Cref{line:NC}, the width of each neighborhood cover ${\cal N}_{i,j}$ is $\omega = n^{O(1/s)}s = n^{O(\epsilon)}/\epsilon$ by \Cref{thm:nbhCover}. Furthermore, because in \Cref{line:NextLevelNW} we have $|A_{i+1}|\leq 4s^{2}\sum_{j'}|C_{i+1,j'}|\leq 4s^{2}\log h|A_{i}|/n^{\epsilon} = O(\log N/(\epsilon^{2}n^{\epsilon}))|A_{i}|$, we can upper bound $d$ by $d \leq O(\log |A_{0}|/\log(\epsilon^{2}n^{\epsilon}/\log N)) = O(1/\epsilon)$. Finally, by the algorithm, we have
\[
|E'| \le O(d\log h)\cdot \omega\cdot n = O(n^{1+O(\epsilon)}\log n/\epsilon^{2}).
\]

Next we show the quality of the shortcut. Before that, we introduce a helper lemma \Cref{lemma:StarGraphRouting}, which shows the demands that each $H_{i,j}$ can route within small steps. 
%We defer its proof to \Cref{Proof:StarGraphRouting}.
\begin{lemma}
\label{lemma:StarGraphRouting}
For each $i,j$, any demand $\hat{D}$ that is $h_{\cov,j}$-length in $G-C_{i+1,j}$, $A_{i}$-respecting and $u_{V(G)}$-respecting can be routed in $H_{i,j}$ with length $2h_{\diam,j}s+1$, congestion $1$ and step $2$.
\end{lemma}
\begin{proof}
Recall that $H_{i,j} = \bigcup_{S\in{\cal N}_{i,j}} H_{S}$. For each demand pair $(u,v)\in \supp(\hat{D})$, we will assign it in $H_{S}$, where $S\in {\cal N}_{i,j}$ is an arbitrary cluster containing both $u$ and $v$. Note that such an $S$ must exist because $\hat{D}$ is $h_{\cov,j}$-length in $G-C_{i+1,j}$ and ${\cal N}_{i,j}$ has covering radius $h_{\cov,j}$ in $G-C_{i+1,j}$. Now, for each $S\in{\cal N}_{i,j}$, its assigned demand $\hat{D}_{S}$ has $V(\hat{D}_{S})\subseteq S$ and it is $A_{i}$-respecting and $u_{V(G)}$-respecting, so $\hat{D}_{S}$ can be routed in $H_{S}$ with length $2h_{\diam.j}s+1$, congestion $1$ and step $2$ by the definition of $H_{S}$.
\end{proof}

Now we show that the shortcut $E'$ constructed by \Cref{algo:emulator} has length slack $O(1/\epsilon^{3})$, congestion slack $n^{O(\epsilon)}$ and step $2^{O(1/\epsilon)}$. To do this, it suffices to show the quality of the forward mapping and backward mapping, i.e. \Cref{lem:forward_mapping} and \Cref{lem:backward_mapping}, whose proofs are given in \Cref{sect:ForwardMapping} and \Cref{sect:BackwardMapping}. Let $G' = G\cup E'$ be the shortcut graph.

\begin{lemma}[Forward Mapping]\label{lem:forward_mapping}
For any feasible $h$-length flow $F$ in $G$, there is a feasible flow $F'$ routing $\Dem(F)$ in $G'$ with $\leng(F')\leq \leng(F)\cdot O(1/\epsilon^{3})$ and $\step(F')\leq 2^{O(1/\epsilon)}$. 
\end{lemma}

\begin{lemma}[Backward Mapping]\label{lem:backward_mapping}
    For any feasible flow $F'$ in $G'$ such that $V(\Dem(F'))\subseteq V(G)$, there is a flow $F$ routing $\Dem(F')$ in $G$ with $\leng(F)\leq \leng(F')$ and $\cong(F)\leq n^{O(\epsilon)}$. 
\end{lemma}

\subsection{Forward Mapping: Proof of \Cref{lem:forward_mapping}}
\label{sect:ForwardMapping}
We will employ a top-down argument. Start from the top level $d$. At the beginning of processing a level $i\geq 0$, we are given a feasible flow $F_{i}$ with the following invariant:
%\begin{itemize}
%\item 
each flow path $P\in\path(F_{i})$ has $\leng(P,G-C_{i+1,j})\leq h_{\cov,j}=4h_{j}$, where $j$ is the minimum index such that $h_{j}\geq \leng(P,G)$ (which means $h_{j}/2\leq \leng(P,G)\leq h_{j}$).
%\end{itemize}
Initially at the top level $d$, we set $F_{d} = F$. Note that $F_{d}$ satisfies the invariant above because $C_{d+1,j}$ is a zero cut for any $j$.

First, at the bottom level $i=0$, we can easily shortcut every flow path in $F_{0}$. For each star graph $H_{0,j}$, we assign it a demand $\hat{D}_{0,j}$ which sums over $\Dem(P)$ for each flow path $P\in\path(F_{0})$ such that $j$ is the minimum index with $h_{j}\geq \leng(P,G)$. Observe that each $\hat{D}_{0,j}$ is an $h_{\cov,j}$-length in $G-C_{1,j}$ and $A_{0}$-respecting, so it can be routed in $H_{0,j}$ with length $2h_{\diam,j}s+1$, congestion $1$ and step $2$ by \Cref{lemma:StarGraphRouting}.

From now on we consider levels $i\geq 1$. When processing a level $i$, for each flow path $P\in\path(F_{i})$, we may shortcut some subpaths of $P$ using shortcut edges in $H$. The subpaths of $P$ that have not been shortcut will be added to $F_{i-1}$, meaning that they are deferred to lower levels to get shortcut. At the end, the final $F_{i-1}$ should ensure the invariants above, and we proceed to the lower level $i-1$.

Now we will explain the shortcut at a level $i\geq 1$ in detail. Fix a flow path $P\in\path(F_{i})$. Let $j$ be the minimum index such that $h_{j}\geq \leng(P,G)$. We consider two cases.

\medskip

\noindent{\textbf{Case 1: Defer.}} When $4s^{2}\cdot \sum_{j'}C_{i,j'}(P)\leq 3$, we simply add $P$ to $F_{i-1}$. We have
\[
\leng(P,G-C_{i,j}) = \leng(P,G) + h_{\diam,j}\cdot s\cdot C_{i,j}(P)\leq 4h_{j}\leq h_{\cov,j},
\]
where the inequality is by $\leng(P,G)\leq h_{j}$, $h_{\diam,j} = 4sh_{j}$ and $C_{i,j}(P)\leq 3/(4s^{2})$. Therefore, in this case, the flow path added to $F_{i-1}$ satisfies the invariant.

\medskip

\noindent{\textbf{Case 2: Shortcut.}} Now suppose $4s^{2}\cdot \sum_{j'}C_{i,j'}(P)>3$. Let $u$ and $v$ be $P$'s endpoints. We say $u$ is the left side and $v$ the right side. For each $0\leq k\leq |P|$, we refer to $w_{k}$ as the $P$-vertex with $k$ steps away from $u$. In particular, $w_{0} = u$ and $w_{|P|} = v$.

We now define two functions $x:V(P)\to \mathbb{R}$ denoting the \emph{budgets} of $P$-vertices from the left and the right respectively: for each vertex $w_{k}\in V(P)$,
\[
x_{P}(w_{k}) = 4s^{2}\sum_{j'}C_{i,j'}(w_{k-1},w_{k}) +  C_{i,j'}(w_{k}) + C_{i,j'}(w_{k},w_{k+1}).
\]
In particular, $x_{P}(w_{0}) = 4s^{2}\sum_{j'}C_{i,j'}(w_{0}) + C_{i,j'}(w_{0},w_{1})$ and $x_{P}(w_{|P|}) = 4s^{2}\sum_{j'}C_{i,j'}(w_{|P|-1},w_{|P|}) + C_{i,j'}(w_{|P|})$. Let $k_{L}$ be the minimum index such that $\sum_{0\leq k\leq k_{L}} x_{P}(w_{k}) \geq 1$, and symmetrically let $k_{R}$ be the maximum index such that $\sum_{k_{R}\leq k\leq |P|} x_{P}(w_{k})\geq 1$. To avoid clutter, we let $L = \{w_{0},...,w_{k_{L}}\}$ and $R = \{w_{k_{R}},...,w_{|P|}\}$. 
The following \Cref{claim:LRoverlap} says that $L$ and $R$ have at most one common vertex.
%A key observation is that $k_{L}\leq k_{R}$.
%In other words, $L$ and $R$ have at most one common vertex. The reason is as follows. By the definition of $k_{L}$ and $k_{R}$, we have $\sum_{0\leq k\leq k_{L}-1}x_{P}(w_{k})<1$ and $\sum_{k_{R}+1\leq k\leq |P|}x_{P}(w_{k})<1$. However, $\sum_{0\leq k\leq |P|}x_{P}(w_{k})\geq 4s^{2}\sum_{j'}C_{i,j'}(P)>3$. This means there exists a vertex $w_{k}$ with $k_{L}-1<k<k_{R}+1$, which implies $k_{L}\leq k_{R}$

\begin{claim}
\label{claim:LRoverlap}
$k_{L}\leq k_{R}$.
\end{claim}
\begin{proof}
By the definition of $k_{L}$ and $k_{R}$, we have $\sum_{0\leq k\leq k_{L}-1}x_{P}(w_{k})<1$ and $\sum_{k_{R}+1\leq k\leq |P|}x_{P}(w_{k})<1$. However, $\sum_{0\leq k\leq |P|}x_{P}(w_{k})\geq 4s^{2}\sum_{j'}C_{i,j'}(P)>3$. This means there exists a vertex $w_{k}$ with $k_{L}-1<k<k_{R}+1$, which implies $k_{L}\leq k_{R}$
\end{proof}

We assign each vertex $w\in L\cup R$ a \emph{load} $x'_{P}(w)\leq x_{P}(w)$, satisfying that $\sum_{w\in L}x'_{P}(w)=1$ and $\sum_{w\in R}x'_{P}(w)=1$. We consider the following three-phase strategy to route $F_{i}(P)$ flow units from $u$ to $v$. \emph{Phase 1:} the vertex $u$ sends $F_{i}(P)$ flow units and each $w\in L$ receives $F_{i}(P)x'_{P}(w)$ units. \emph{Phase 2:} Each vertex $w\in L$ sends $F_{i}(P)x'_{P}(w)$ units and each vertex $w\in R$ receives $F_{i}(P)x'_{P}(w)$ units. \emph{Phase 3:} Each vertex $w\in R$ sends exactly $F_{i}(P)x'_{P}(w)$ units, and the vertex $v$ receives $F_{i}(P)$ units.
%\begin{enumerate}
%\item The vertex $u$ sends $F_{i}(P)$ flow units and each $w\in L$ receives $F_{i}(P)x'_{P}(w)$ units.
%\item Each vertex $w\in L$ sends $F_{i}(P)x'_{P}(w)$ units and each vertex $w\in R$ receives $F_{i}(P)x'_{P}(w)$ units.
%\item Each vertex $w\in R$ sends exactly $F_{i}(P)x'_{P}(w)$ units, and the vertex $v$ receives $F_{i}(P)$ units.
%\end{enumerate}

\medskip

\noindent{\underline{The First and Third Phases.}} Roughly speaking, for the first and third phases, we will add their corresponding original flows into $F_{i-1}$, meaning that they will be deferred to lower levels to get shortcut. 

Regarding the first phase, recall that for each $w_{k}\in L$, we want to route $F_{i}(P)x'(w_{k})$ units from $u$ to $w_{k}$. To do this, we add into $F_{i-1}$ a flow path $P_{k} = P[w_{0},w_{k-1}]$ with value $F_{i}(P)\cdot x'(w_{k})$. That is, we require the lower levels to give a shortcut that routes $F_{i}(P)\cdot x'(w_{k})$ from $u$ to $w_{k-1}$ (the $P$-vertex one step closer to $u$ than $w_{k}$). Then, we route $F_{i}(P)x'(w_{k})$ units from $w_{k-1}$ to $w_{k}$ using the original edge $(w_{k-1},w_{k})\in E(G)$.

The third phase is handled in a similar way. For each $w_{k}\in R$, we route $F_{i}(P)x'(w_{k})$ flow units from $w_{k}$ to $w_{k+1}$ using the original edge $(w_{k-1},w_{k})$, and then we add into $F_{i-1}$ a flow path $P_{k} = P[w_{k+1},v]$ with value $F_{i}(P)\cdot x'(w_{k})$.

To proceed to the lower level $i-1$, it remains to show that the flow paths added into $F_{i-1}$ satisfy the invariant. Consider a flow path $P_{k} = P[u,w_{k-1}]$ added from the first phase (where $w_{k}\leq L$, i.e. $0\leq k\leq k_{L}$). We want to show that $\leng(P_{k}, G - C_{i,j_{k}})\leq h_{\cov,j_{k}}$, where $j_{k}$ is the minimum index such that $h_{j_{k}}\geq \leng(P_{k},G)$. By the definition of the budget function $x$, we have
%\[
$\sum_{j'}C_{i,j'}(P_{k})\leq \frac{1}{4s^{2}}\sum_{0\leq k'\leq k-1}x_{P}(w_{k'})< 1/(4s^{2})$,
%\]
where the second inequality is by $k\leq k_{L}$ and $\sum_{0\leq k'\leq k_{L}-1}x_{P}(w_{k'})<1$ (from the definition of $k_{L}$). In particular, $C_{i,j_{k}}(P_{k})\leq 1/(4s^{2})$. Because $C_{i,j_{k}}$ is an $(h_{\diam,j_{k}}s)$-length moving cut, we have
\[
\leng(P_{k}, G-C_{i,j_{k}}) = \leng(P_{k},G) + h_{\diam,j_{k}}\cdot s\cdot C_{i,j_{k}}(P_{k})\leq 2h_{j_{k}}\leq h_{\cov,j_{k}},
\]
as desired, where the first inequality uses $h_{\diam,j_{k}} = 4s h_{j_{k}}$. By a similar argument, we can show that each flow path added from the third phase also satisfies the invariant, so we will not explain it in detail.

\medskip

\noindent{\underline{The Second Phase.}} The second phase is where the shortcut happens. We define an arbitrary (multi-commodity) demand $\hat{D}_{i,P}$ capturing this single-commodity demand. Namely, $\hat{D}_{i,P}$ satisfies that (1) $|\hat{D}_{i,P}| = F_{i}(P)$; (2) for each $w\in L$, $\hat{D}_{i,P}(w,\cdot) = F_{i}(P)x'_{P}(w)$; and (3) for each $w\in R$, $\hat{D}_{i,P}(\cdot,w) = F_{i}(P)x'_{P}(w)$. We will \emph{assign} $\hat{D}_{i,P}$ to $H_{i,j}$, meaning that we route $\hat{D}_{i,P}$ using shortcut edges in $H_{i,j}$. 

By the above assignment, for each $H_{i,j}$ at level $i$, its total assigned demand, denoted by $\hat{D}_{i,j}$, sums over $\hat{D}_{i,P}$ of all $P\in\path(F_{i})$ s.t. $j$ is the minimum index with $h_{j}\geq \leng(P,G)$. The following \Cref{lemma:Reroute} showing that $\hat{D}_{i,j}$ can be routed with low steps, small length and congestion $1$.

\begin{lemma}
\label{lemma:Reroute}
For each $H_{i,j}$ at level $i$, its total assigned demand $\hat{D}_{i,j}$ can be routed in $H_{i,j}$ with length $2h_{\diam,j}s+1$, congestion $1$ and step $2$.
\end{lemma}
\begin{proof}
First, we show that the demand $\hat{D}_{i} := \sum_{P\in\path(F_{i})} \hat{D}_{i,P}$ is $A_{i}$-respecting, which means $\hat{D}_{i,j}\preceq \hat{D}_{i}$ is also $A_{i}$-respecting. To see this, consider any vertex $w\in V(G)$, and we have
\begin{align*}
\max\{\hat{D}_{i}(w,\cdot),\hat{D}_{i}(\cdot,w)\}&\leq \sum_{P\in\path(F_{i})}\max\{\hat{D}_{i,P}(w,\cdot),\hat{D}_{i,P}(\cdot,w)\}
\leq \sum_{P\in\path(F_{i})\text{ s.t. }w\in P} F_{i}(P)x_{P}(w)\\
&\leq 4s^{2}\left(u_{G}(w)\sum_{j'}C_{i,j'}(w) + \sum_{e\in E(G)\text{ 
incident to }w}u_{G}(e)\cdot \sum_{j'}C_{i,j'}(e)\right)\\
& = A_{i}(w).
\end{align*}
Here the key step is the third inequality, which is because (1) $F_{i}$ is feasible, and (2) $x_{P}(w)$ is the total level-$i$ cut values of $w$ and its incident $P$-edges scaled up by $4s^2$.

Next, $\hat{D}_{i,j}$ is $h_{\cov,j}$-length in $G-C_{i+1,j}$ because for each demand pair $(u,v)\in \supp(\hat{D}_{i,j})$, $u,v$ is on some $P\in\path(F_{i})$ and $\leng(P,G-C_{i+1,j})\leq h_{\cov,j}$.

Finally, since $\hat{D}_{i,j}$ is $h_{\cov,j}$-length in $G-C_{i+1,j}$, $A_{i}$-respecting and trivially $u_{V(G)}$-respecting (since $\hat{D}_{i}$ can be routed in $G$ according to the feasible flow $F_{i}$), $\hat{D}_{i,j}$ can be routed in $H_{i,j}$ with length $2h_{\diam,j}s+1$, congestion $1$ and step $2$ by \Cref{lemma:StarGraphRouting}.
\end{proof}

\medskip

\noindent{\textbf{Quality of the Forward Mapping.}} We now show that $\Dem(F)$ can be routed in $G'=G\cup \bigcup_{i,j} H_{i,j}$ with length $\leng(F)\cdot O(1/\epsilon^{3})$, congestion $1$ and step $2^{O(1/\epsilon)}$.

\medskip

\noindent{\underline{Length and Step.}} We first argue the length and step bounds using induction on levels. It suffices to prove the following statement for each level $i$:
%\begin{itemize}
%\item 
for each flow path $P\in\path(F_{i})$, the above forward mapping scheme maps it to a flow $\hat{F}_{i,P}$ in $G\cup \sum_{i'\leq i,j'}H_{i,j'}$ with length $20(i+1)\cdot s^{2}\leng(P,G)$ and step $6\cdot 2^{i+1}-4$. 
%\end{itemize}
The statement for the top level $i=d$ implies the length and step bounds after plugging in $d\leq O(1/\epsilon)$ and $s = 1/\epsilon$.

Consider the base case $i=0$. A flow path $P\in\path(F_{0})$ is shortcut by $H_{0,j}$, where $j$ is the minimum index with $h_{j}\geq \leng(P,G)$, which implies $h_{j}/2\leq \leng(P,G)$. As we discussed above, the shortcut has length $2h_{\diam,j}s+1$ and step $2$. Since $2h_{\diam,j}s+1\leq 8s^{2}h_{j}+1\leq 20s^{2}\leng(P,G)$, we complete the proof of the base case.

Consider an inductive step $i\geq 1$. Assume the statement holds for $i-1$. For each flow path $P\in\path(F_{i})$, if it belongs to Case 1, the length and step of $\hat{F}_{i,P}$ can be directly bounded by the induction hypothesis, so we assume it belongs to Case 2 from now on. Recall that $\hat{F}_{i,P}$ is the concatenation of routings from three phases. 
%\begin{itemize}
%\item 

%\medskip\noindent
\emph{Phases 1 and 3.} Consider the first-phase routing (the third-phase routing is analogous). By the induction hypothesis, its corresponding flow path $P_{k} = P[u,w_{k-1}]$ deferred to $F_{i-1}$ will be mapped to a flow $\hat{F}_{i-1,P_{k}}$ in the shortcut graph with length $20is^{2}\cdot\leng(P[u,w_{k-1}],G)$ and step $6\cdot 2^{i}-4$. The extra edge at the end of the routing $(w_{k-1},w_{k})$ has length $\leng(P[w_{k-1},w_{k}],G)$ and step $1$. In conclusion, the first-phase routing has length $\max_{w_{k}\in L}(20is^{2}\cdot \leng(P[u,w_{k-1}],G) + \leng(P[w_{k-1},w_{k}]))\leq 20is^{2}\cdot\leng(P[u,w_{k_{L}}],G)$ and step $6\cdot 2^{i}-3$.

%\noindent
\emph{Phase 2.}
The second-phase routing is on $H_{i,j}$, where $j$ satisfies $h_{j}/2\leq \leng(P,G)$. For each $w_{k}\in L$, by \Cref{lemma:Reroute} and an argument similar to the base case, the second-phase routing has length $20s^{2}\leng(P,G)$ and step $2$. 

Finally, by \Cref{claim:LRoverlap}, $\hat{F}_{i,P}$ has length $20(i+1)s^{2}\cdot \leng(P,G)$ and step $6\cdot 2^{i+1}-4$.

\medskip

\noindent{\underline{Congestion.}} The congestion in $\bigcup_{i,j} H_{i,j}\subseteq G'$ is $1$ by \Cref{lemma:Reroute}. The congestion on $G\subseteq G'$ is also $1$ because the flow we left in the original graph (in the second phase) is at most the feasible flow $F_{i}$.

\subsection{Backward Mapping: Proof of \cref{lem:backward_mapping}}
\label{sect:BackwardMapping}

We first decompose the flow $F'$ in $G'$ into flows in subgraphs $G$ and $H_{i,j}$ (for all $i,j$) of $G'$, denoted by $F'_{G}$ and $F'_{i,j}$ respectively. Formally speaking, $F'_{i,j}$ collects, for all $P\in\path(F')$, all flow subpaths in $P\cap H_{i,j}$ (the values of these subpaths are still $F'(P)$), and $F'_{G}$ collects all flow subpaths in $P\cap G$ for all $P\in\path(F')$. 

Next, we show that $\Dem(F'_{i,j})$ can be routed by a flow $F_{i,j}$ in $G$ with length $h_{\diam,j}\cdot s$ and congestion $O(\omega\cdot\log N/\phi)$, where $\omega = n^{O(\epsilon)}\log n$ is the width of the neighborhood cover ${\cal N}_{i,j}$ and $\phi = 1/n^{O(\epsilon)}$ is the expansion of the hierarchy. Because $F'_{i,j}$ is a feasible flow in $H_{i,j}$, its demand $\Dem(F'_{i,j})$ is $(\omega\cdot A_{i})$-respecting and $h_{\diam,j}$-length in $G-C_{i+1,j}$ by \Cref{def:StarGraphs} of star graphs and the diameter of ${\cal N}_{i,j}$. Because $A_{i}$ is $(h_{\diam,j},s)$-length expanding in $G-C_{i+1,j}$, $\Dem(F'_{i,j})$ can be routed in $G-C_{i+1,j}$ with length $h_{\diam,j}\cdot s$ and congestion $O(\log N/\phi)$ by \Cref{thm:routing_LCVE}.

Finally, we let the flow $F$ in $G$ routing $\Dem(F')$ be the concatenation of $F'_{G}$ and all $F_{i,j}$. We have $\leng(F)\leq \leng(F')$ because each $F_{i,j}$ has length $h_{\diam,j}\cdot s\leq \leng(F'_{i,j}) = 2h_{\diam,j}s + 1$. The congestion of $F$ is \[
1 + (d\log h)\cdot O(\omega\cdot\log N/\phi) = n^{O(\epsilon)}\log n\log^{2}N/\epsilon^{2} = O(n^{O(\epsilon)}\log^{3}n/\epsilon^{2}) 
\]
because $F'_{G}$ has congestion $1$ in $G$, each $F_{i,j}$ has congestion $O(\omega\cdot\log N/\phi)$, and there are totally $d\cdot \log h$ many $F_{i,j}$.

%% file: appendix_3_directed.tex
\section{Length-Constrained Directed Expansion (Continued)}

\subsection{Connection to Classic Directed Expander Decomposition}
\label{sec:appendix_directedED}

%\thatchaphol{TODO: add reasonable narrative for this section.}

%\paragraph{Classic Directed Expanders.}
%\future{$|E(S, V\setminus S)|$ or $\sum_{e \in E(S, V\setminus S)} u_G(e)$}
The classic setting of a cut in a graph is a bipartition of the vertex set. 
A cut is considered sparse if few edges connect the two partitions.
\begin{definition}[Classic Directed Sparse Cut]
Given a directed graph $G = (V,E)$, 
let $E(A, B)$ denote the set of edges with head vertices in $A$ and tail vertices in $B$,
a cut $(S,V\setminus S)$ where $\vol_{G}^{sum}(S) \le \vol_{G}^{sum}(V\setminus S)$ is $\phi$-out-sparse if $\sum_{e \in E(S, V\setminus S)} u_G(e) < \phi \cdot \vol_{G}^{sum}(S)$. Similarly, the cut is $\phi$-in-sparse if $\sum_{e \in E(V\setminus S, S)} u_G(e) < \phi \cdot \vol_{G}^{sum}(S)$.
\end{definition}
In directed graphs, the standard definition of an expander only requires that the graph be free from sparse cuts in both directions. That is, a directed graph is a $\phi$-expander if and only if it is strongly connected and there is no $\phi$-in-sparse or $\phi$-out-sparse cut.
Below, we extend the above definition by interpreting an expander as the union of strongly connected components where each does not admit a sparse cut.
\begin{definition}[Classic Directed Expander]
    Given a directed graph $G = (V,E)$, let $X_1,...,X_n \subset V$ be a partition of $V$ and each $X_i$ is a strongly connected component (SCC) of $G$. For $\phi \in (0,1)$, a graph $G$ is a classic directed $\phi$-expander if every induced graph $G[X_i]$ has no $\phi$-out-sparse cut and $\phi$-in-sparse cut.
\end{definition}

%This helps to succinctly state the classic directed expander decomposition and makes generalization easier to the length-constrained setting.

%\begin{definition}[Classic Directed Expander Decomposition]
%    Given a directed graph $G = (V, E)$, a $\phi$-expander decomposition is a subset of edges $E' \subset E$ where $\sum_{e \in E'} u_G(e) \le \phi\cdot \vol_G^{sum}(V)$ such that $G \setminus E'$ is a $\phi$-expander.
%\end{definition}

Now we discuss the connection between classic directed expanders and length-constrained directed expanders, assuming we only consider \emph{pure cuts} when defining LC-directed expanders. Here a pure cut are just a classic edge cuts, i.e. a subset of edges. We can also interpret a pure cut as a moving cut with cut values either $0$ or $1$. 

We say a graph $G$ is an $(h,s)$-length $\phi$-expander if $\deg_{G}^{sum}$ is $(h,s)$-length $\phi$-expanding in $G$. 
Then in fact, for sufficient large $h$, e.g. $h = n^{2}N$, $G$ is an $(h,s)$-length $\phi$-expander if and only if $G$ is a classic $\Theta(\phi)$-expander. In other words, LC-directed expanders and classic directed expanders are equivalent when we consider only pure cuts and very large $h$.

It is relatively easy to see that if $G$ is not a classic $\phi$-expander, then $G$ is not an $(h,s)$-length $O(\phi)$-expander. Let $(S,X\setminus S)$ for some SCC $X$ of $G$ (with $\vol^{sum}_{G[X]}(S)\leq \vol^{sum}_{G[X]}(X\setminus S)$) be the witnessing $\phi$-sparse (say $\phi$-out-sparse) cut. Then the pure cut $C = E(S,X\setminus S)$ is an $(h,s)$-length $\phi$-sparse cut of $G$, because we can construct a $\deg^{sum}_{G}$-respecting $h$-length symmetric demand $D$ with $|D| \geq \Omega(\vol^{sum}_{G[X]}(S))$ by considering a bipartite graph (with multiple edges) between $X$ and $X\setminus S$ such that each vertex $v \in X$ incident to exactly $\deg^{sum}_{G[X]}(v)$ edges and each vertex $v\in X\setminus S$ incident to at most $\deg^{sum}_{G[X]}(v)$ edges.

Now we show that if $G$ is not an $(h,s)$-length $\phi$-expander, then $G$ is not a classic $\phi$-expander. Let $C$ be an $(h,s)$-length $\phi$-sparse cut of $G$ with witnessing demand $D$. Because $D$ is $h$-length and symmetric, $D$ has no inter-SCC demand pairs. This means there exists an SCC $X$ of $G$ such that $\spars_{hs}(C_{X},D_{X})\leq \phi$, where $C_{X}$ and $D_{X}$ are $C$ and $D$ restricted on $G[X]$. Let $X'_{1},X'_{2},...,X'_{k}$ be the SCCs of $G[X]\setminus C_{X}$. We say a cut $(S,X\setminus S)$ is \emph{valid} if
\begin{itemize}
\item it respects the partition $\{X'_{1},X'_{2},...,X'_{k}\}$, i.e. $S$ is the union of a subset of $X'_{i}$, and
\item $S$ has no out-edges in $G\setminus C_{X}$.
\end{itemize}
Observe that any valid cut $(S,X\setminus S)$ has $\min\{|E(S,X\setminus S)|,|E(X\setminus S,S)|\}\leq |C_{X}|$. Thus it remains to find a valid cut $(S^{*},X\setminus S^{*})$ with $\min\{\vol^{sum}_{G[X]}(S^{*}),\vol^{sum}_{G[X]}(X\setminus S^{*})\} = \Omega(|D_{X}|)$.

Let $d(X'_{i}) = \sum_{u\in X'_{i},v\in X}D_{X}(u,v)$. By definition, $\sum_{X'_{i}} d(X'_{i}) = |D_{X}|$. Because $D_{X}$ is $\deg^{sum}_{G[X]}$-respecting (since $D$ has no inter-SCC demand pairs), each $X'_{i}$ has $\vol^{sum}_{G[X]}(X'_{i})\geq d(X'_{i})$. Therefore, we just need to find a valid cut $(S^{*},X\setminus S^{*})$ with $\sum_{X'_{i}\subseteq S^{*}}d(X'_{i}) = \Omega(|D_{X}|)$, which is an easy exercise by considering a topological order of $\{X'_{1},...,X'_{k}\}$.

%This is because for any pure cut $C$, 

%We note that it is a natural generalization from the classic directed expander.
%We note that any $\phi$-in-sparse cut can be interpreted as a moving cut with large enough length parameter.
%This gives the same sparsity w.r.t the degree demand where each vertex from the part with smaller volume sends out the same amount of value as its degree.  

\subsection{Routing Characterization: Proof of \Cref{thm:routing_LCDE}}\label{sec:appendix_LC_routing}

In this section, we give the proof of \cref{thm:routing_LCDE}. We note that each direction in \Cref{thm:routing_LCDE} can be explained with the following \Cref{thm:non_expanding_bad_routing} and \Cref{thm:flowcutgap} respectively.

We start with the easier direction \Cref{thm:non_expanding_bad_routing}, which shows that there exists a worst-case demand that cannot be routed with desirable congestion and length given the low conductance.

\begin{theorem}\label{thm:non_expanding_bad_routing}
    In a directed graph $G$, suppose that for a demand $D$ there exists a $2h$-length moving cut $C$ with $2h$-length sparsity $\phi = \spa_{2h}(C, D)$. Then, any $h$-length routing of $D$ has congestion at least $\frac{1}{2\phi}$.
\end{theorem}
\begin{proof}
    Without loss of generality, we can assume all demand pairs $(u,v)\in \supp(D)$ is $2h$-length separated by $C$, i.e. $\sep_{2h}(C,D) = |D|$, otherwise we consider the sub-demand of $D$ that is $2h$-length separated by $C$ instead.

    Let $F$ be an $h$-length flow routing $D$ in $G$ with minimum congestion denoted by $\gamma$. For each flow path $P\in\path(F)$, we have $\leng(P,G)\leq h$ but $\leng(P,G-C)\geq 2h$, so 
    \[
    \sum_{P\in\path(F)}F(P)\cdot \sum_{e\in P}\ell_{C,2h}(e)\geq \vvalue(F)\cdot h = |D|\cdot h.
    \]
    On the other hand, we have
    \begin{align*}
    \sum_{P\in\path(F)} F(P)\cdot \sum_{e\in P}\ell_{C,2h}(e) &= \sum_{e\in E(G)}\ell_{C,2h}(e)\sum_{P\in\path(F)\text{ s.t. }P\ni e} F(P)\\
    &\leq\sum_{e\in E(G)}\ell_{C,2h}(e)\cdot \gamma\cdot u(e)\\
    &= \gamma\cdot\sum_{e\in E(G)}2h\cdot C(e)\cdot u(e)\\
    &= \gamma\cdot 2h\cdot |C|.
    \end{align*}
    Therefore, we have $\gamma \geq \frac{|D'| \cdot h}{2h \cdot |C|} = \frac{1}{2\phi}$ where we use $\frac{|C|}{|D'|} = \phi$ 
\end{proof}

On the other hand, \Cref{thm:flowcutgap} says that if there exists a lower bound over the sparsity of some demand $D$, we can also construct an $h$-length routing of $D$ with a bounded congestion.
\begin{theorem}\label{thm:flowcutgap}
    In a directed graph $G$, suppose that $D$ is a demand such that every $h$-length moving cut $C$ has sparsity at least $\phi$ for a $D$, i.e., $\spa_{h}(C,D) \ge \phi$. Then, $D$ can be routed along $h$-length paths with congestion at most $O(\frac{\log N}{\phi})$.
\end{theorem}

The remainder of this section is the proof 
\Cref{thm:flowcutgap}. For a given $h$-length demand $D$ which sends $d_i$ amount of demand from some vertex $s_i$ to some vertex $t_i$, we consider the (exponential size) concurrent multicommodity flow LP and its dual given in \Cref{fig:primaldual-flows}. In it we denote with ${\cal P}_i$ the (exponential size) set of all $h$-length directed paths that go from $s_i$ to $t_i$:

\newcommand{\primal}{\textsc{ConcurrentFlow}}
\newcommand{\dual}{\textsc{Cut}}

\begin{figure}[ht]
	\begin{center}
		\begin{tabular}{rl|rl}
			\multicolumn{2}{c|}{Primal: \primal$(D,h)$} & \multicolumn{2}{c}{Dual: \dual$(D,h)$
				} \\ \hline  
			maximize & $z$ & minimize & $L = \sum_{e \in E} u(e) \cdot \ell_e$ \\
			subject to: & & subject to: & \\
			
			$\forall i\in [k]$: &  $\sum_{p\in {\cal P}_i} f_i(p) \geq z \cdot d_i$ & $\forall i\in [k], p \in {\cal P}_i$: & $\sum_{e\in p} \ell_e \geq c_i$ \\
			
			$\forall e \in E$: & $\sum_{p\ni e} f_i(p) \leq u(e)$ & & $\sum_{i \in [k]} d_i c_i \geq 1$ \\
			
			$\forall i\in [k],p$: & $f_i(p) \geq 0$ & $\forall e\in E$: & $\ell_e \geq0$ \\
			
			& & $\forall i\in [k]$: & $c_i\geq 0$				
		\end{tabular}
	\end{center}
	\vspace{-0.5cm}
	\caption{The concurrent flow LP relaxation and its dual.}
	\label{fig:primaldual-flows}
\end{figure}

We note that the primal LP exactly computes a routing for $z\cdot D$ while constraining the congestion to be no more than 1. It is equivalent to the inverse of the optimal $h$-length congestion of a fractional routing for $D$.

\begin{lemma}\label{lem:flow-LP-equals-congestion}
Let $z, \{f_i(p) \mid i\in [k], p\in {\cal P}_i\}$ be an optimal solution for \primal$(D,h)$. Then the optimal $h$-length congestion for any fractional $h$-length routing of $D$ is exactly $\frac{1}{z}$. 
\end{lemma}
\begin{proof}
Scale the flows by $\frac{1}{z}$. They now satisfy the demands (thanks to the first set of primal constraints) and cause at most congestion $\frac{1}{z}$ on any edge (thanks to the second set of primal constraints). Therefore the optimal $h$-length congestion of fractionally routing $D$ is at most $\frac{1}{z}$. 

On the other hand, any fractional $h$-length routing of $D$ with congestion $\gamma$ can be scaled down by a factor of $\frac{1}{\gamma}$ to give a feasible solution for \primal$(D,h)$ with $z = \frac{1}{\gamma}$. Thus if the optimal $h$-length congestion for any fractional $h$-length demand is $\gamma' < \frac{1}{z}$, then we can get a feasible solution with $z' = \frac{1}{\gamma'} > z$, which contradicts that $z$ is the optimal value. 
\end{proof}

Next, we show that a feasible solution to the dual LP with value $L$ can be transformed into a $h$-length moving cut $C$ for $D$ of sparsity $\spa_h(C,D) = O(L \cdot \log N)$. The intuition is to think of the dual variables $c_i$ as indicating whether or not (or to what extend) we are separating the demands $d_i$ while the $\ell_i$ variables will correspond to scaled length increases in the moving cut. The scaling is such that we are (fractionally) cutting one unit of demand (as forced by the second set of dual demands). With this intuition in mind, the following lemma \cite{haeupler2022hopexpander} shows that the $c_i$ values can be ``rounded" into an appropriate subset of demands to separate, together with the appropriate scaling factor.  

\begin{lemma}\label{lem:bucketing-lemma}
	Given sequences $c_1, \ldots, c_k, d_1, \ldots, d_k \in \reals_{\ge 0}$ with $\sum_{i\in [k]} d_i \cdot c_i \ge 1$ there exists a non-empty subset $I \subseteq [k]$ with $\min_{i \in I} c_i \ge \frac{1}{\alpha \cdot \sum_{i\in I}d_i}$ for $\alpha = 1 + \ln \left( \frac{\sum_{i\in [k]} d_i}{\min_{i\in [k]} d_i} \right )$.
\end{lemma}
\begin{proof}
	Suppose, without loss of generality, that $c_1 \geq c_2 \geq \ldots c_k$ and assume for the sake of contradiction that none of the sets $[1], [2], \ldots, [k]$ satisfy the condition. In other words, if we let $d([i]) = \sum_{j = 1}^i d_j$, then $c_i < \frac{1}{\alpha} \cdot \frac{1}{d([i])}$ for all $i \in [k]$. Multiplying both sides by $d_i$ and summing them up, we get that $1 \le \sum_{i=1}^k d_i c_i < \frac{1}{\alpha} \sum_{i=1}^k \frac{d_i}{d([i])}$. Reordering terms, this implies $\sum_{i=1}^k \frac{d_i}{d([i])} > \alpha$.
	
	Define $f(x)$ as $1/d_1$ on $[0, d_1)$; $1/(d_1+d_2)$ on $[d_1, d_1+d_2)$; ...; $1/d([i])$ on $[d([i-1]), d([i]))$ for $i \in [k]$. Now we have
	\begin{align*}
	\int_0^{d([k])} f(x) = \frac{d_1}{d_1} + \frac{d_2}{d_1 + d_2} + \frac{d_3}{d_1 + d_2 + d_3} + \ldots + \frac{d_k}{d([k])} > \alpha.
	\end{align*}
	However, since $f(x) \le 1/x$
	\begin{align*}
	\int_0^{d([k])} f(x) & = \int_0^{d_1} f(x)\, dx + \int_{d_1}^{d([k])} f(x)\, dx \\
	& \le 1 + \int_{d_1}^{d([k])} \frac{1}{x}\, dx = 1 + \ln \frac{d([k])}{d_1} .
	\end{align*}
	Hence we have $\alpha < \int_0^{d([k])} f(x) \le 1 + \ln \frac{d([k])}{d_1}$, but we set $\alpha = 1 + \ln \left( \frac{\sum_{i\in [k]} d_i}{\min_{i\in [k]} d_i} \right ) \ge 1 + \ln \frac{d([k])}{d_1}$, where we reach a contradiction and finish the proof.
\end{proof}

%\begin{lemma}
%Let $z, \{f_i(p) \mid i\in [k], p\in {\cal P}_i(T)\}$ be an optimal solution for \primal$(D,h)$. Then the optimal $h$-hop congestion for any fractional $h$-hop routing of $D$ is exactly $\frac{1}{z}$. 
%\end{lemma}

\begin{lemma}\label{lem:dual-to-moving-cut}
Suppose $\{\ell_e \mid e \in E\}, \{c_i \mid i\in [k]\}$ is a feasible solution for \dual$(D,h)$. Then there exists an $h$-length moving cut $C$ for $D$ with $h$-length sparsity $\spa_h(C,D) = O(L\cdot \log d_{\ratio})$ where $L = \sum_{e \in E} u(e) \cdot \ell_e$ and $d_{\ratio} = \sum_{i\in[k]} d_{i}/\min_{i\in[k]} d_{i}$.
\end{lemma}
\begin{proof}
Applying \Cref{lem:bucketing-lemma} to the $c_i$ values results in a subset $I \subseteq [k]$  of demands with $c_{\min} = \min_{i \in I} c_i \ge \frac{1}{\alpha \cdot \sum_{i\in I}d_i}$ for $\alpha = 1 + \ln \left(\frac{\sum_{i\in [k]} d_i}{\min_{i\in [k]} d_i} \right ) = O(\log d_{\ratio})$. Given the scaling factor $c_{\min}$, 
for each edge $e$, we round down the value of $\min\{1, \frac{2\ell_e}{c_{\min}}\}$ to multiple of $1/h$, and define it to be the cut value over $e$ from the $h$-length moving cut $C$.
%we define the $h$-length moving cut $C$
%through $C(e) = \min\{1, \frac{\ell_e}{c_{\min}}\}$. 

We show that $C$ successfully $h$-separates every demand $i \in I$. For this, we note that the first set of dual constraints guarantees that for every demand $i \in I$ and for every path $p \in \calP_i$, the length of $p$, assuming every edge $e$ has length $\ell_e$, is at least $c_{\min}$, i.e., $\forall i \in I, p\in \calP_i: \sum_{e\in p} \ell_e \geq c_{\min}$. 
Recall that edge lengths are positive integers as mentioned in \Cref{sec:preliminray}, so each edge has length at least $1$ in $G$. 
The length of a path $p \in \calP_i$ under cut $C$ is therefore at least 
\begin{align*}
\sum_{e\in p}\ell_{G-C}(e) &= \sum_{e\in P}\ell_{G}(e) + h\cdot C(e)\\
&\geq \sum_{e\in P} 1 + h\cdot (\min\{1,2\ell_{e}/c_{\min}\} - 1/h)\\
&\geq h\cdot\min\{1,\frac{2\sum_{e\in P}\ell_{e}}{c_{\min}}\}\\
&= 2h>h
\end{align*}
This guarantees that $C$ indeed $h$-separates every demand $i \in I$.

The amount of demands in $D$ separated by $C$ therefore satisfies
\[
\sep_h(C,D) \geq \sum_{i\in I} d_i \geq \frac{1}{\alpha \cdot c_{\min}},
\]
where the last inequality comes from the construction if $I$ guarantees that $c_{\min} \ge \frac{1}{\alpha \cdot \sum_{i\in I}d_i}$. By the definition of $C$, the size $|C|$ of the $h$-length cut $C$, is \[
\sum_e u(e) \cdot C(e) \leq \sum_e \frac{u(e) \cdot 2\ell_e}{c_{\min}} = \frac{\sum_e u(e)\cdot 2\ell_e}{c_{\min}} = \frac{2L}{c_{\min}}.
\]
Finally, the $h$-length sparsity $\spa_h(C,D)$ of $C$ for $D$ is $\frac{|C|}{\sep_h(C,D)} \leq \frac{2L/c_{\min}}{1/(\alpha \cdot c_{\min})} = 2L\cdot \alpha = O(L\cdot \log d_{\ratio})$, as desired.
\end{proof}

The proof of \Cref{thm:flowcutgap} follows now immediately from \Cref{lem:flow-LP-equals-congestion}, \Cref{lem:dual-to-moving-cut} and strong LP duality.

\begin{proof}[Proof of \Cref{thm:flowcutgap}]
First of all, $D$ must be $h$-length, otherwise the zero cut $C$ has $\spars_{h}(C,D) = 0$. Let $d_{\min} = \min\{D(u,v)\mid \forall u,v\text{ s.t. }D(u,v)>0\}$ and $d_{\ratio} = |D|/d_{\min}$. We can assume $d_{\ratio} \leq n^{4}N$ by the following reasons.
\begin{itemize}
\item First we have $|D|\leq n^{2}N/\phi$. Assume the opposite, the $h$-length moving cut $C$ with $C(e) = 1$ for all $e\in E(G)$ trivially has $\sep_{h}(C,D) = |D|$ and $|C| = \sum_{e\in E(G)}u_{G}(e)\leq n^{2}N$ (here we use that edge capacities are at most $N$). This means $\spars_{h}(C,D) = |C|/|D| < \frac{n^{2}N}{n^{2}N/\phi} = \phi $, a contradiction.
\item When $d_{\min}\leq 1/(\phi n^{2})$, we consider another demand $D'$ which drops all demand pairs $(u,v)$ with $D(u,v)\leq 1/(\phi n^{2})$, which means $d'_{\min} \geq 1/(\phi n^{2})$. Also, $D'$ still satisfies $\spars_{h}(C,D)\geq \phi$ for all $h$-length moving cuts $C$. Furthermore, if $D'$ can be routed with length $h$ and congestion $O(\log N/\phi)$, then $D$ can also be routed with length $h$ and congestion $O(\log N/\phi) + n^{2}/(\phi n^{2}) = O(\log N/\phi)$ (since each edge has capacity at least $1$).
\end{itemize}

Providing $d_{\ratio}\leq n^{4}N$, we prove the original statement by proving the counter-positive. That is, if a demand $D$ with $d_{\ratio}\leq n^{4}N$ cannot be routed along $h$-length paths with congestion at most $\gamma$ then there must exist a cut $C$ with $h$-length sparsity $\spa_h(C,D) = O(\frac{\log N}{\gamma})$ for $D$.

To see this note that, due to \Cref{lem:flow-LP-equals-congestion}, the value of \primal$(D,h)$ and because of strong duality also the value of the dual \dual$(D,h)$ is at most $\frac{1}{\gamma}$. \Cref{lem:dual-to-moving-cut} now directly implies that there exist the desired $h$-length cut with sparsity $O(\frac{\log d_{\ratio}}{\gamma}) = O(\frac{\log N}{\gamma})$ for $D$.
\end{proof}

% \clearpage

\subsection{Linkedness}\label{sec:appendix_linkedness}

We can actually show the existence of a strengthened version of length-constrained expander decompositions called ``linked'' length-constrained expander decompositions. For each applied moving cut $C$, it will slightly increase the original node-weighting to allow more demand in the graph, which makes it more difficult for the graph to be a length-constrained expander. We define such an increase as a new node-weighting added to the original weighting:
\begin{definition}[Linked Node-Weighting]
    Let $C$ be an $h$-length moving cut of a graph $G = (V,E)$ and let $\l$ be a positive integer divisible by $h$. The linked node-weighting $L_{C}^\l$ assigns to each vertex $v$ the value of $\l \cdot \deg_{C}(v)$ where $\deg_{C}(v) = \sum_{e\ni v}u_G(e)\cdot C(e)$.
\end{definition}

\begin{definition}[Linked Length-Constrained Directed Expander Decomposition]
    Given a directed graph $G = (V,E)$, a directed $\l$-linked $(h,s)$-length $\phi$-expander decomposition for a node-weighting $A$ with cut slack $\kappa$ and length slack $s$ is an $h \cdot s$-length cut $C$ of size at most $\kappa \cdot \phi|A|$ such that $A + L_C^\l$ is $(h,s)$-length $\phi$-expanding in $G - C$.
\end{definition}

We have defined the sequence of moving cuts in \cref{dfn:seq_moving_cuts}, but now we also need to incorporate the linked node-weighting added whenever we apply the cut to the graph.
% For the existence of the length-constrained expander decomposition for a graph $G$ w.r.t some node-weighting $A$, we can find sparse cuts iteratively from the graph. Namely, if the graph $G$ is not an expander, it is guaranteed to admit a sparse moving cut $C$. We can apply this cut to the graph and get a new graph $G' = G - C$. For the linked expander decompositions, we also need to update the node-weighting accordingly. 
% This can be done iteratively until the updated graph is already expander, or in other words, there does not exist any sparse moving cut.
% This gives a sequence of moving cuts, and we can combine them as a single moving cut to show the existence of expander decompositions.
We similarily describe the sequence of linked moving cuts as follows:

\begin{definition}[Sequence of Linked Moving Cuts]\label{dfn:seq_linked_moving_cuts}
Given a directed graph $G = (V, E)$, and node-weighting $A$, let $(C_1, C_2, \dots,C_n)$ be a sequence of $h \cdot s$ moving cuts, let $G - \sum_{j < i} C_j$ denote the graph that is applied with cuts from $C_1$ to $C_{i-1}$ and let $A + \sum_{j < i}L_{C_j}^{\l}$ denote the node-weighting that is added with linked node-weighting from $L_{C_1}^\l$ to $L_{C_{i-1}}^\l$. 
We define $(C_1, C_2, \dots,C_n)$ as a sequence of $\l$-linked $\phi$-sparse moving cuts
if and only if the $(h,s)$-length sparsity of $C_i$ w.r.t $A + \sum_{j < i}L_{C_j}^{\l}$ in $G - \sum_{j < i} C_j$ is at most $\phi$.
\end{definition}

We again need to bound the overall size of those moving cuts. We employ a similar potential argument, but what is different that the petential will not only decrease from the cut but also increase from the added linked node-weighting. 
We divide the potential update into two phases to argue about the potential change more clearly, and show that the potential is overall decreasing.
\begin{lemma}\label{lem:existential_linked_cut_sequence}
Let $C_1, \ldots, C_k$ be an sequence of $\l$-linked $\phi$-sparse $h\cdot s$-length
cuts for some node-weighting $A$ in the graph $G$ where $h > 1$, $\phi < 1$, $\l \le 2^{- 8\alpha - 4} \cdot \frac{1}{\phi \ln n}$, $\alpha \ge 1$ and
$s > \frac{4\log_2 n}{\alpha}$, then $\sum_{i} |C_i| \leq (2^{8\alpha+3}\phi\ln n)\cdot |A|$.
\end{lemma}
\begin{proof}
    Let $G_1$ denote the initial graph $G$, $G_i = G - \sum_{j < i} C_j$, and let $A_i = A + \sum
    _{j < i} L_{C_j}^\l$. 
    From our assumption of the sequence of moving cuts, for every $i$ there exists a symmetric $h$-length $A_i$-respecting demand $D^{*}_{i}$ in the graph $G_i$ such that the $h\cdot s$ sparsity of $C_i$ w.r.t $D^{*}_{i}$ is at most $\phi$. 
    %W.l.o.g, we can assume $D_i$ is fully separated, and let $\sep^i_{h \cdot s}(C_i, D_i)$ denote the amount of demand separated by $C_i$ in $G_i$, we have that $\sep^i_{h\cdot s}(C_i, D_i) \ge \frac{|C_i|}{\phi}$.

    We again introduce the exponential demand for each graph $G_i$ and the corresponding node-weighting $A_i$. For simplicity, let $w_i$ denote the exponential distance weight $w^{\alpha}_h$ w.r.t $G_i$. 
    We further use $D_i$ to denote the corresponding exponential demand $D^i_{h,A_i}$ w.r.t graph $G_i$ and node-weighting $A_i$.
    It is worth noting that exponential demands are different from each other because the graph $G_i$ and the node-weighting $A_i$ is changed for each $i$.
    Specifically, from \cref{lem:separationexpdemand}, we have
    \begin{align*}
        \sep^i_{h\cdot s/2}(C_i, D_i) \ge 2^{-8\alpha-1} \cdot \frac{|C_i|}{\phi}
    \end{align*}

    Further, we define the same potential function $P_i: V \to \reals$ w.r.t the graph $G_i$ and node-weighting $A_i$. It assigns a value to each vertex $u$ with the amount of $P_i(u) =  A_i(u)\ln(w_{i}(u))$. 
    %We note the fact that $w^i_{h}(u) \cdot N \ge 1$, which guarantees that $P_i(u) \ge 0$ for all $i$ and vertices $u$. 

    We divide the potential change into two phases to simplify analysis. In the first phase, we apply cut $C_i$ to the graph $G_i$ and then the resulting graph is $G_{i+1}$. The distance of some edges increases due to $C_i$, but the node-weighting still remains as the $A_{i}$. Then in the second phase, we added the corresponding linked node-weighting $L_{C_i}^\l$ to $A_i$ and get $A_{i+1}$. The distance between vertices in the graph $G_{i+1}$ remains the same in this phase.
    
    Start with graph $G_i$, each vertex $u$ will have potential $P_i(u)$. After applying cut $C_i$ to the graph $G_i$, we first get the resulting graph $G_{i+1}$ with same node-weighting $A_i$, and let $P'_i(u)$ denote the potential of vertex $u$ at this intermediate phase, we have $P'_i(u) = A_i(u)\ln(w_{i+1}(u))$. Since we only increase the length of some edges in $G_i$, the exponential weight can only decrease between any vertex pairs. Consequently, there is indeed a decrease from $P_i(u)$ to $P'_i(u)$, and we have the same result as in \cref{lem:existential_cut_sequence},
    \begin{align}
        % P_i(u) - P'_{i}(u) & = A_i(u)\cdot (\log(w^i_{h}(u)\cdot N) - \log(w^{i+1}_{h}(u)\cdot N))\\
        %  & = A_i(u) \cdot (- \ln(1 - (1 - \frac{w^{i+1}_{h}(u)\cdot N}{w^{i}_{h}(u)\cdot N})) )\\
        % & \ge A_i(u) \cdot (1 - \frac{w^{i+1}_{h}(u)}{w^{i}_{h}(u)}) \label{ineq: log_ineq}\\
        % & \ge A_i(u) \cdot ( \frac{w^{i}_{h}(u) - w^{i+1}_{h}(u)}{w^{i}_{h}(u)})
        P_i(u) - P'_{i}(u) \ge A_{i}(u) \cdot ( \frac{w_{i}(u) - w_{i+1}(u)}{w_{i}(u)})
    \end{align}
    %We use the fact that $-\ln (1 - x) \ge x$ when $0 \le x < 1$ for inequality~(\ref{ineq: log_ineq}).
    % We note that if we add $C_i$ to $G_i$, the distance of all demand pairs that contribute to $\sep^i_{h\cdot s/2}(C_i,D^i_{h,A_i})$ will be at least $\frac{h\cdot s}{2} > \frac{2h\log_2 N}{\alpha}$. Specifically, suppose node pairs $(u,v)$ with $D^i_{h,A_i}(u,v)$ contributes to the separation, which means that $\dist_{G_{i+1}}(u,v) > \frac{2h\log_2 N}{\alpha}$ and $\rdist_{G_{i+1}}(u,v) = \rdist_{G_{i+1}}(v,u) > \frac{2h\log_2 N}{\alpha}$. , we will have $w^{i+1}_{h}(u,v) = w^{i+1}_{h}(v,u) = 0$. This further means that 
    % \begin{align*}
    %     A_i(u)\cdot \frac{w^{i}_{h}(u,v)}{w^{i}_{h}(u)} + A_i(v)\cdot \frac{w^{i}_{h}(v,u)}{w^{i}_{h}(v)}
    % \end{align*}
    % is contributed to the total potential reduction $\sum_{u \in V} (P_i(u) - P'_i(u))$. And we note that the above amount is exactly $D^i_{h,A_i}(u,v)$. However, if $D^i_{h,A_i}(u,v)$ and $D^i_{h,A_i}(v,u)$ both contribute to the separation, then the amount that contributes to the potential reduction still remains as the above. 
    The same analysis gives that the overall potential reduction is still at least 
    \begin{align*}
        \sum_{u \in V} P_i(u) - P'_i(u) \ge \frac{1}{2} \cdot \sep^i_{h\cdot s/2}(C_i, D^i_{h,A_i}) \ge 2^{-8\alpha-2} \cdot \frac{|C_i|}{\phi}
    \end{align*}
    In the second phase, the length of edges in graph $G_{i+1}$ does not increase, but the node-weighting $A_i$ is added with $L_{C_i}^\l$ to get $A_{i+1}$. Thus the potential over vertex $u$ increases from $P'_i(u)$ to $P_{i+1}(u)$, and we have
    \begin{align*}
        P_{i+1}(u) - P'_i(u) & = (A_{i+1}(u) - A_i(u))\cdot \ln(w^{i+1}_{h}(u))\\
        & = \deg_{C_i}(u) \cdot \l \cdot \ln(w^{i+1}_{h}(u))\\
    \end{align*}
    By summing up, we have
    \begin{align*}
        \sum_{u \in V} P_{i+1}(u) - P'_i(u) & = \sum_{u \in V}\deg_{C_i}(u) \cdot \l \cdot \ln(w_{i+1}(u))\\
        & \le 2|C_i|\cdot \l \cdot \ln{n}\\
        & \le 2^{-8\alpha-3}\cdot\frac{|C_i|}{\phi}
    \end{align*}
    We use the condition that $\l \le 2^{- 8\alpha - 4} \cdot \frac{1}{\phi \ln n}$ in second inequality.
    We show that the overall potential $P_i$ decreases at least $2^{-8\alpha-2} \cdot \frac{|C_i|}{\phi}$ to $P'_i$, and then increases at most $2^{-8\alpha-3} \cdot \frac{|C_i|}{\phi}$ to $P_{i+1}$, which is equivalent to that the overall potential is monotonously decreasing as the following:
    \begin{align*}
        \sum_{u \in V} P_{i}(u) - P_{i+1}(u) \ge 2^{-8\alpha-3} \cdot \frac{|C_i|}{\phi}
    \end{align*}
    where $\sum_{u\in V}P_1(u) \ge \sum_{u\in V}P_2(u) \ge \dots \ge \sum_{u\in V}P_{k+1}(u)$. By this, we can come to a similar conclusion over the summation of size of all cuts.
    \begin{align*}
        % \sum_{i}|C_i| & \le 2^{7\alpha+3}\phi \cdot \sum_{i}\sum_{u\in V} (P_{i}(u) - P_{i+1}(u)) \\
        % &\le 2^{7\alpha+3}\phi \cdot \sum_{u\in V} P_{1}(u)\\
        % & \le 2^{7\alpha+3}\phi \cdot \sum_{u\in V} A(u)\log_2(w^1_{h}(u) \cdot N)\\
        % &\le 2^{7\alpha+4}\phi \cdot |A| \cdot \log_2 N.
        \sum_{i}|C_i| \le 2^{8\alpha+3}\phi \cdot |A| \cdot \ln n
    \end{align*}
\end{proof}

\begin{theorem}\label{thm:linked_edge_ED}
    For any $G = (V,E)$, a node-weighting $A$, $h > 1$, $\alpha \ge 1$, $\l \le 2^{- 8\alpha - 4} \cdot \frac{1}{\phi \ln n}$, $\phi < 1$ and a length slack parameter $s = O(\log n)$, there is a directed $\l$-linked $(h,s)$-length $\phi$-expander decomposition for $A$ with cut slack $\kappa = O(n^{O(\frac{1}{s})}\log n)$.
\end{theorem}
%\future{$\kappa = n^{O(\frac{1}{s})}\log n$}
\begin{proof}
    Similar as the analysis in \cref{thm:directed_edge_ED}, we could find cut iteratively from the graph $G$ until the resulting graph is $\phi$-expander. 
    We take the union of cuts as $C_{\le k}$ and by definition $C_{\le k}$ is a valid expander decomposition for $G$ and $A$.
    Further \cref{lem:existential_linked_cut_sequence} guarantees that $|C_{\le k}| = \sum_{j\le k}|C_j| \le 2^{8\alpha+3}\phi \cdot |A| \cdot \ln n$ as long as $s > \frac{4\ln n}{\alpha}$.
    This gives that $\kappa \le 2^{8\alpha+3}\cdot\ln n = O(n^{O(\frac{1}{s})}\log n)$.
\end{proof}

%% file: main.bbl
\newcommand{\etalchar}[1]{$^{#1}$}
\begin{thebibliography}{HKGW23}

\bibitem[ADK23]{agassy2023expander}
Daniel Agassy, Dani Dorfman, and Haim Kaplan.
\newblock Expander decomposition with fewer inter-cluster edges using a spectral cut player.
\newblock In {\em 50th International Colloquium on Automata, Languages, and Programming (ICALP 2023)}. Schloss-Dagstuhl-Leibniz Zentrum f{\"u}r Informatik, 2023.

\bibitem[BBG{\etalchar{+}}20]{bernstein2020fully}
Aaron Bernstein, Jan van~den Brand, Maximilian~Probst Gutenberg, Danupon Nanongkai, Thatchaphol Saranurak, Aaron Sidford, and He~Sun.
\newblock Fully-dynamic graph sparsifiers against an adaptive adversary.
\newblock {\em arXiv preprint arXiv:2004.08432}, 2020.

\bibitem[BBST24]{bernstein2024maximum}
Aaron Bernstein, Joakim Blikstad, Thatchaphol Saranurak, and Ta-Wei Tu.
\newblock Maximum flow by augmenting paths in $n^{2+o(1)}$ time.
\newblock {\em arXiv preprint arXiv:2406.03648}, 2024.

\bibitem[BGS20]{BernsteinGS20}
Aaron Bernstein, Maximilian~Probst Gutenberg, and Thatchaphol Saranurak.
\newblock Deterministic decremental reachability, scc, and shortest paths via directed expanders and congestion balancing.
\newblock In {\em Annual Symposium on Foundations of Computer Science, {FOCS}}, pages 1123--1134, 2020.

\bibitem[BGS22]{bernstein2022deterministic}
Aaron Bernstein, Maximilian~Probst Gutenberg, and Thatchaphol Saranurak.
\newblock Deterministic decremental sssp and approximate min-cost flow in almost-linear time.
\newblock In {\em 2021 IEEE 62nd Annual Symposium on Foundations of Computer Science (FOCS)}, pages 1000--1008. IEEE, 2022.

\bibitem[BH23]{bodwin2023folklore}
Greg Bodwin and Gary Hoppenworth.
\newblock Folklore sampling is optimal for exact hopsets: Confirming the $\sqrt{n}$ barrier.
\newblock In {\em 2023 IEEE 64th Annual Symposium on Foundations of Computer Science (FOCS)}, pages 701--720. IEEE, 2023.

\bibitem[CE15]{chekuri2015all}
Chandra Chekuri and Alina Ene.
\newblock The all-or-nothing flow problem in directed graphs with symmetric demand pairs.
\newblock {\em Mathematical Programming}, 154:249--272, 2015.

\bibitem[CGL{\etalchar{+}}20]{DBLP:conf/focs/ChuzhoyGLNPS20}
Julia Chuzhoy, Yu~Gao, Jason Li, Danupon Nanongkai, Richard Peng, and Thatchaphol Saranurak.
\newblock A deterministic algorithm for balanced cut with applications to dynamic connectivity, flows, and beyond.
\newblock In {\em IEEE Symposium on Foundations of Computer Science (FOCS)}, pages 1158--1167, 2020.

\bibitem[CK19]{chuzhoy2019new}
Julia Chuzhoy and Sanjeev Khanna.
\newblock A new algorithm for decremental single-source shortest paths with applications to vertex-capacitated flow and cut problems.
\newblock In {\em Proceedings of the 51st Annual ACM SIGACT Symposium on Theory of Computing}, pages 389--400, 2019.

\bibitem[CK24]{chuzhoy2024maximum}
Julia Chuzhoy and Sanjeev Khanna.
\newblock Maximum bipartite matching in $n^{2+o(1)}$ time via a combinatorial algorithm.
\newblock In {\em Proceedings of the 56th Annual ACM Symposium on Theory of Computing}, pages 83--94, 2024.

\bibitem[CKL{\etalchar{+}}22]{ChenKLPGS22}
Li~Chen, Rasmus Kyng, Yang~P. Liu, Richard Peng, Maximilian~Probst Gutenberg, and Sushant Sachdeva.
\newblock Maximum flow and minimum-cost flow in almost-linear time.
\newblock In {\em 63rd {IEEE} Annual Symposium on Foundations of Computer Science, {FOCS}}, pages 612--623. {IEEE}, 2022.

\bibitem[CKS05]{chekuri2005multicommodity}
Chandra Chekuri, Sanjeev Khanna, and F~Bruce Shepherd.
\newblock Multicommodity flow, well-linked terminals, and routing problems.
\newblock In {\em Proceedings of the thirty-seventh annual ACM symposium on Theory of computing}, pages 183--192, 2005.

\bibitem[CMGS25]{chen2025parallel}
Daoyuan Chen, Simon Meierhans, Maximilian~Probst Gutenberg, and Thatchaphol Saranurak.
\newblock Parallel and distributed expander decomposition: Simple, fast, and near-optimal.
\newblock In {\em Proceedings of the 2025 Annual ACM-SIAM Symposium on Discrete Algorithms (SODA)}, pages 1705--1719. SIAM, 2025.

\bibitem[CS19]{chang2019improved}
Yi-Jun Chang and Thatchaphol Saranurak.
\newblock Improved distributed expander decomposition and nearly optimal triangle enumeration.
\newblock In {\em Proceedings of the 2019 ACM Symposium on Principles of Distributed Computing}, pages 66--73, 2019.

\bibitem[CS20]{chang2020deterministic}
Y.~Chang and T.~Saranurak.
\newblock Deterministic distributed expander decomposition and routing with applications in distributed derandomization.
\newblock In {\em IEEE Symposium on Foundations of Computer Science (FOCS)}, pages 377--388, full version arXiv: 2007.14898, 2020.

\bibitem[EHHL25]{el2025fully}
Antoine El-Hayek, Monika Henzinger, and Jason Li.
\newblock Fully dynamic approximate minimum cut in subpolynomial time per operation.
\newblock In {\em Proceedings of the 2025 Annual ACM-SIAM Symposium on Discrete Algorithms (SODA)}, pages 750--784. SIAM, 2025.

\bibitem[GPPG24]{gottesburen2024practical}
Lars Gottesb{\"u}ren, Nikos Parotsidis, and Maximilian Probst~Gutenberg.
\newblock Practical expander decomposition.
\newblock In {\em 32nd Annual European Symposium on Algorithms (ESA 2024)}, volume 308, pages 61--1. Schloss Dagstuhl-Leibniz-Zentrum f{\"u}r Informatik, 2024.

\bibitem[GR98]{goldreich1998sublinear}
Oded Goldreich and Dana Ron.
\newblock A sublinear bipartiteness tester for bounded degree graphs.
\newblock In {\em Proceedings of the thirtieth annual ACM symposium on Theory of computing}, pages 289--298, 1998.

\bibitem[GRST21]{GRST21}
Gramoz Goranci, Harald R{\"a}cke, Thatchaphol Saranurak, and Zihan Tan.
\newblock The expander hierarchy and its applications to dynamic graph algorithms.
\newblock In {\em Proceedings of the 2021 ACM-SIAM Symposium on Discrete Algorithms (SODA)}, pages 2212--2228, 2021.

\bibitem[Hes03]{Hesse03}
William Hesse.
\newblock Directed graphs requiring large numbers of shortcuts.
\newblock In {\em Proceedings of the Fourteenth Annual {ACM-SIAM} Symposium on Discrete Algorithms, January 12-14, 2003, Baltimore, Maryland, {USA}}, pages 665--669. {ACM/SIAM}, 2003.

\bibitem[HHG25]{haeupler2025cut}
Bernhard Haeupler, Jonas Huebotter, and Mohsen Ghaffari.
\newblock A cut-matching game for constant-hop expanders.
\newblock In {\em Proceedings of the 2025 Annual ACM-SIAM Symposium on Discrete Algorithms (SODA)}, pages 1651--1678. SIAM, 2025.

\bibitem[HHL{\etalchar{+}}24]{Haeupler2024emulator}
Bernhard Haeupler, D.~Ellis Hershkowitz, Jason Li, Antti Roeyskoe, and Thatchaphol Saranurak.
\newblock Low-step multi-commodity flow emulators.
\newblock In {\em Proceedings of the 56th Annual ACM Symposium on Theory of Computing}, STOC ’24, page 71–82. ACM, June 2024.

\bibitem[HHT24]{haeupler2024new}
Bernhard Haeupler, D~Ellis Hershkowitz, and Zihan Tan.
\newblock New structures and algorithms for length-constrained expander decompositions.
\newblock {\em arXiv preprint arXiv:2404.13446}, 2024.

\bibitem[HKGW23]{hua2023maintaining}
Yiding Hua, Rasmus Kyng, Maximilian~Probst Gutenberg, and Zihang Wu.
\newblock Maintaining expander decompositions via sparse cuts.
\newblock In {\em Proceedings of the 2023 Annual ACM-SIAM Symposium on Discrete Algorithms (SODA)}, pages 48--69. SIAM, 2023.

\bibitem[HLS24]{haeupler2024dynamic}
Bernhard Haeupler, Yaowei Long, and Thatchaphol Saranurak.
\newblock Dynamic deterministic constant-approximate distance oracles with $n^{\epsilon}$ worst-case update time.
\newblock {\em arXiv preprint arXiv:2402.18541}, 2024.

\bibitem[HP21]{HuangP21}
Shang{-}En Huang and Seth Pettie.
\newblock Lower bounds on sparse spanners, emulators, and diameter-reducing shortcuts.
\newblock {\em {SIAM} J. Discret. Math.}, 35(3):2129--2144, 2021.

\bibitem[HRG22]{haeupler2022hopexpander}
Bernhard Haeupler, Harald R\"{a}cke, and Mohsen Ghaffari.
\newblock Hop-constrained expander decompositions, oblivious routing, and distributed universal optimality.
\newblock In {\em Proceedings of the 54th Annual ACM SIGACT Symposium on Theory of Computing}, STOC, page 1325–1338, 2022.

\bibitem[HWZ20]{DBLP:conf/focs/HaeuplerWZ20}
Bernhard Haeupler, David Wajc, and Goran Zuzic.
\newblock Network coding gaps for completion times of multiple unicasts.
\newblock In {\em IEEE Symposium on Foundations of Computer Science (FOCS)}, pages 494--505, 2020.

\bibitem[JPP25]{jiang2025new}
Yonggang Jiang, Merav Parter, and Asaf Petruschka.
\newblock New oracles and labeling schemes for vertex cut queries.
\newblock {\em arXiv preprint arXiv:2501.13596}, 2025.

\bibitem[JS22]{jin2022fully}
Wenyu Jin and Xiaorui Sun.
\newblock Fully dynamic st edge connectivity in subpolynomial time.
\newblock In {\em 2021 IEEE 62nd Annual Symposium on Foundations of Computer Science (FOCS)}, pages 861--872. IEEE, 2022.

\bibitem[JST24]{jin2024fully}
Wenyu Jin, Xiaorui Sun, and Mikkel Thorup.
\newblock Fully dynamic min-cut of superconstant size in subpolynomial time.
\newblock In {\em Proceedings of the 2024 Annual ACM-SIAM Symposium on Discrete Algorithms (SODA)}, pages 2999--3026. SIAM, 2024.

\bibitem[KLOS14]{kelner2014almost}
Jonathan~A Kelner, Yin~Tat Lee, Lorenzo Orecchia, and Aaron Sidford.
\newblock An almost-linear-time algorithm for approximate max flow in undirected graphs, and its multicommodity generalizations.
\newblock In {\em Proceedings of the twenty-fifth annual ACM-SIAM symposium on Discrete algorithms}, pages 217--226. SIAM, 2014.

\bibitem[KVV04]{kannan2004clusterings}
Ravi Kannan, Santosh Vempala, and Adrian Vetta.
\newblock On clusterings: Good, bad and spectral.
\newblock {\em Journal of the ACM (JACM)}, 51(3):497--515, 2004.

\bibitem[LNPS23]{li2023near}
Jason Li, Danupon Nanongkai, Debmalya Panigrahi, and Thatchaphol Saranurak.
\newblock Near-linear time approximations for cut problems via fair cuts.
\newblock In {\em Proceedings of the 2023 Annual ACM-SIAM Symposium on Discrete Algorithms (SODA)}, pages 240--275. SIAM, 2023.

\bibitem[LPS25]{long2025connectivity}
Yaowei Long, Seth Pettie, and Thatchaphol Saranurak.
\newblock Connectivity labeling schemes for edge and vertex faults via expander hierarchies.
\newblock In {\em Proceedings of the 2025 Annual ACM-SIAM Symposium on Discrete Algorithms (SODA)}, pages 1--47. SIAM, 2025.

\bibitem[LR99]{leighton1999multicommodity}
Tom Leighton and Satish Rao.
\newblock Multicommodity max-flow min-cut theorems and their use in designing approximation algorithms.
\newblock {\em Journal of the ACM (JACM)}, 46(6):787--832, 1999.

\bibitem[LS21]{li2021deterministic}
Jason Li and Thatchaphol Saranurak.
\newblock Deterministic weighted expander decomposition in almost-linear time.
\newblock {\em arXiv preprint arXiv:2106.01567}, 2021.

\bibitem[LS22]{DBLP:conf/focs/LongS22}
Yaowei Long and Thatchaphol Saranurak.
\newblock Near-optimal deterministic vertex-failure connectivity oracles.
\newblock In {\em IEEE Symposium on Foundations of Computer Science (FOCS)}, pages 1002--1010, 2022.

\bibitem[NSWN17]{spanningforest}
Danupon Nanongkai, Thatchaphol Saranurak, and Christian Wulff-Nilsen.
\newblock Dynamic minimum spanning forest with subpolynomial worst-case update time.
\newblock In {\em Proceedings of the 58th Annual IEEE Symposium on Foundations of Computer Science}, pages 950--961, 10 2017.

\bibitem[NSY23]{nalam2023deterministic}
Chaitanya Nalam, Thatchaphol Saranurak, and Sorrachai Yingchareonthawornchai.
\newblock Deterministic $ k $-vertex connectivity in $ k^2$ max-flows.
\newblock {\em arXiv preprint arXiv:2308.04695}, 2023.

\bibitem[Pel00]{peleg2000distributed}
David Peleg.
\newblock {\em Distributed computing: a locality-sensitive approach}.
\newblock SIAM, 2000.

\bibitem[RST14]{RST14}
Harald R{\"{a}}cke, Chintan Shah, and Hanjo T{\"{a}}ubig.
\newblock Computing cut-based hierarchical decompositions in almost linear time.
\newblock In {\em Proceedings of ACM-SIAM Symposium on Discrete Algorithms (SODA)}, pages 227--238, 2014.

\bibitem[SG24]{sulser2024simplenearoptimalalgorithmdirected}
Aurelio~L. Sulser and Maximilian~Probst Gutenberg.
\newblock A simple and near-optimal algorithm for directed expander decompositions, 2024.

\bibitem[She13]{sherman2013nearly}
Jonah Sherman.
\newblock Nearly maximum flows in nearly linear time.
\newblock In {\em 2013 IEEE 54th Annual Symposium on Foundations of Computer Science}, pages 263--269. IEEE, 2013.

\bibitem[ST04]{spielman2004nearly}
Daniel~A Spielman and Shang-Hua Teng.
\newblock Nearly-linear time algorithms for graph partitioning, graph sparsification, and solving linear systems.
\newblock In {\em Proceedings of the thirty-sixth annual ACM symposium on Theory of computing}, pages 81--90, 2004.

\bibitem[SW19]{DBLP:conf/soda/SaranurakW19}
Thatchaphol Saranurak and Di~Wang.
\newblock Expander decomposition and pruning: Faster, stronger, and simpler.
\newblock In {\em Proceedings of ACM-SIAM Symposium on Discrete Algorithms (SODA)}, pages 2616--2635, 2019.

\bibitem[SY22]{saranurak2022deterministic}
Thatchaphol Saranurak and Sorrachai Yingchareonthawornchai.
\newblock Deterministic small vertex connectivity in almost linear time.
\newblock In {\em 2022 IEEE 63rd Annual Symposium on Foundations of Computer Science (FOCS)}, pages 789--800. IEEE, 2022.

\end{thebibliography}
